\documentclass[amsmath,amssymb,aps,prb,onecolumn]{revtex4-1}

\usepackage{bibentry}
\nobibliography*

\usepackage{amsmath}
\usepackage{amssymb}
\usepackage{amsthm}

\usepackage{tikz-cd}

\usepackage{mathtools}

\usepackage{empheq}

\usepackage{tensor}
\usepackage{xcolor}
\usepackage{hyperref}

\usepackage{graphicx}
\usepackage{tikz}
\usepackage{pgfplots}

\usepackage{nicefrac}

\usepackage{enumitem}

\newboolean{plotTikzPics}
\setboolean{plotTikzPics}{true}

\newboolean{useSnapshots}
\setboolean{useSnapshots}{false}

\newboolean{forArxiv}
\setboolean{forArxiv}{true}

\newcommand{\wrt}{w.\,r.\,t.}
\newcommand{\eg}{e.\,g.}
\newcommand{\ie}{i.\,e.}
\newcommand{\cf}{cf.}

\newcommand{\formComma}{\,\text{,}}
\newcommand{\formPeriod}{\,\text{.}}

\newcommand{\surf}{\mathcal{S}}
\newcommand{\surfDiscrete}{\surf_h}
\newcommand{\triangulation}{\mathcal{T}}
\newcommand{\feSpace}{\mathcal{V}(\surfDiscrete)}
\newcommand{\weak}[2]{\left( #1 \ , \ #2 \right)}
\newcommand{\dS}{\textup{d}\surf}
\newcommand{\R}{\mathbb{R}}

\newcommand{\tanPenP}{\omega_t^{\dirf}}
\newcommand{\tanPenQ}{\omega_t^{\qten}}

\newcommand{\pnorm}{\omega_{\text{n}}}

\newcommand{\ddt}{\frac{\textup{d}}{\textup{d}t}}
\newcommand{\dup}{\textup{d}}

\newcommand{\para}{\boldsymbol{X}}
\newcommand{\dirfC}{p}
\newcommand{\dirf}{\boldsymbol{\dirfC}}
\newcommand{\DirfC}{P}
\newcommand{\Dirf}{\boldsymbol{\DirfC}}
\newcommand{\qtenC}{q}
\newcommand{\qten}{\boldsymbol{\qtenC}}
\newcommand{\QtenC}{Q}
\newcommand{\Qten}{\boldsymbol{\QtenC}}
\newcommand{\rtenC}{r}
\newcommand{\rten}{\boldsymbol{\rtenC}}

\newcommand{\vnor}{v}
\newcommand{\normal}{\boldsymbol{\nu}}
\newcommand{\orderp}{S}

\newcommand{\xvar}{\xi}
\newcommand{\dxvar}{\delta\xvar}

\newcommand{\shop}{\boldsymbol{B}}
\newcommand{\shopC}{B}
\newcommand{\meanc}{\mathcal{H}}
\newcommand{\gaussc}{\mathcal{K}}
\newcommand{\g}{\boldsymbol{g}}
\newcommand{\LC}{\boldsymbol{E}}
\newcommand{\LCC}{E}

\newcommand{\proj}{\Pi}
\newcommand{\Id}{\operatorname{Id}}

\newcommand{\Tr}{\operatorname{Tr}}

\newcommand{\Lie}{\mathcal{L}}

\newcommand{\dsurf}{\delta_{\surf}}

\newcommand{\tangent}{\operatorname{T}\!}
\newcommand{\tangentT}[1]{\tensor{\tangent}{#1}}
\newcommand{\qspace}{\mathcal{Q}\surf}

\newcommand{\PotE}{\mathcal{U}}
\newcommand{\hatPotE}{{\hat{\PotE}}}
\newcommand{\hatPotEP}{\hatPotE^{\dirf}}
\newcommand{\hatPotEQ}{\hatPotE^{\Qten}}
\newcommand{\PotEP}{\PotE^{\dirf}}
\newcommand{\PotEQ}{\PotE^{\Qten}}

\newcommand{\helfrichEnergy}{\PotE_{\text{H}}}
\newcommand{\frankOseenEnergy}{\PotE_{\text{FO}}}
\newcommand{\landaudegennesEnergy}{\PotE_{\text{LdG}}}
\newcommand{\normalizationEnergy}{\PotE_{\text{n}}}
\newcommand{\bulkEnergy}{\PotE_{\text{B}}}
\newcommand{\elasticEnergy}{\PotE_{\text{El}}}
\newcommand{\areaEnergy}{\PotE_{\text{a}}}

\newcommand{\helfrich}{Helfrich}
\newcommand{\frankOseen}{Frank-Oseen}
\newcommand{\landauDeGennes}{Landau-de Gennes}
\newcommand{\frankOseenHelfrich}{\frankOseen-\helfrich}
\newcommand{\landauDeGennesHelfrich}{\landauDeGennes-\helfrich}

\newcommand{\kinConst}{k}
\newcommand{\kinConstP}{\kinConst^{\dirf}}
\newcommand{\kinConstQ}{\kinConst^{\qten}}

\newcommand{\areaPen}{\omega_{\text{a}}}
\newcommand{\area}{A}
\newcommand{\areaZero}{\area_0}

\newcommand{\inputTikzPic}[1]{\ifthenelse{\boolean{plotTikzPics}}{\input{#1}}{\fbox{\centering\begin{minipage}[t][5cm][t]{0.9\textwidth}\centering\textbf{\color{red}enable plotTikzPics}\end{minipage}}}}

\newcommand{\Grad}{\nabla}
\newcommand{\Div}{\operatorname{div}}%
\newcommand{\BigRot}{\operatorname{Rot}}%
\newcommand{\DivSurf}{\Div_{\!\surf}}%
\newcommand{\GradSurf}{\Grad_{\!\surf}}

\newcommand{\BigRotSurf}{\BigRot_{\surf}}

\newcommand{\vecNabla}{\boldsymbol{\nabla}}
\newcommand{\levicivitaPlain}{\vecNabla}
\newcommand{\levicivita}[2]{\vecNabla_{\!\!#1}{#2}}
\newcommand{\laplaceBeltrami}{\Delta_{\surf}}
\newcommand{\vecLaplace}{\boldsymbol{\Delta}}

\newcommand{\laplaceBochner}{\vecLaplace^{\!\textup{B}}}

\newcommand{\insertColorbarVertical}[5]{
	\begin{minipage}{1.5cm}
		\begin{flushleft}
			\begin{tikzpicture}
				\node (colorbar) at (0,0) {\includegraphics[width=0.6cm]{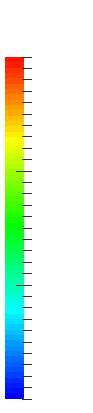}};
				\draw (0.01,1.4) node {\scriptsize #1};
				\draw (-0.15,0.965) node[anchor=west] {\scriptsize #5};
				\draw (-0.15,0.22666667) node[anchor=west] {\scriptsize #4};
				\draw (-0.15,-0.5116667) node[anchor=west] {\scriptsize #3};
				\draw (-0.15,-1.25) node[anchor=west] {\scriptsize #2};
			\end{tikzpicture}
		\end{flushleft}
	\end{minipage}
}

\newtheorem{theorem}{\bf Theorem}[section]

\newtheorem{proposition}{\bf Proposition}[section]
\ifthenelse{\boolean{forArxiv}}
{
	
}
{
}

\ifthenelse{\boolean{forArxiv}}
{
}
{
	\jname{rspa}
	\Journal{Proc R Soc A\ }
}

\begin{document}

\title{Liquid Crystals on Deformable Surfaces}

\author{
Ingo Nitschke$^{1}$, Sebastian Reuther$^{1}$ and Axel Voigt$^{1,2,3,4}$}

\ifthenelse{\boolean{forArxiv}}
{
	\affiliation{$^{1}$Institute of Scientific Computing, TU Dresden, 01062 Dresden, Germany\\
	$^{2}$Dresden Center for Computational Materials Science (DCMS), TU Dresden, 01062 Dresden, Germany\\
	$^{3}$Center for Systems Biology Dresden (CSBD), Pfotenhauerstra{\ss}e 108, 01307 Dresden, Germany\\
	$^{4}$Cluster of Excellence Physics of Life (PoL), 01062 Dresden, Germany}
}
{
	\address{$^{1}$Institute of Scientific Computing, TU Dresden, 01062 Dresden, Germany\\
	$^{2}$Dresden Center for Computational Materials Science (DCMS), TU Dresden, 01062 Dresden, Germany\\
	$^{3}$Center for Systems Biology Dresden (CSBD), Pfotenhauerstra{\ss}e 108, 01307 Dresden, Germany\\
	$^{4}$Cluster of Excellence Physics of Life (PoL), 01062 Dresden, Germany}
}

\ifthenelse{\boolean{forArxiv}}
{
}
{
\subject{biophysics, materials science, mathematical modelling}
}


\ifthenelse{\boolean{forArxiv}}
{
}
{
\corres{Sebastian Reuther\\
\email{sebastian.reuther@tu-dresden.de}}
}

\begin{abstract}
Liquid crystals with molecules constrained to the tangent bundle of a curved surface 
show interesting phenomena resulting from the tight coupling of the elastic and bulk 
free energies of the liquid crystal with geometric properties of the surface. We derive thermodynamically consistent \frankOseenHelfrich\ 
and \landauDeGennesHelfrich\ models which consider the simultaneous relaxation of the director/Q-tensor fields 
and the surface. The resulting systems of vector- or tensor-valued surface partial differential equation 
and geometric evolution laws are numerically solved to tackle the rich dynamics of these systems and to compute the
resulting equilibrium shapes. The results strongly depend on the intrinsic and extrinsic curvature 
contributions and can lead to unexpected asymmetric shapes.
\end{abstract}


\ifthenelse{\boolean{forArxiv}}
{
	\maketitle
	
	\section{Introduction}

	In-plane order on two-dimensional manifolds has been the subject of much research elucidating the intimate relation between topological defects and the geometry of the manifold. Depending on the topology of the manifold these defects can not only be energetically favourable, but topologically necessary. They can play fundamental roles leading to striking results and structures in the ground state that would be highly suppressed in flat systems. Such properties are summarized in \cite{Bowicketal_AP_2009,Serra_LC_2016} for positional and orientational order, respectively. Colloidal crystals assembled on spherical or toroidal surfaces provide good examples to study the role of curvature on positional ordering. Here the number of defects can deviate from the minimal topologically required defect set, leading to scars and pleats in the ground state, see \eg\ \cite{Bauchetal_Science_2003, Irvineetal_Nature_2010}. Orientational ordering has been analysed in the context of liquid crystals. Due to the close relation between principle curvature and director field the coupling between geometric properties and defects becomes even tighter, leading to \eg\ non-required defects on a torus \cite{Jeseneketal_SM_2015}. On a sphere the ground state has been identified as a tetrahedral configuration \cite{Lubenskyetal_JPII_1992} and coated spherical colloids with functionalized defects have been proposed as building blocks for colloidal crystals with a tetrahedral structure \cite{Nelson_NL_2002}.
}
{
	\begin{fmtext}
	\section{Introduction}

	In-plane order on two-dimensional manifolds has been the subject of much research elucidating the intimate relation between topological defects and the geometry of the manifold. Depending on the topology of the manifold these defects can not only be energetically favourable, but topologically necessary. They can play fundamental roles leading to striking results and structures in the ground state that would be highly suppressed in flat systems. Such properties are summarized in \cite{Bowicketal_AP_2009,Serra_LC_2016} for positional and orientational order, respectively. Colloidal crystals assembled on spherical or toroidal surfaces provide good examples to study the role of curvature on positional ordering. Here the number of defects can deviate from {\color{white} blablablabla}
	\end{fmtext}

	\maketitle

	\noindent the minimal topologically required defect set, leading to scars and pleats in the ground state, see \eg\ \cite{Bauchetal_Science_2003, Irvineetal_Nature_2010}. Orientational ordering has been analysed in the context of liquid crystals. Due to the close relation between principle curvature and director field the coupling between geometric properties and defects becomes even tighter, leading to \eg\ non-required defects on a torus \cite{Jeseneketal_SM_2015}. On a sphere the ground state has been identified as a tetrahedral configuration \cite{Lubenskyetal_JPII_1992} and coated spherical colloids with functionalized defects have been proposed as building blocks for colloidal crystals with a tetrahedral structure \cite{Nelson_NL_2002}.  
}

In contrast to these studies on fixed surfaces, we allow the manifold to change. 
This additional relaxation mechanism provides a path to locally overcome geometric frustration.
The tendency to accommodate some preferential in-plane order is the mechanism behind the buckling of crystalline sheets \cite{Seungetal_PRA_1988} and also the origin of the icosahedral shape of various viral capsides \cite{Lidmaretal_PRE_2003}. 
While the interplay of crystalline order and shape relaxation has been studied theoretically and computationally, see \eg\ \cite{Seungetal_PRA_1988,Lidmaretal_PRE_2003,Bruinsmaetal_PRL_2003,Kohyamaetal_PRL_2007,Alandetal_MMS_2012}, the situation for polar and nematic order is less explored. 
Defects in polar and nematic shells are intensively studied on fixed geometries, such as a sphere \cite{Nelson_PRB_1983,Lubenskyetal_JPII_1992,Dzubiellaetal_PRE_2000,Batesetal_SM_2010,Shinetal_PRL_2008,Lopez-Leonetal_NP_2011,Dhakaletal_PRE_2012} and under more complicated constraints, see \eg\ \cite{Prinsenetal_PRE_2003,Selingeretal_JPCB_2011,Nguyenetal_SM_2013,Martinezetal_NM_2014,Segattietal_PRE_2014,Alaimoetal_SR_2017}. 
Most of these studies use particle methods. 
However, also field theoretical descriptions exist, see \eg\ \cite{Kraljetal_SM_2011,Napolietal_PRE_2012,Napolietal_PRL_2012,Golovatyetal_JNS_2017,Nestleretal_JNS_2018,Nitschkeetal_PRSA_2018}. 
These models differ in details and strongly depend on the assumptions made in the derivation. While most of them only focus on the steady state, some also account for dynamics via gradient flows, see \eg\ \cite{Nestleretal_JNS_2018,Nitschkeetal_PRSA_2018}. 
These approaches are well suited to be extended towards changing manifolds. 
An attempt in this direction can be found in \cite{Nitschkeetal_PRF_2019} for polar order, where the model in \cite{Nestleretal_JNS_2018} is generalized to manifolds with a prescribed evolution. 
We are interested in the resulting equilibrium shapes if both the manifold and the orientational order (polar and nematic) on it are allowed to relax. 
Analytical results within a simplified phenomenological mean-field theory \cite{Parketal_EPL_1992} suggest shape changes from spherical to ellipsoidal and tetrahedral, for polar and nematic order, respectively. 
We will show that the situation can become more complex if not only intrinsic but also extrinsic curvature contributions are considered. 
At least for weak bending forces extrinsic curvature contributions can break the symmetry, making the tetrahedral defect configuration unstable. The change from a tetrahedral to a planar defect arrangement has strong implications for the aforementioned applications. For polar order the changes are less dramatic. However, also in this case extrinsic curvature contributions have an influence and e.g. lead to stronger distortions of the shape in the vicinity of the defects.  
Such shape changes are also relevant in the understanding of morphological changes during development and the design of bio-inspired materials that are capable of self-organization. 
Defect dynamics and corresponding morphological transitions have been observed by restricting suspensions of microtubule filaments driven by kinesin motors to the membrane of vesicles \cite{Keberetal_Science_2014}. 
The observed shape changes are tunable and lead to ring-shaped, spindel-shaped and motil droplets with filipodia-like protrusions. 
To fully understand these complex out-of-equilibrium systems, first an understanding of systems without active components is necessary. 
We will therefore only focus on passive systems and leave the investigation of additional active terms for future research.  

The outline of the paper is as follows: 
We start in Section \ref{sec:model} with the derivation of thermodynamically consistent models resulting from \frankOseenHelfrich\ and \landauDeGennesHelfrich\ energies and discuss their relations. 
Section \ref{sec:numerics} described the numerical approach to solve the resulting geometric and vector- and tensor-valued surface partial differential equations. 
The equations are solved in Section \ref{sec:results} for specialized situations with increasing complexity, including the evolution of the director/Q-tensor field on surfaces with a prescribed normal velocity, the surface response to a stationary director/Q-tensor field and the fully coupled system. 
We further discuss the results and explain the observed phenomena as a result of the tight coupling of intrinsic and extrinsic curvature with the director/Q-tensor field and the corresponding topological defects.
All modeling and numerical details are provided in the Appendix.

\section{Notation and Model Derivation} \label{sec:model}
To obtain the desired equilibrium shapes we derive steepest decent models for the \frankOseenHelfrich\ energy $ \hatPotEP = \hatPotEP(\surf,\dirf) $ and the \landauDeGennesHelfrich\ energy $ \hatPotEQ = \hatPotEQ(\surf,\qten) $, which are specified below.
Thereby, $ \dirf\in\tangentT{^1}\surf $ is a tangential surface director field, with $ \| \dirf \| = 1 $ and $ \qten\in\qspace$ is a tangential surface Q tensor, with $\qspace = \{\mathbf{r} \subset\tangentT{^2}\surf \, : \, \Tr \mathbf{r} = 0, \mathbf{r} = \mathbf{r}^T \}$.
Furthermore, $ \surf = \surf(t) $ is a time dependent surface without boundary, which evolves only in normal direction, and $\tangentT{^n}\surf$ denote the tangential tensor bundle of degree $n\geq0$.
 
We will consider variations of the energies $\hatPotEP$ and $\hatPotEQ$. This requires to identify dependencies and to provide ways how to handle them. 
The restriction to normal deformations allows to describe the surface by an independent scalar-valued variable $ \xi=\xi(t)\in\tangentT{^0}\surf $.
We deploy the ansatz $ \para=\para(\xvar)=\para(\xvar(t,y^1,y^2),y^1,y^2) $ for a time depending parametrization to locally describe the surface $\surf$. 
We thus obtain
\begin{align}
    \label{eq:paraDot}
    \dot{\para} &:= \ddt\para = \frac{\partial\xvar}{\partial t}\frac{\partial \para}{\partial \xvar} = \vnor\normal
\end{align} 
with normal velocity $\vnor = \frac{\partial\xvar}{\partial t}$ and normal vector $\normal = \frac{\partial \para}{\partial \xvar}$.
Moreover, a perturbation of the surface gives a scalar-valued perturbation in normal direction by $ \dxvar $, \ie
\begin{align}
    \label{eq:deltaPara}
    \delta\para &= \frac{\partial \para}{\partial \xvar}\dxvar =  \dxvar\normal \formPeriod
\end{align} 
Considering the perturbed surface $ \tilde{\surf} $ locally defined by $ \tilde{\para} = \para + \varepsilon\delta\para $, we introduce the surface derivative $ \dsurf := \frac{d}{d\varepsilon}\big|_{\varepsilon=0} $ for all quantities of $ \surf $, which are extendable to $ \tilde{\surf} $ and are sufficiently smooth.
Note that this derivative is consistent \wrt\ the perturbation $ \tilde{\xvar} = \xvar + \varepsilon\dxvar $ in the sense that for functions $ (\hat{f}\circ\para)(\xvar) = f(\xvar) $  holds $ \dsurf\hat{f} = \dsurf f $, since $ \levicivita{\normal}{\hat{f}} = \frac{d}{d\xvar}\hat{f} = \frac{\partial f}{\partial \xvar}  $.
Especially, for functionals $ \hat{F}= \hat{F}[\para] $ and $ F=F[\xvar] = (\hat{F}\circ\para)[\xvar] $ the derivative $ \dsurf $ implies the functional derivative $\frac{\delta F}{\delta \xvar}$ \wrt\ $ \xvar $, \ie
\begin{align}
    \dsurf\hat{F} &=  \dsurf F
        = \lim_{\varepsilon\rightarrow 0} \frac{F[\xvar] - F[\xvar + \varepsilon\dxvar]}{\varepsilon} 
        = \int_{\surf} \frac{\delta F}{\delta \xvar} \dxvar \mu 
        = \left(\frac{\delta F}{\delta \xvar}\right)^{\!*}\![\dxvar]\formPeriod
\end{align}

Obviously, the surface derivative $ \dsurf $ is not a tensor operator on $ \tangentT{^n}\surf $ for $n \geq 1$, but on $ \tangentT{^n}\R^3|_{\surf} $.
Therefore, claiming $ \dsurf\dirf = 0 $ to set $ \dirf $ as an independent variable \wrt\ the surface, would be over-determined in $ \tangent\R^3|_{\surf} $ and we rather require $\proj_{\surf}\dsurf\dirf=0$ instead, where $\proj_{\surf}$ denotes the projection into the tangent space.
This results in $ \dsurf\dirf = \left( \levicivita{\dirf}{\dxvar} \right)\normal $ or in contravariant components $\dsurf\dirfC^{i} = \left[ \shop\dirf \right]^{i}\dxvar$ with the shape operator $\shop$, see eq. \eqref{eq_codsurfdirf}.
Similar arguments hold for the Q-tensor $\qten$. Claiming $ \dsurf\qten = 0 $ to set $ \qten $ as an independent variable \wrt\ the surface, would again be over-determined in $ \tangentT{^2}\R^3|_{\surf} $. 
Thus, we rather require $ \proj_{\surf}\dsurf\qten = \proj_{\qspace}\dsurf\qten=0 $ instead and finally get $ \dsurf\qtenC^{ij} = \left[ \shop\qten + \qten\shop \right]^{ij}\dxvar $ for the contravariant components, see eq. \eqref{eq_codsurfqten}. 
Thereby, $\proj_{\qspace}$ denotes the projection into the space of tangential Q-tensors, see below. 
For further details we refer to \autoref{sec:app:one}.
    
\subsection{Free Energies}

The \frankOseenHelfrich\ energy $ \PotEP = \PotEP[\xvar,\dirf] = \hatPotEP[\para,\dirf] $ and the \landauDeGennesHelfrich\ energy $ \PotEQ = \PotEQ[\xvar,\qten] = \hatPotEQ[\para,\qten] $ are given by 
\begin{align}
    \PotEP[\xvar,\dirf] &= \helfrichEnergy[\xvar] + \frankOseenEnergy[\xvar,\dirf] + \normalizationEnergy[\xvar,\dirf] + \areaEnergy[\xvar] \\
    \PotEQ[\xvar,\qten] &= \helfrichEnergy[\xvar] + \landaudegennesEnergy[\xvar,\dirf] + \areaEnergy[\xvar]
\end{align}
where 
\begin{align} 
    \label{eq:energies1}
    \helfrichEnergy &= \frac{\alpha}{2}\int_{\surf} \meanc^2\mu \\
    \frankOseenEnergy &= \frac{K}{2}\int_{\surf} \left\| \GradSurf\dirf \right\|^2 + \left\| \shop\dirf \right\|^2 \mu \\
    \normalizationEnergy &= \frac{\pnorm}{4} \int_{\surf} \left( \left\| \dirf \right\|^2 - 1 \right)^2 \mu \\
    \landaudegennesEnergy &= \elasticEnergy + \bulkEnergy \\
    \label{eq_elenrgy}
    \elasticEnergy &= \frac{L}{2} \int_{\surf} \left\| \GradSurf\qten \right\|^2 + \left\| \shop \right\|^2\Tr\qten^2 +2\orderp\meanc\left\langle \shop, \qten \right\rangle + \frac{\orderp^2}{2}\left\| \shop \right\|^2\mu \\
    \bulkEnergy &= \int_{\surf} a'\Tr\qten^2 + c\Tr\qten^4 + C_1\mu \\
    \label{eq:energies7}
    \areaEnergy &= \frac{\areaPen}{2}\left(\frac{\area-\areaZero}{\areaZero}\right)^2 
\end{align}
are the \helfrich\ energy, the \frankOseen\ energy, the normalization penalty energy, the \landauDeGennes\ energy, the elastic energy, the bulk energy and the surface area penalization energy, respectively, with the bending rigidity $\alpha$, mean curvature $\meanc$, \frankOseen\ constant $K$, normalization penalty parameter $\pnorm$, \landauDeGennes\ constant $L$, nematic order parameter $\orderp$, thermotropic parameters $a,b$ and $c$, $a'= a + \nicefrac{b}{3}\orderp + \nicefrac{c}{6}\orderp^2$, $C_1 = \orderp^2(\nicefrac{a}{6}-\nicefrac{b}{54}\orderp + \nicefrac{c}{72}\orderp^2)$, surface area penalization parameter $\areaPen$, desired surface area $\areaZero$, actual surface area $\area := \int_{\surf}\mu$ and the covariant surface gradient $\GradSurf$.
The \frankOseenHelfrich\ and \landauDeGennesHelfrich\ energies are the one-constant approximations of the the energies obtained in \cite{Nestleretal_JNS_2018,Nitschkeetal_PRSA_2018} with the parameter choice $K_1 = K_2 = K_3 = K$ and $L_1=L $, $ L_2=L_3=L_6=0$ and eigenvalue $\beta=-\nicefrac{\orderp}{3}$ of the auxiliary proper Q-tensor
\begin{align}\label{eq_qtoQ}
    \Qten = \qten + \frac{\orderp}{6}\proj_{\surf} - \frac{\orderp}{3}\normal\otimes\normal\in\mathcal{Q}\R^3\big|_{\surf}\formPeriod
\end{align} 
The energies in \cite{Nestleretal_JNS_2018,Nitschkeetal_PRSA_2018} result as a thin-film limit of the corresponding energies in $\R^3$ with appropriate boundary conditions. 
In the nematic regime the minimizers of both energies are the same, as shown in \cite{Majumdaretal_ARMA_2010} for $\R^3$. 
On the surface we can state $\qten = \orderp\proj_{\qspace}[\dirf\otimes\dirf]$ with $\dirf\in\tangentT{^1}\surf$ and $\|\dirf\|=1$.
Therefore, the thermotropic stationary bulk energy $\bulkEnergy\big|_\xvar = \bulkEnergy\big|_\xvar[\qten]= \bulkEnergy\big|_\xvar[\orderp]$ only depends on the nematic order parameter and has a minimum at $\orderp = (-b + \sqrt{b^2-24ac})/(4c)$.
Moreover, up to a constant induced from the bulk energy, also the \landauDeGennesHelfrich\ energy is the same as the \frankOseenHelfrich\ energy for $ L=\nicefrac{K}{2\orderp^2} $.
The relation between $\frankOseenEnergy$ and $\elasticEnergy$ in $\R^3$ and on the surface can thus be outlined in the commutative diagram
\begin{center}
	\begin{minipage}{\textwidth}
		\centering
		\fboxrule=1pt
		\fboxsep=1pt
		\def\boxWidth{0.375\textwidth}
		\def\boxHeigth{2cm}
		\def\boxAlignVertical{c}
		\def\boxAlignHorizontal{c}
		\def\boxOffset{-0.2cm}
		\definecolor{boxColor}{cmyk}{0.125, 0.0625, 0.015, 0.075}
		\def\tipLength{1.1}
		\def\tipHeight{0.8}
		\definecolor{arrowColor}{cmyk}{0.5, 0.25, 0.06, 0.3}
		\newcommand{\flowChartArrowTwo}[6]{%
			\rotatebox{#1}{
				\begin{tikzpicture}[every node/.style={transform shape}]
					\pgfmathsetmacro{\tipLength}{#3}
					\pgfmathsetmacro{\tipHeight}{#4}
					\pgfmathsetmacro{\textOffset}{#5}
					\fill[arrowColor] (0,0) -- (0,\tipHeight) -- (\tipLength,0.5*\tipHeight) -- (0,0);
					\draw (\textOffset*\tipLength,0.5*\tipHeight) node[anchor=west] {\rotatebox{#2}{\color{white}\bfseries{\small #6}}};
				\end{tikzpicture}
			}
		}
		\scriptsize
		\begin{minipage}{\textwidth}
			\centering
			\begin{minipage}{\boxWidth}
				\centering
				\fcolorbox{boxColor}{boxColor}
				{
					\begin{minipage}[\boxAlignHorizontal][\boxHeigth][\boxAlignVertical]{0.9\textwidth}
						\vspace{\boxOffset}
						\begin{align*}
							\frac{K}{2}\int_{\surf_{\textup{TF}}}\left\| \nabla\Dirf \right\|^2&\textup{d}\mathbf{x} \\
							\left\langle \Dirf, \normal \right\rangle|_{\partial\surf_{\textup{TF}}} &= 0 \\
							\proj_{\partial\surf_{\textup{TF}}}\left[ \levicivita{\normal}{\Dirf} \right]|_{\partial\surf_{\textup{TF}}} &= 0
						\end{align*}
					\end{minipage}
				}
			\end{minipage}
			\begin{minipage}{0.1\textwidth}
				\centering
				\flowChartArrowTwo{0}{0}{\tipLength}{\tipHeight}{0.0}{\scriptsize$\lambda\rightarrow0$}
			\end{minipage}
			\begin{minipage}{\boxWidth}
				\centering
				\fcolorbox{boxColor}{boxColor}
				{
					\begin{minipage}[\boxAlignHorizontal][\boxHeigth][\boxAlignVertical]{0.9\textwidth}
						\vspace{\boxOffset}
						\begin{align*}
							\frac{K}{2}\int_{\surf}\left\| \GradSurf\dirf \right\|^2 + \left\| \shop\dirf \right\|^2\mu
						\end{align*}
					\end{minipage}
				}
			\end{minipage}
		\end{minipage}
		\begin{minipage}{\textwidth}
			\centering
			\begin{minipage}{\boxWidth}
				\centering
				\flowChartArrowTwo{90}{-90}{\tipLength}{\tipHeight}{0.0}{}
			\end{minipage}
			\begin{minipage}{0.1\textwidth}
				\centering
				\ 
			\end{minipage}
			\begin{minipage}{\boxWidth}
				\centering
				\flowChartArrowTwo{90}{-90}{\tipLength}{\tipHeight}{0.0}{}
			\end{minipage}
		\end{minipage}
		\begin{minipage}{\textwidth}
			\centering
			\begin{minipage}{\boxWidth}
				\centering
				\fcolorbox{boxColor}{boxColor}
				{
					\begin{minipage}[\boxAlignHorizontal][\boxHeigth][\boxAlignVertical]{0.9\textwidth}
						\vspace{\boxOffset}
						\begin{align*}
							\frac{L}{2}\int_{\surf_{\textup{TF}}}\left\| \nabla\Qten \right\|^2&\textup{d}\mathbf{x}\\
							\Qten\normal|_{\partial\surf_{\textup{TF}}} &= \beta\normal \\
							\proj_{\partial\surf_{\textup{TF}}}\left[ \levicivita{\normal}{\Qten} \right]|_{\partial\surf_{\textup{TF}}} &= 0
						\end{align*}
					\end{minipage}
				}
			\end{minipage}
			\begin{minipage}{0.1\textwidth}
				\centering
				\flowChartArrowTwo{0}{0}{\tipLength}{\tipHeight}{0.0}{\scriptsize$\lambda\rightarrow0$}
			\end{minipage}
			\begin{minipage}{\boxWidth}
				\centering
				\fcolorbox{boxColor}{boxColor}
				{
					\begin{minipage}[\boxAlignHorizontal][\boxHeigth][\boxAlignVertical]{0.9\textwidth}
						\vspace{\boxOffset}
						\begin{align*}
							&\frac{L}{2} \int_{\surf} \left\| \GradSurf\qten \right\|^2 + 2\orderp\meanc\left\langle \shop, \qten \right\rangle \\
													&\qquad + \left\| \shop \right\|^2\left( \Tr\qten^2 + \frac{\orderp^2}{2} \right)\mu
						\end{align*}
					\end{minipage}
				}
			\end{minipage}
		\end{minipage}
	\end{minipage}
\end{center}
with $\Dirf$ the director field in $\R^3$, normalization constraint $\left\| \Dirf \right\|^2= 1$, consistence $\Dirf\overset{\lambda\rightarrow 0}{\longrightarrow}\dirf$ 
and thin film $\surf_{\textup{TF}}=\surf\times [-\nicefrac{\lambda}{2},\nicefrac{\lambda}{2}]$ of thickness $\lambda$ and with normal $\normal$ on $\partial\surf_{\textup{TF}}$. For the upward direction on the left-hand side we use $\Qten = \orderp\left( \Dirf\otimes\Dirf-\frac{1}{3}\Id_{\R^3} \right)$. 
However, as in $\R^3$ the \frankOseenHelfrich\ energy
does not respect the head-to-tail symmetry in which $\dirf$ should be equivalent to $-\dirf$. Differences resulting from this drawback have been outlined in \cite{Balletal_ARMA_2011} in $\R^3$ and will also be present on surfaces, which will be shown below.
   
\subsection{Steepest Decent Methods}

For the \frankOseenHelfrich\ energy $ \PotEP = \PotEP[\xvar,\dirf] = \hatPotEP[\para,\dirf] $ and the \landauDeGennesHelfrich\ energy $ \PotEQ = \PotEQ[\xvar,\qten] = \hatPotEQ[\para,\qten] $ we consider the $ L^2 $-gradient flows \wrt\ $ \xvar\in\tangentT{^0}\surf $, $ \dirf\in\tangentT{^1}\surf $ and $ \proj_{\surf}\dsurf\dirf=0 $ or $  \qten\in\qspace $ and $ \proj_{\qspace}\dsurf\qten=0 $, respectively, \ie
\begin{align}
    -\frac{\delta \PotEP}{\dxvar} &= \kinConst\dot{\xvar} = \kinConst\vnor \\
    -\frac{\delta \PotEP}{\delta\dirf}
     &=\kinConstP\dot{\dirf} = \kinConstP\proj_{\surf}\frac{d\dirf}{dt}  = \kinConstP\left(\left\{\partial_t\dirfC^i \right\} - \vnor\shop\dirf\right)
\end{align} 
and
\begin{align}
    -\frac{\delta \PotEQ}{\dxvar} &= \kinConst\dot{\xvar} = \kinConst\vnor \\
    -\frac{\delta \PotEQ}{\delta\qten}
        &= \kinConstQ\dot{\qten}
         = \kinConstQ\proj_{\qspace}\frac{d\qten}{dt}
         = \kinConstQ\proj_{\surf}\frac{d\qten}{dt}  
         = \kinConstQ\left(\left\{\partial_t\qtenC^{ij} \right\} - \vnor\left(\shop\qten + \qten\shop\right)\right)
\end{align} 
where $\kinConst,\kinConstP,\kinConstQ\geq 0$ denote kinematic constants. 
With initial conditions we observe for $ \dirf $
\begin{align}
    \label{eq:gradientflow:p:vnor}
    \kinConst\vnor &= -\alpha\left(\laplaceBeltrami\meanc + \meanc\left( \frac{\meanc^2}{2} - 2\gaussc \right)\right)
             +\frac{\pnorm}{4}\meanc\left( \left\| \dirf \right\|^2 - 1 \right)^2 \notag\\
          &\quad  -K\left( \DivSurf\left( \left( \levicivita{\dirf}{\meanc} + 2\levicivita{\shop}{\dirf} \right)\dirf \right) 
                    + \left\langle \sigma^{\text{E}}_{\surf}, \proj_{\qspace}\shop \right\rangle\right) 
                    + \frac{\areaPen}{\areaZero^2}\left(\area-\areaZero\right)\meanc \\
    \label{eq:gradientflow:p:dirf}
    \kinConstP\dot{\dirf} &= K\left( \laplaceBochner\dirf - \shop^2\dirf \right) - \pnorm\left( \left\| \dirf \right\|^2 - 1 \right)\dirf
\end{align}
with Gaussian curvature $\gaussc$, surface divergence $\DivSurf$, directional derivative $\levicivitaPlain$, material time derivative $\dot{\dirf}$, Laplace-Beltrami operator $\laplaceBeltrami$, Bochner Laplacian $\laplaceBochner$, extrinsic surface Ericksen stress tensor $\sigma^{\text{E}}_{\surf} = (\GradSurf\dirf)^{T}\GradSurf\dirf + \shop\dirf\otimes\shop\dirf\in\tangentT{^2}\surf $ and orthogonal Q tensor projection $\proj_{\qspace}\shop = \shop - \frac{\meanc}{2}\g$ of the shape operator, see \autoref{sec:app:one} for details.

Similarly, with initial conditions we get the steepest decent method for the \landauDeGennesHelfrich\ energy
\begin{align}
    \label{eq:gradientflow:q:vnor}
    \kinConst\vnor &= -\left(\alpha+\frac{L}{2}\orderp^2\right)\left(\laplaceBeltrami\meanc + \meanc\left( \frac{\meanc^2}{2} - 2\gaussc \right)\right)\notag\\
          &\quad +\meanc\left( \left( a'- L\left( \frac{\meanc^2}{2} - 2\gaussc \right) \right)\Tr\qten^2 + c\Tr\qten^4 + C_1 \right)\notag\\
          &\quad -2L\DivSurf\left( \qten\levicivita{\shop}{\qten} - \levicivita{\qten\shop}{\qten} +\qten:(\GradSurf\qten)\shop + \frac{1}{2}\Tr\qten^2\GradSurf\meanc
                 + \orderp\left(\levicivita{\shop}{\qten} + \qten\GradSurf\meanc\right) \right)\notag\\
          &\quad -L\left( \left\langle (\GradSurf\qten)^{T_{(123)}}:\GradSurf\qten,\proj_{\qspace}\shop \right\rangle 
                          +\orderp\left\| \shop \right\|^2\left\langle \shop, \qten \right\rangle\right) + \frac{\areaPen}{\areaZero^2}\left(\area-\areaZero\right)\meanc \\
    \label{eq:gradientflow:q:qten}
    \kinConstQ\dot{\qten} 
          &= L\left( \laplaceBochner\qten - \orderp\meanc\proj_{\qspace}\shop \right) 
             - \left( L\left\| \shop \right\|^2+2a'+2c\Tr\qten^2 \right)\qten \formPeriod
\end{align}
For the definition of the used operators in the above equations we refer to \autoref{sec:opkunde} and \cite{Nitschkeetal_PRSA_2018}. Interestingly, even without bending rigidity, $\alpha = 0$, surface regularization in terms of Helfrich energy contributions is present.

We observe that
\begin{align}
    \frac{d}{dt}\PotEP &= \int_{\surf} \left( \frac{\delta \PotEP}{\dxvar} \right)\vnor + \left\langle\frac{\delta \PotEP}{\delta\dirf}, \dot{\dirf}  \right\rangle \mu = -\int_{\surf} \vnor^2 + \left\| \dot{\dirf} \right\|^2 \mu \le 0 \\
    \frac{d}{dt}\PotEQ &= \int_{\surf} \left( \frac{\delta \PotEQ}{\dxvar} \right)\vnor + \left\langle\frac{\delta \PotEQ}{\delta\qten}, \dot{\qten}  \right\rangle \mu = -\int_{\surf} \vnor^2 + \left\| \dot{\qten} \right\|^2 \mu \le 0 \formComma
\end{align}
see \autoref{sec:kinetic}.
Therefore, the energies $ \PotEP $ and $\PotEQ$ dissipate as long as the solution and the shape are not stationary.

\section{Numerical Approach} \label{sec:numerics}

To numerically solve the systems \eqref{eq:gradientflow:p:vnor}-\eqref{eq:gradientflow:p:dirf} and \eqref{eq:gradientflow:q:vnor}-\eqref{eq:gradientflow:q:qten} we use a semi-implicit Euler timestepping scheme, an operator splitting ansatz and the generic finite element approach proposed in \cite{Nestleretal_JCP_2019}, which is based on a reformulation of all operators and quantities in Cartesian coordinates and penalization of normal components. 
Applications of the latter approach can be found in \eg\ \cite{Nestleretal_JNS_2018,Reutheretal_PF_2018,Nitschkeetal_PRF_2019}.
For better readability we use the abbreviations 
\begin{align*}
    \tilde\alpha &:= \alpha+\frac{L}{2}\orderp^2 \\
    \beta_{\dirf} &:= - K\left( \left\langle \sigma^{\text{E}}_{\surf}, \proj_{\qspace}\shop \right\rangle\right) + \frac{\pnorm}{4}\meanc\left( \left\| \dirf \right\|^2 - 1 \right)^2 \\
    \boldsymbol{w}_{\dirf} &:= -K\left( \levicivita{\dirf}{\meanc} + 2\levicivita{\shop}{\dirf} \right)\dirf \\
    \beta_{\qten} &:= \meanc\left( \left( a'- L\left( \frac{\meanc^2}{2} - 2\gaussc \right) \right)\Tr\qten^2 + c\Tr\qten^4 + C_1 \right) \notag \\
    &\qquad- L\left( \left\langle (\GradSurf\qten)^{T_{(123)}}:\GradSurf\qten,\proj_{\qspace}\shop \right\rangle + \orderp\left\| \shop \right\|^2\left\langle \shop, \qten \right\rangle\right) \\
    \boldsymbol{w}_{\qten} &:= - 2L\left( \qten\levicivita{\shop}{\qten} - \levicivita{\qten\shop}{\qten} +\qten:(\GradSurf\qten)\shop + \frac{1}{2}\Tr\qten^2\GradSurf\meanc + \orderp\left(\levicivita{\shop}{\qten} + \qten\GradSurf\meanc\right)\right) 
\end{align*}
in the following derivations.

\subsection{Timediscretization}
\label{sec:timediscretization}

Let $0 = t^0 < t^1 < t^2 < \dots$ be a partition of the time with timestep width $\tau^m:=t^{m}-t^{m-1}$.
Each variable/quantity with a superscript index $m$ corresponds to the respective variable/quantity at time $t^m$. 
The overall operator splitting algorithm for both systems \eqref{eq:gradientflow:p:vnor}-\eqref{eq:gradientflow:p:dirf} and \eqref{eq:gradientflow:q:vnor}-\eqref{eq:gradientflow:q:qten} reads as follows: 
for $m=1,2,\dots$ do 
\begin{enumerate}
    \item Move geometry according to eq. \eqref{eq:paraDot}, \ie\ in the timediscrete setting
            \begin{align}
                \label{eq:paraDot:discrete}
                \para^{m} &= \para^{m-1} + \tau^m\vnor^{m-1}\normal^{m-1}
            \end{align} 
            with the parametrization of the initial geometry $\para^{0}$ and corresponding initial normal vector $\normal^{0}$.
    \item Update normal vector according to $\partial_t\normal = -\GradSurf\vnor$, see eq. \eqref{eq:normalevolution}. Thus, in the timediscrete setting the normal vector at the new timestep $t^m$ is determined by
            \begin{align}
                \normal^{m} &= \normal^{m-1} - \tau^m\GradSurf\vnor^{m-1} \formPeriod
            \end{align}
    \item Update all other geometric quantities, \ie\ the mean curvature $\meanc^{m}$, the Gaussian curvature $\gaussc^{m}$, the projection $\proj_{\surf}^{m}$ and the shape operator $\shop^{m}$, by using the prior computed normal vector $\normal^{m}$. 
    \item Update director field $\dirf^{m}$/Q tensor field $\qten^{m}$ according to eqs. \eqref{eq:gradientflow:p:dirf} and \eqref{eq:gradientflow:q:qten}, which read in the timediscrete setting
    \begin{align}
        \label{eq:gradientflow:p:dirf:time_discrete}
        \kinConstP\textup{d}^m_{\dirf} &= K\left( \laplaceBochner\dirf^{m} - (\shop^{m})^2\dirf^{m} \right) - \pnorm f^{\dirf}(\dirf^{m}, \dirf^{m-1})
    \end{align}
    and
    \begin{align}
        \label{eq:gradientflow:q:qten:time_discrete}
        \kinConstQ\textup{d}^m_{\qten} &= L\left( \laplaceBochner\qten^{m} - \orderp\meanc\proj_{\qspace}^{m}\shop^{m} \right) - \left( L\left\| \shop^{m} \right\|^2+2a' \right)\qten^{m} - 2cf^{\qten}(\qten^{m}, \qten^{m-1})\formComma
    \end{align}
    respectively.
    Thereby, $\textup{d}^m_{\dirf}:=\nicefrac{1}{\tau^m}(\proj_{\surf}^{m}\dirf^{m}-\proj_{\surf}^{m}\dirf^{m-1})$ and $\textup{d}^m_{\qten}:=\nicefrac{1}{\tau^m}(\proj_{\surf}^{m}\qten^{m}\proj_{\surf}^{m}-\proj_{\surf}^{m}\qten^{m-1}\proj_{\surf}^{m})$ denote the discrete material time derivatives and $f^{\dirf}(\dirf^{m}, \dirf^{m-1})$ as well as $f^{\qten}(\qten^{m}, \qten^{m-1})$ are linearizations of the terms $(\|\dirf^{m}\|^2 - 1)\dirf^{m}$ and $\Tr(\qten^{m})^2\qten^{m}$, respectively, see \cite{Nestleretal_JNS_2018}.
    \item Update normal velocity $\vnor^{m}$ according to eqs. \eqref{eq:gradientflow:p:vnor} and \eqref{eq:gradientflow:q:vnor}, which read in the timediscrete setting 
    \begin{align}
        \label{eq:gradientflow:p:vnor:time_discrete}
        \kinConst\vnor^{m} &= -\alpha\left(\laplaceBeltrami\meanc^{m} + \meanc^{m}\left( \frac{(\meanc^{m})^2}{2} - 2\gaussc^{m} \right)\right) + \beta_{\dirf}^{m} + \DivSurf\boldsymbol{w}_{\dirf}^{m} \notag \\
        &\qquad+ \frac{\areaPen}{\areaZero^2}\left(\area^{m}-\areaZero\right)\meanc^{m}
    \end{align}
    and 
    \begin{align}
        \label{eq:gradientflow:q:vnor:time_discrete}
        \kinConst\vnor^{m} &= -\tilde\alpha\left(\laplaceBeltrami\meanc^{m} + \meanc^{m}\left( \frac{(\meanc^{m})^2}{2} - 2\gaussc^{m} \right)\right) + \beta_{\qten}^{m} + \DivSurf\boldsymbol{w}_{\qten}^{m} \notag \\
        &\qquad+ \frac{\areaPen}{\areaZero^2}\left(\area^{m}-\areaZero\right)\meanc^{m} \formComma
    \end{align}
    respectively. 
\end{enumerate}

\subsection{Spacediscretization}
The remaining step is to discretize eqs. \eqref{eq:gradientflow:p:dirf:time_discrete},
\eqref{eq:gradientflow:q:qten:time_discrete},
\eqref{eq:gradientflow:p:vnor:time_discrete} and \eqref{eq:gradientflow:q:vnor:time_discrete} from the above algorithm in space by using either the generic surface finite element method for tensor-valued surface PDEs proposed in \cite{Nestleretal_JCP_2019} or the surface finite element method for scalar-valued surface PDEs from \cite{Dziuketal_AN_2013}. 
Let $\surfDiscrete=\surfDiscrete(t)|_{t=t^m}$ be an interpolation of the surface $\surf=\surf(t)|_{t=t^m}$ at time $t^m$ such that $\surfDiscrete := \bigcup_{T\in\triangulation}T$, where $\triangulation$ denotes a conforming triangulation. 
Furthermore, the finite element space is introduced as $\feSpace := \lbrace v\in\mathcal{C}^0(\surfDiscrete) : v|_T\in \mathcal{P}^1(T), \forall v\in\triangulation \rbrace$ with $\mathcal{C}^k(\surfDiscrete)$ the space of $k$-times continuously differentiable functions on $\surfDiscrete$ and $\mathcal{P}^l(T)$ polynomials of degree $l$ on the triangle $T\in\triangulation$. 
We use the finite element space $\feSpace$ twice as trail and as test space and additionally introduce the $L_2$ inner product on $\surfDiscrete$, \ie\ $\weak{\alpha}{\beta} := \int_{\surfDiscrete}\langle\alpha,\beta\rangle\dS$. 
Thus, the finite element approximations of eqs. \eqref{eq:gradientflow:p:dirf:time_discrete}, \eqref{eq:gradientflow:q:qten:time_discrete}, \eqref{eq:gradientflow:p:vnor:time_discrete} and \eqref{eq:gradientflow:q:vnor:time_discrete} read: 
find $\dirf^{m}\in\feSpace^{3}$ such that $\forall\mathbf{v}\in\feSpace^{3}$
\begin{align}
    \label{eq:gradientflow:p:dirf:timeandspace_discrete}
    \weak{\kinConstP\textup{d}^m_{\dirf}}{\mathbf{v}} + \weak{K\GradSurf\dirf^{m}}{\GradSurf\mathbf{v}} &= \weak{\tanPenP\normal^T\dirf^{m}\normal - K(\shop^{m})^2\dirf^{m} - \pnorm f^{\dirf}(\dirf^{m}, \dirf^{m-1})}{\mathbf{v}}
    \formComma
\end{align}
find $\qten^{m}\in\feSpace^{3\times3}$ such that $\forall\boldsymbol{V}\in\feSpace^{3\times3}$
\begin{align}
    \label{eq:gradientflow:q:qten:timeandspace_discrete}
    \weak{\kinConstQ\textup{d}^m_{\qten}}{\boldsymbol{V}} &= -\weak{L\GradSurf\qten^{m} }{\GradSurf\boldsymbol{V}} - \weak{L\left( \orderp\meanc\proj_{\qspace}^{m}\shop^{m} \right) + (L\left\| \shop^{m} \right\|^2+2a')\qten^{m} }{\boldsymbol{V}} \notag\\
    &\qquad + \weak{\tanPenQ\left(\qten^{m}\proj_{\surf}^\perp + \proj_{\surf}^\perp\qten^{m} + \proj_{\surf}^\perp\qten^{m}\proj_{\surf}^\perp\right) - 2cf^{\qten}(\qten^{m}, \qten^{m-1})}{\boldsymbol{V}}
    \formComma
\end{align}
find $\vnor^{m}\in\feSpace$ such that $\forall\zeta\in\feSpace$
\begin{align}
    \label{eq:gradientflow:p:vnor:timeandspace_discrete}
    \weak{\kinConst\vnor^{m}}{\zeta} &= 
    \weak{\alpha\GradSurf\meanc^{m} - \boldsymbol{w}_{\dirf}^{m}}{\GradSurf\zeta} - \weak{\alpha\meanc^{m}\left( \frac{(\meanc^{m})^2}{2} - 2\gaussc^{m} \right)}{\zeta} \notag \\
    &\qquad+ \weak{\beta_{\dirf}^{m} + \frac{\areaPen}{\areaZero^2}\left(\area^{m}-\areaZero\right)\meanc^{m}}{\zeta}
\end{align}
and find $\vnor^{m}\in\feSpace$ such that $\forall\gamma\in\feSpace$
\begin{align}
    \label{eq:gradientflow:q:vnor:timeandspace_discrete}
    \weak{\kinConst\vnor^{m}}{\gamma} &= 
    \weak{\tilde\alpha\GradSurf\meanc^{m} - \boldsymbol{w}_{\qten}^{m}}{\GradSurf\gamma} - \weak{\tilde\alpha\meanc^{m}\left( \frac{(\meanc^{m})^2}{2} - 2\gaussc^{m} \right)}{\gamma} \notag \\
    &\qquad+ \weak{\beta_{\qten}^{m} + \frac{\areaPen}{\areaZero^2}\left(\area^{m}-\areaZero\right)\meanc^{m}}{\gamma}
    \formComma
\end{align}
respectively.
Note that we here use the same symbols for the extended director field/Q tensor field as for the local surface quantities. 
According to the generic surface finite element method from \cite{Nestleretal_JCP_2019} penalty terms $(\tanPenP\normal^T\dirf^{m}\normal\ , \ \mathbf{v})$ and $(\tanPenQ(\qten^{m}\proj_{\surf}^\perp + \proj_{\surf}^\perp\qten^{m} + \proj_{\surf}^\perp\qten^{m}\proj_{\surf}^\perp)\ , \ \boldsymbol{V})$ with $\proj_{\surf}^\perp=\normal\normal^T$ and penalty parameters $\tanPenP$ and $\tanPenQ$ are added to eqs. \eqref{eq:gradientflow:p:dirf:timeandspace_discrete} and \eqref{eq:gradientflow:q:qten:timeandspace_discrete} to ensure tangentiality of the respective director/Q tensor field. 
Eqs. \eqref{eq:gradientflow:p:vnor:timeandspace_discrete} and \eqref{eq:gradientflow:q:vnor:timeandspace_discrete} have to be stabilized, which is realized by artificial diffusion following ideas of \cite{Smereka_JSC_2003} for surface diffusion. 
In particular, we add the terms $(D(\GradSurf\vnor^{m-1} - \GradSurf\vnor^{m})\ , \ \GradSurf\zeta)$ and $(D(\GradSurf\vnor^{m-1} - \GradSurf\vnor^{m})\ , \ \GradSurf\gamma)$ with the artificial diffusion parameter $D$ to the right hand sides of eqs. \eqref{eq:gradientflow:p:vnor:timeandspace_discrete} and \eqref{eq:gradientflow:q:vnor:timeandspace_discrete}, respectively. 
This approach is significantly simpler than other proposed discretization schemes, such as \cite{Rusu_IFB_2005,Barrettetal_JCP_2007,Dziuk_NM_2008}, but also leads to sufficient accuracy for our purposes, see the convergence study in \autoref{sec:numericaltests}. 
Recently, at least for mean curvature flow, a similar method which relays on the evolution of geometric quantities, has been used and analyzed in \cite{Kovacsetal_NM_2019}. 
For more details, especially for evaluating the local inner products in the $L_2$ inner products for the extended director/Q tensor field, we refer to \cite{Nestleretal_JCP_2019}.

\section{Results} \label{sec:results}

In the following we compare both models \eqref{eq:gradientflow:p:vnor}-\eqref{eq:gradientflow:p:dirf} and \eqref{eq:gradientflow:q:vnor}-\eqref{eq:gradientflow:q:qten}. 
In case of the \frankOseenHelfrich\ model we simply consider the director field $\dirf$ for visualization. 
To visualize the Q tensor field $\qten$ resulting from the \landauDeGennesHelfrich\ model, we consider the principal director of the Q tensor, \ie\ the eigenvector for the eigenvalue with the largest absolute value of the Q tensor. 
This director has no direction as it is $\pi$ symmetric and thus represents the same Q tensor if rotated by an angle of $\pi$ in the tangent plane. 

Firstly, we prescribe the normal velocity of the surface and determine the response of the appropriate quantity due to shape changes. 
Secondly, by changing the model parameters the response of the surface to a given director/Q tensor field is investigated. 
Thirdly, the full model including the interplay of the liquid crystal dynamics and the surface evolution is considered. 
Thereby, we consider known examples for fixed surfaces from the literature and let the surface evolve. 
The resulting stationary shapes are analysed in detail and compared to their intrinsic counterparts.
All used model and simulation parameters are shown in \autoref{tab:parameters}.

\subsection{Prescribed Normal Velocity}
    
Here, we only consider eqs. \eqref{eq:gradientflow:p:dirf} and  \eqref{eq:gradientflow:q:qten} on a prescribed evolving surface $\surf(t)$.
We focus on the alignment of the director to minimal curvature lines and the localization of topological defects according to curvature. For detailed discussions on stationary surfaces we refer to \cite{Nestleretal_JNS_2018,Nitschkeetal_PRSA_2018}. 
We adapt the movement used in \cite{Nitschkeetal_PRF_2019}, where a rotationally symmetric ellipsoidal shape with major axes parameters $a_0=a_1=1$ and $a_2=1.25$ is considered as starting geometry. 
These parameters are now considered to be time-dependent, such that the surface area is preserved over time. 
During the evolution the ellipsoid collapses to a sphere and deforms back to an ellipsoid with a different orientation.
As initial condition we use the close-to-equilibrium solutions obtained on the stationary initial ellipsoid with initial conditions
\begin{align*}
    \dirf|_{t=0} &= \frac{1}{\|\GradSurf\psi_0\|}\GradSurf\psi_0 \formComma \quad
    \psi_0 := x + y + z
\end{align*}
and 
\begin{align*}
    \qten|_{t=0} &= \orderp\left(\left(\normal\times\dirf_{\qten}\right)\otimes\left(\normal\times\dirf_{\qten}\right) - \frac{1}{2}\proj_{\surf}\right) \formComma \quad
    \dirf_{\qten} = \frac{1}{\|\proj_{\surf}\tilde{\dirf}_{\qten}\|}\proj_{\surf}\tilde{\dirf}_{\qten} \formComma
    \quad
    \tilde{\dirf}_{\qten} = \begin{cases}
        (1,4,0)^T \!\!& , z \geq 0 \\
        (-4,1,0)^T \!\!& , z < 0
    \end{cases}
\end{align*}
for $\para = (x, y ,z)^T\in\surf$.
\autoref{fig:director_response:results} shows the results which clearly indicate the alignment of the directors to the minimal curvature lines.
To be more precise, in the beginning the directors are aligned in a north-south symmetry. 
After passing the sphere geometry an ellipsoid with different orientation evolves and both directors start to rearrange to be aligned in an east-west symmetry, as expected. 
Interestingly, the director resulting from the \landauDeGennes\ energy aligns faster than the director resulting from the \frankOseen\ energy, which is a first indication of the existing differences of both models and is consistent with the larger solution space of the \landauDeGennes\ energy. 

\begin{table}[!h]
	\centering
	\begin{tabular}{rrrrrrr}
		\hline\noalign{\smallskip}
		& \autoref{fig:director_response:results}-\ref{fig:director_response:angle} & \autoref{fig:surface_response:results}-\ref{fig:surface_response:schematic} & \autoref{fig:results:pField}-\ref{fig:results:schematic} & \autoref{fig:results:qtensor2}-\ref{fig:results:intrinsic} & \autoref{fig:convergence_study} & \autoref{fig:surfaceAreaTest} \\
		\noalign{\smallskip}\hline\noalign{\smallskip}
        $\kinConstP$ & 
            $1$ & 
            $10^8$ & 
            $1$ & 
            $1$ & 
            -- & 
            $1$ \\
        $K$          & 
            $0.5$ & 
            $1$ & 
            $1$ & 
            $1$ & 
            -- & 
            $1$ \\
        $\pnorm$     & 
            $50$ & 
            $100$ & 
            $100$ & 
            $100$ & 
            -- & 
            $100$ \\
        $\tanPenP$   & 
            $10^4$ & 
            $10^5$ & 
            $10^3$ & 
            $10^3$ & 
            -- & 
            $10^3$ \\
		\noalign{\smallskip}\hline\noalign{\smallskip}
        $\kinConstQ$ & 
            $1$ & 
            $10^8$ & 
            $1$ & 
            $1$ & 
            -- & 
            -- \\
        $L$          & 
            $0.1966$ & 
            $0.3931$ &
            $0.3931$ &
            $0.3931$ &
            -- & 
            -- \\
        $a$          & 
            $-6.4862$ & 
            $-13.3333$ & 
            $-13.3333$ & 
            $-13.3333$ & 
            -- & 
            -- \\
        $b$          & 
            $-4.9138$ & 
            $-10$ & 
            $-10$ & 
            $-10$ & 
            -- & 
            -- \\
        $c$          & 
            $9.8275$ & 
            $20$ & 
            $20$ & 
            $20$ & 
            -- & 
            -- \\
        $\tanPenQ$   & 
            $10^3$ & 
            $10^3$ & 
            $10^3$ & 
            $10^3$ & 
            -- & 
            -- \\
		\noalign{\smallskip}\hline\noalign{\smallskip}
        $\alpha$     & 
            -- & 
            $0.5$ & 
            $0.5$ & 
            $0.8$ & 
            $1$ & 
            $0.5$ \\
		$\kinConst$  & 
            -- & 
            $100$ & 
            $100$ & 
            $100$ & 
            $100$ & 
            $100$ \\
		$\areaPen$   & 
            -- & 
            $500$ & 
            $500$ & 
            $500$ & 
            $500$ & 
            -- \\
        $D$          & 
            -- & 
            $166.6$ & 
            $166.6$ & 
            $166.6$ & 
            $166.6$ & 
            $166.6$ \\
        $\tau_m$     & 
            $0.1$ & 
            $0.0025$ & 
            $0.01$ & 
            $0.001$ & 
            $0.01$ & 
            $0.025$ \\
		\noalign{\smallskip}\hline
	\end{tabular}
 	\caption{Used model and simulation parameters. All values are treated as non-dimensional.}
 	\label{tab:parameters}
\end{table}

\begin{figure}[!h]
    \centering
    \ifthenelse{\boolean{useSnapshots}}
    {
		\ifthenelse{\boolean{forArxiv}}
		{
			\includegraphics[width=0.75\textwidth]{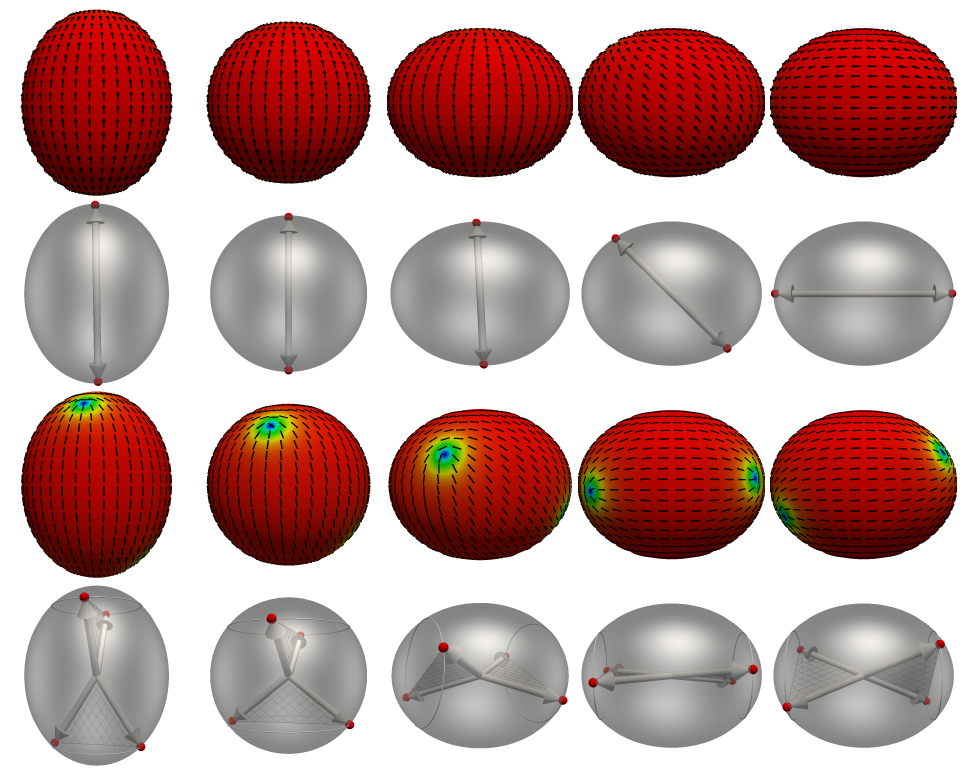}
		}
		{
			\includegraphics[width=\textwidth]{snapshots/fig_1.png}
		}
    }
    {
		\ifthenelse{\boolean{forArxiv}}
		{
			\def\picwidth{0.15\textwidth}
		}
		{
			\def\picwidth{0.19\textwidth}
		}
		\begin{minipage}{\textwidth}
			\centering
			\includegraphics[width=\picwidth]{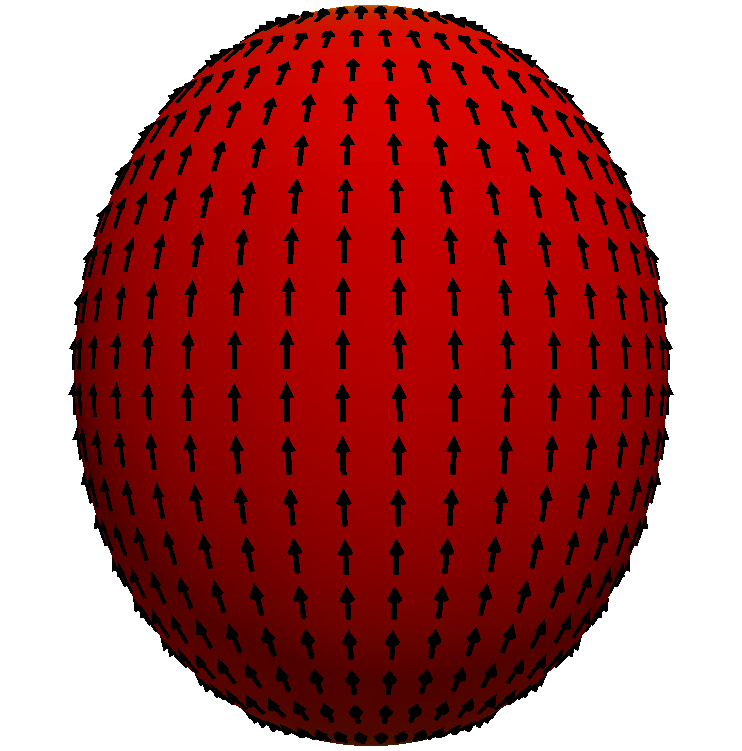}
			\includegraphics[width=\picwidth]{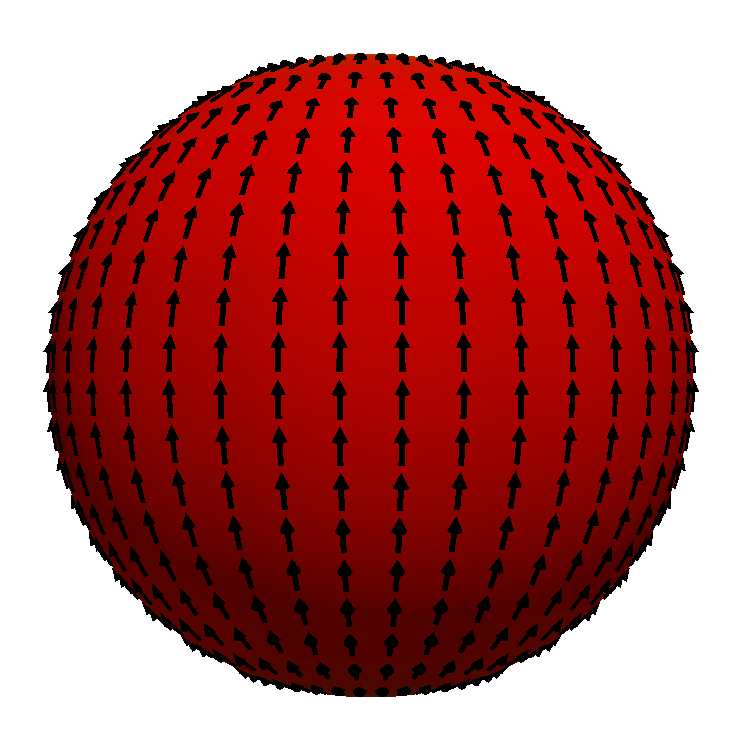}
			\includegraphics[width=\picwidth]{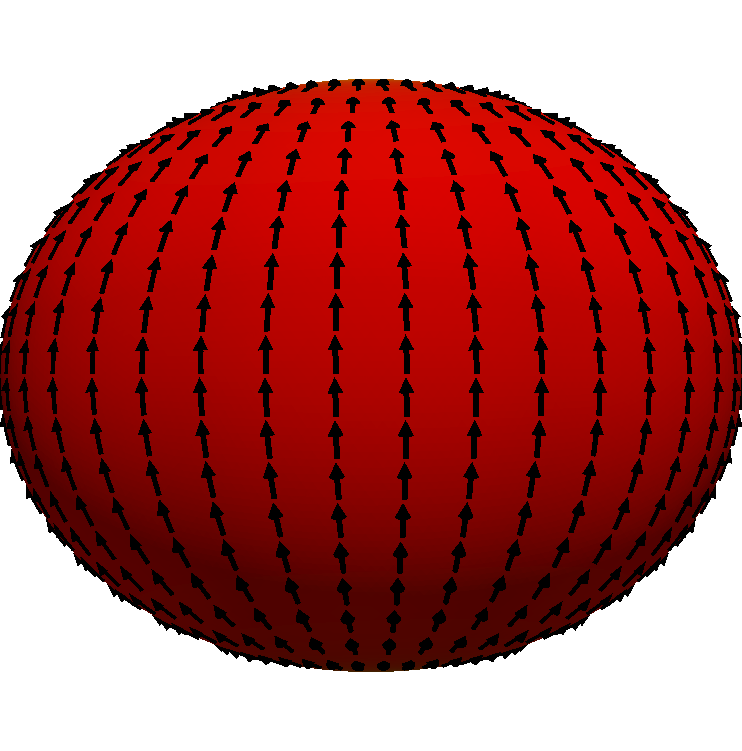}
			\includegraphics[width=\picwidth]{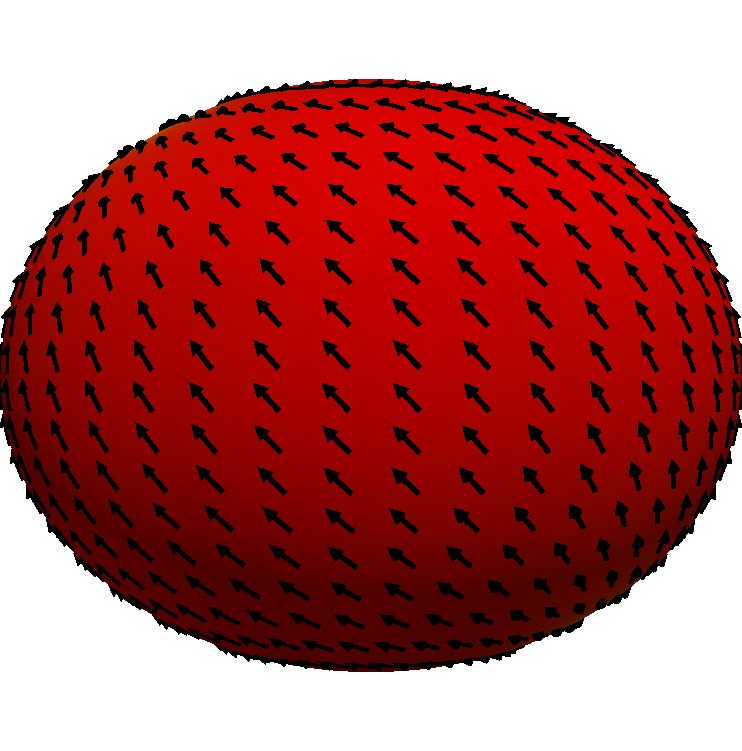}
			\includegraphics[width=\picwidth]{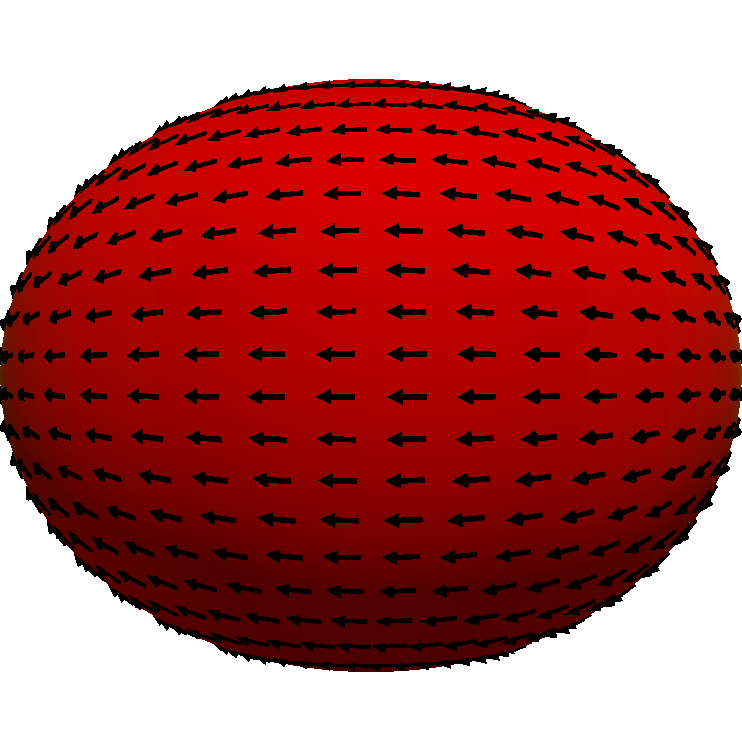} \\
			\includegraphics[width=\picwidth]{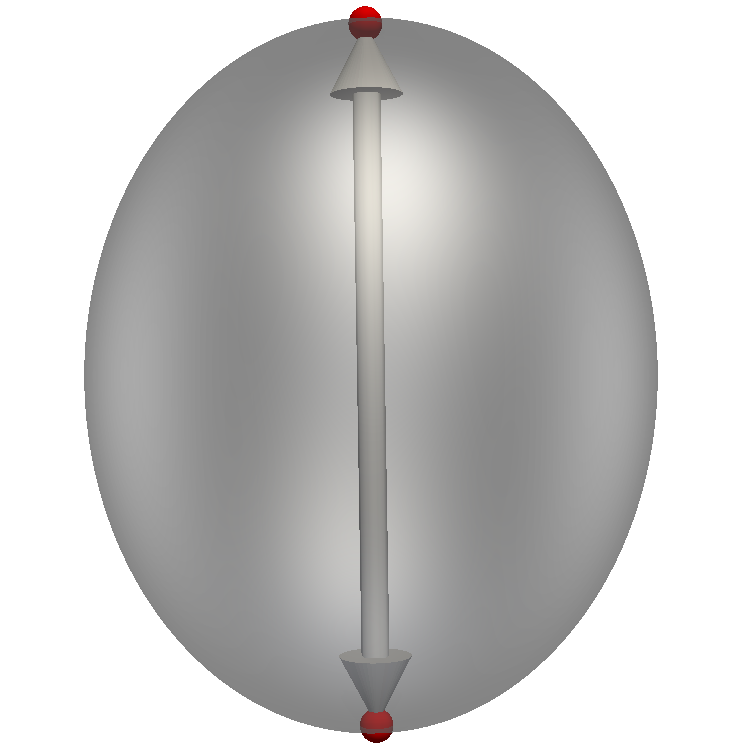}
			\includegraphics[width=\picwidth]{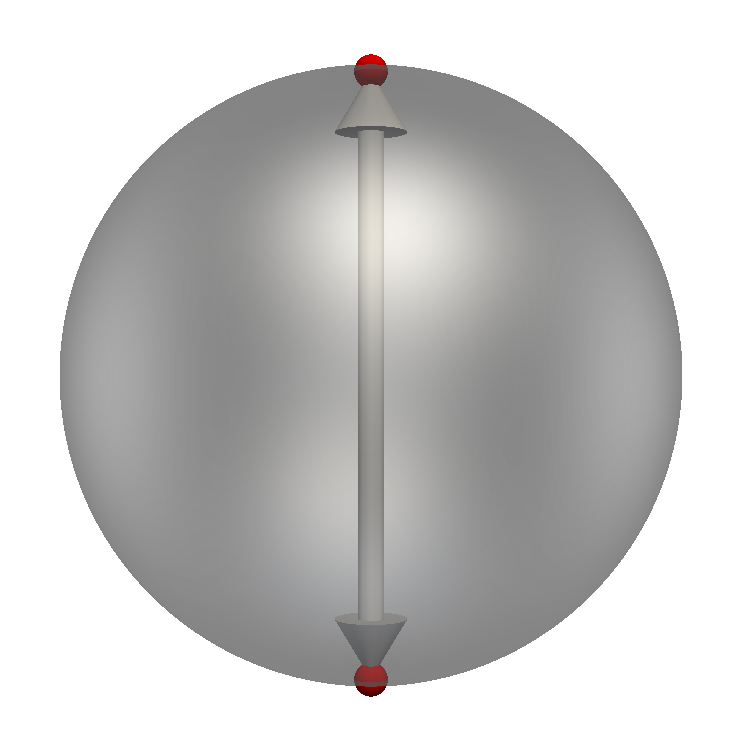}
			\includegraphics[width=\picwidth]{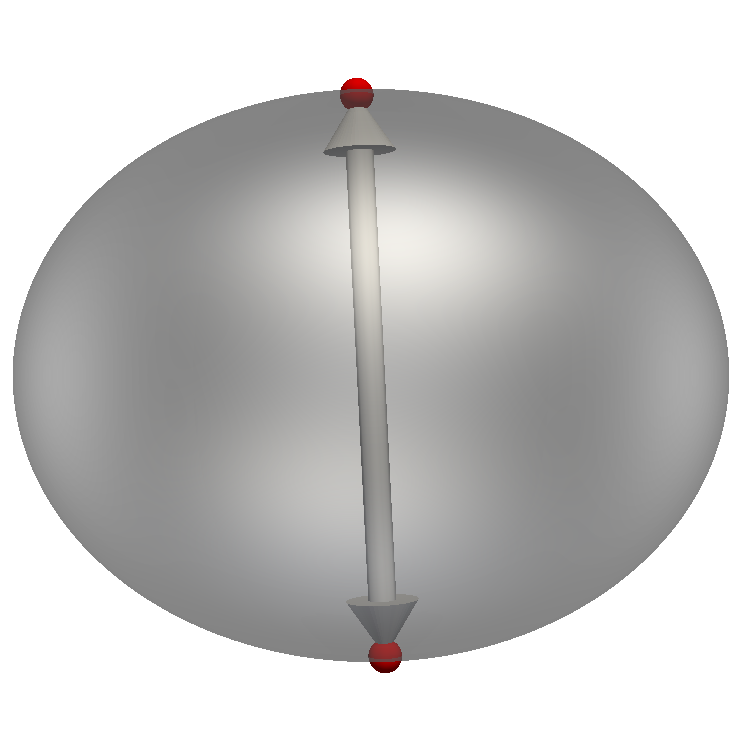}
			\includegraphics[width=\picwidth]{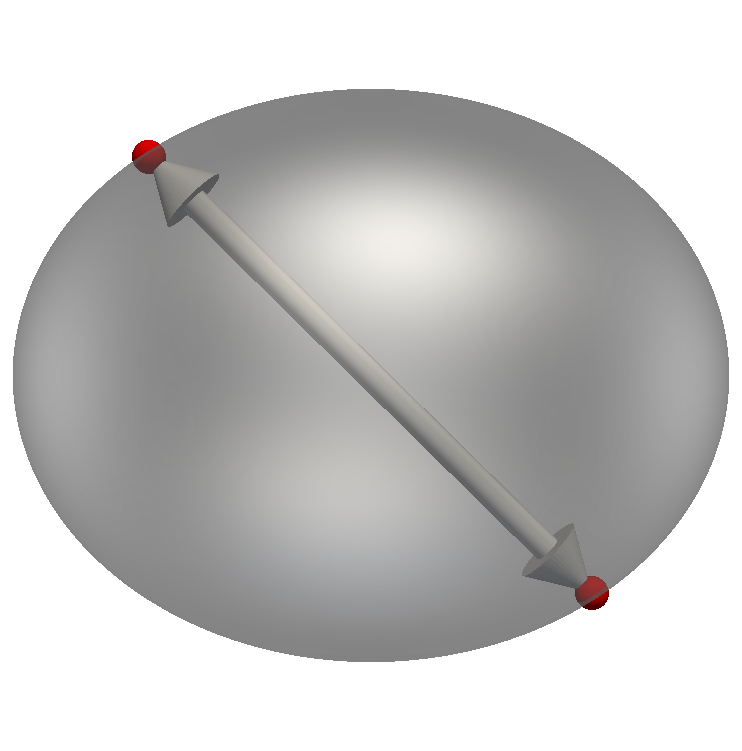}
			\includegraphics[width=\picwidth]{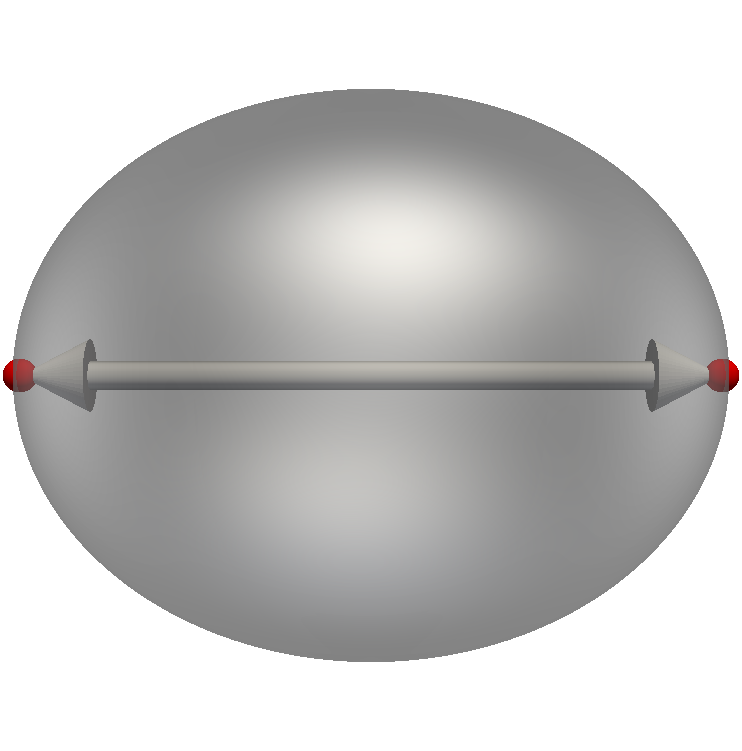}
		\end{minipage}
		\begin{minipage}{\textwidth}
			\centering
			\includegraphics[width=\picwidth]{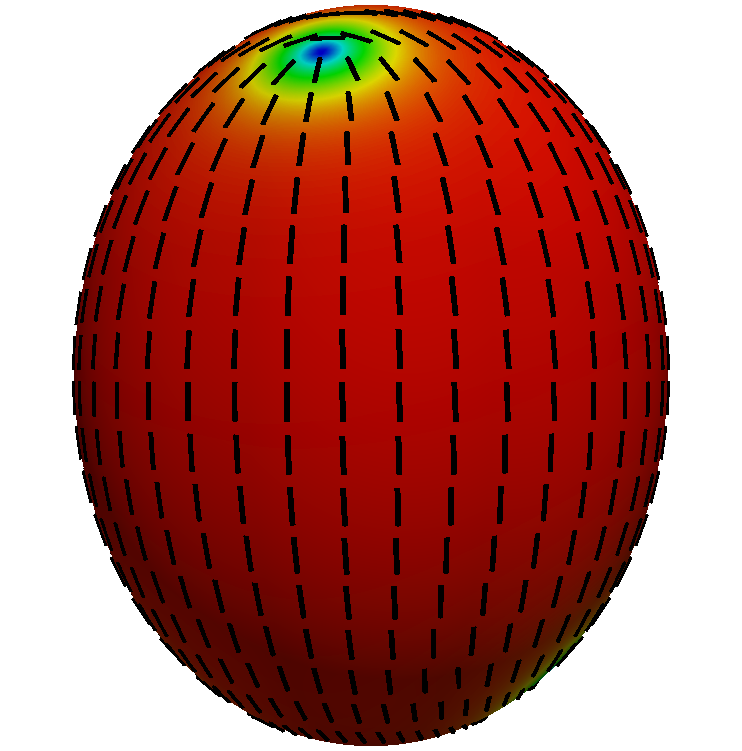}
			\includegraphics[width=\picwidth]{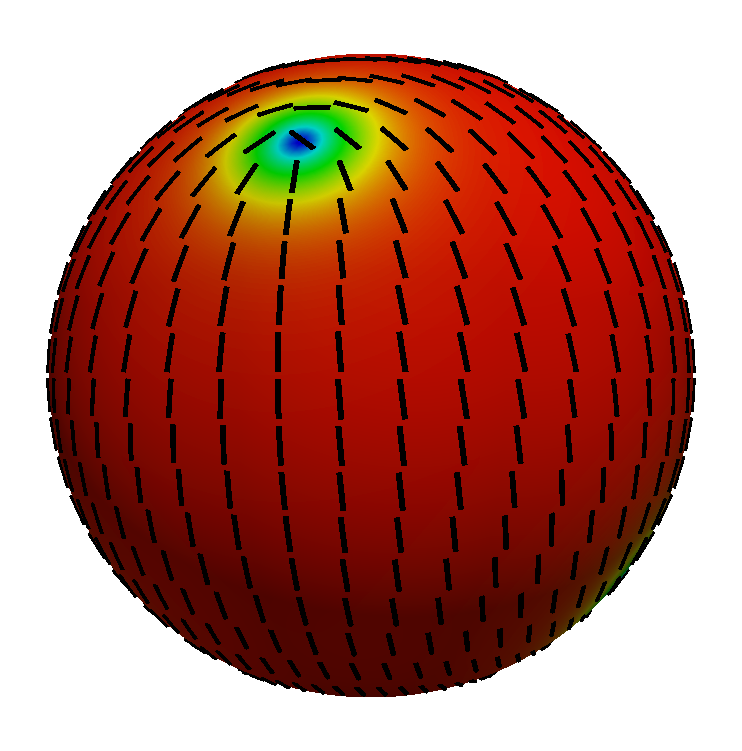}
			\includegraphics[width=\picwidth]{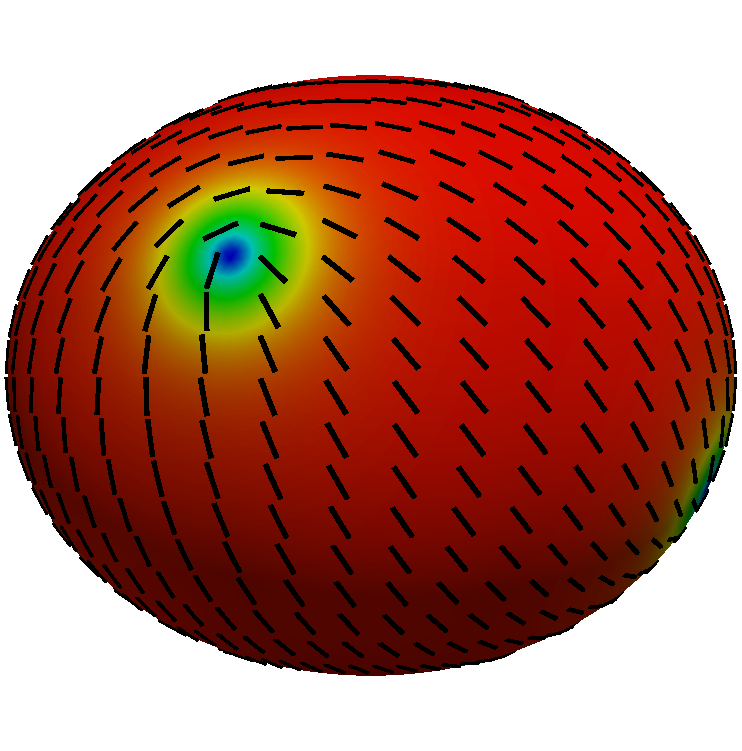}
			\includegraphics[width=\picwidth]{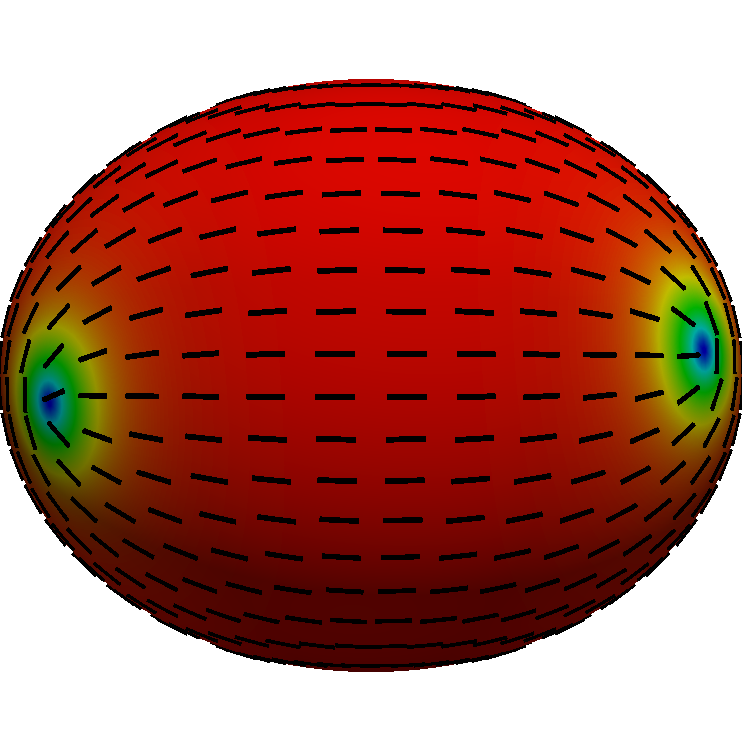}
			\includegraphics[width=\picwidth]{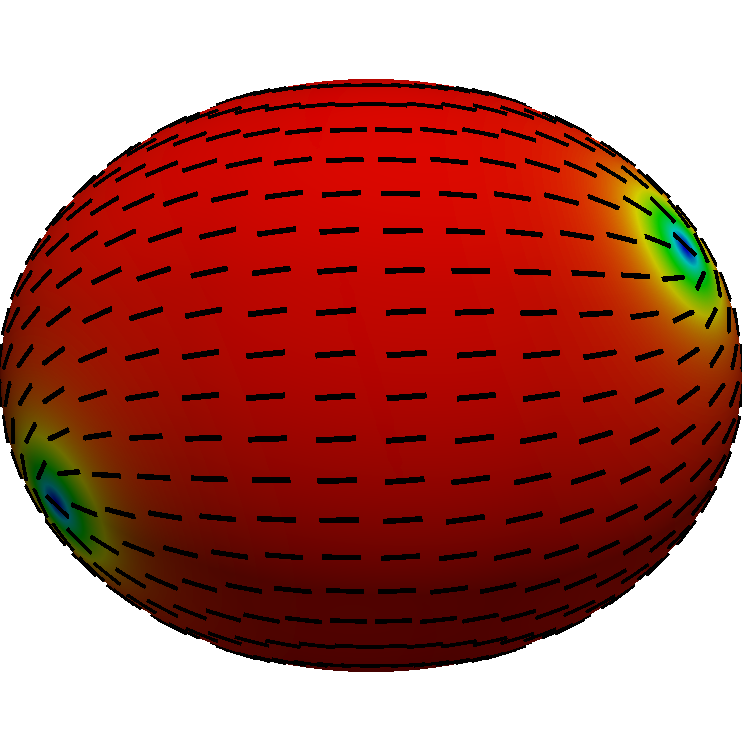} \\
			\includegraphics[width=\picwidth]{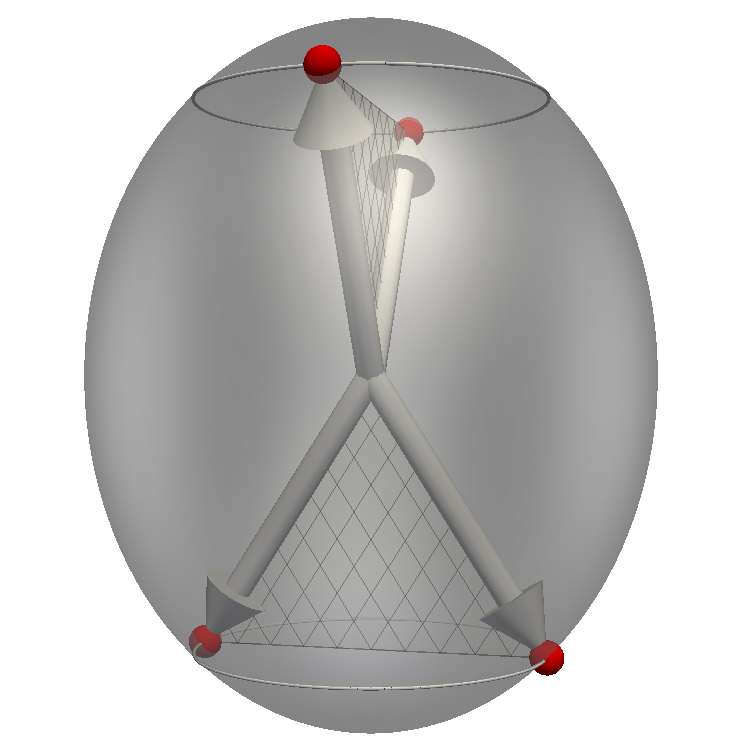}
			\includegraphics[width=\picwidth]{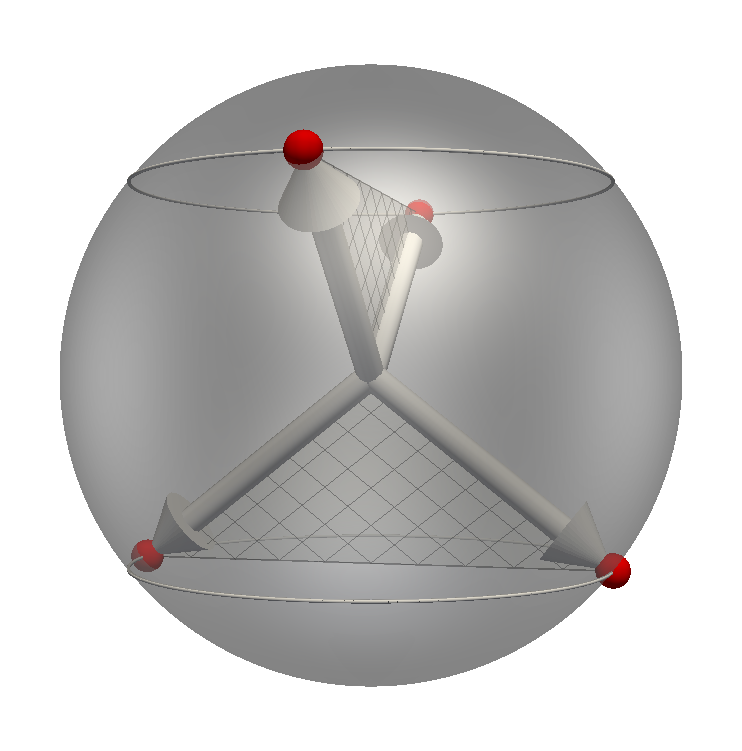}
			\includegraphics[width=\picwidth]{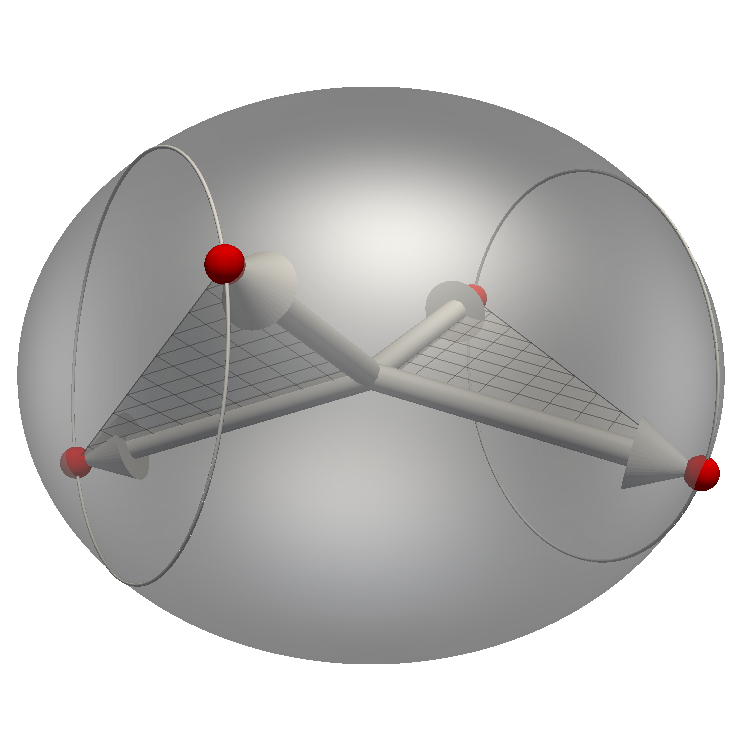}
			\includegraphics[width=\picwidth]{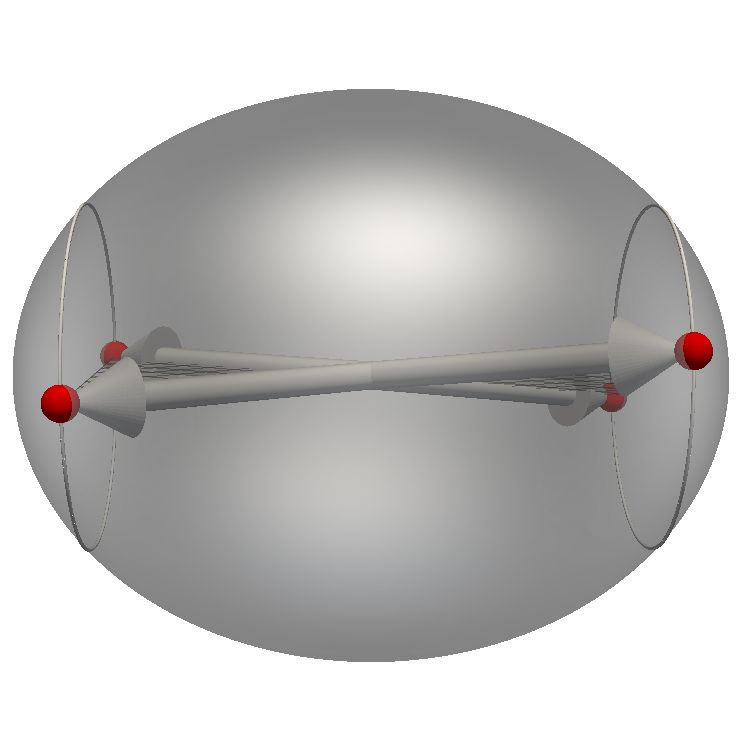}
			\includegraphics[width=\picwidth]{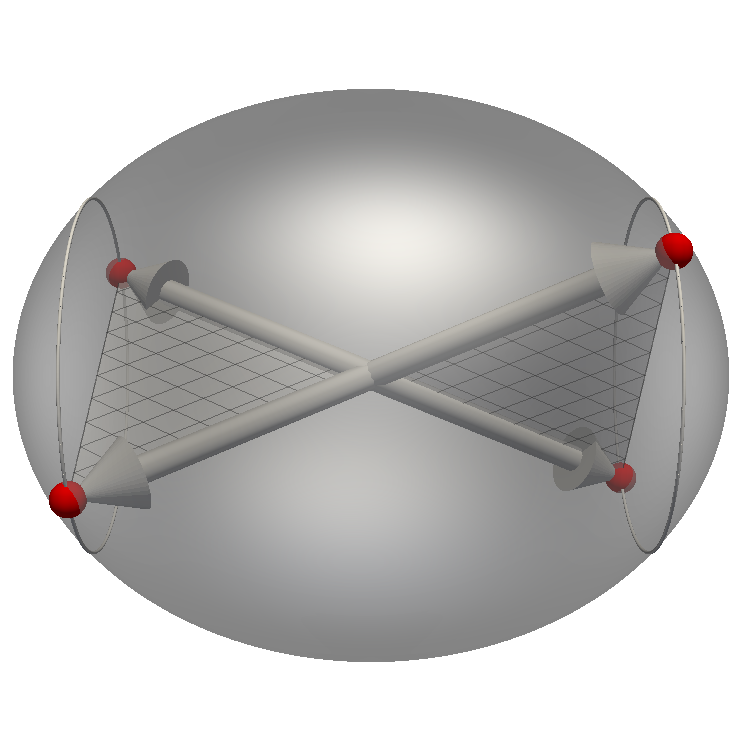}
		\end{minipage}
    }
    \caption{First and second row: Evolution of the director field $\dirf$ for $t=0, 40, 90, 110, 140$ on the prescribed evolving ellipsoid (top row) and respective $+1$ defect locations visualized as red dots and gray arrows (bottom row). 
    Third and fourth row: Evolution of the Q tensor field $\qten$ for $t=0, 40, 70, 110, 300$ on the prescribed evolving ellipsoid (top row) and respective $+\nicefrac{1}{2}$ defect locations visualized as red dots and gray arrows (bottom row). 
    The defect pairs with the two minimal angles to each other are indicated by the triangular planes.
    The gray circles indicate rotationally symmetric solutions. 
    In both cases -- the \frankOseenHelfrich\ as well as the \landauDeGennesHelfrich\ model -- the response of the defects to changes in curvature and the alignment along the minimal curvature lines can be observed.}
    \label{fig:director_response:results}
\end{figure}
    
More obvious differences can be observed at the defects, shown in Figure \ref{fig:director_response:results} as regions of disturbed orientational order (first and third row) and highlighted as red dots and gray arrows (second and fourth row). For the \frankOseenHelfrich\ model the expected defect configuration with two $+1$ defects, which maximize their distance, and for the \landauDeGennesHelfrich\ model the tetrahadral configuration with four $+\nicefrac{1}{2}$ defects is found. On a sphere they are known as global minimizers, subject to rotation \cite{Lubenskyetal_JPII_1992}. On an ellipsoid the two $+1$ defects are know to be located at umbilical points \cite{Ehrigetal_PRE_2017}, which here correspond with the maxima of curvature at the long axis. The equilibrium configuration is thus uniquely determined. The four $+\nicefrac{1}{2}$ defects, even if also geometrically attracted to these points are pushed away, due to repulsive forces between equally charged defects. The competition between these forces determines the equilibrium configuration, which is still free to rotate along the long axis of the ellipsoid \cite{Alaimoetal_SR_2017}. The dynamics follow these equilibrium configurations, with a short delay required to readjust the defect positions, see the solution at $t = 110$ for both cases. Interestingly for the \landauDeGennesHelfrich\ model this adjustment also requires to pass situation where the four $+\nicefrac{1}{2}$ defects are arranged in a planar configuration, see Figure \ref{fig:director_response:angle}.

\begin{figure}[!h]
    \centering
    \ifthenelse{\boolean{useSnapshots}}
    {
		\ifthenelse{\boolean{forArxiv}}
		{
			\includegraphics[width=0.5625\textwidth]{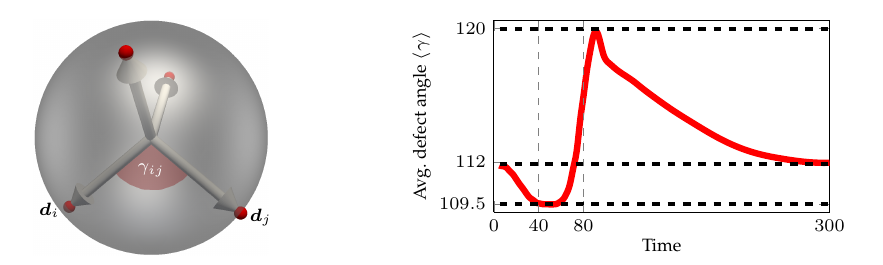}
		}
		{
			\includegraphics[width=0.75\textwidth]{snapshots/fig_2.png}
		}
	}
    {
		\ifthenelse{\boolean{forArxiv}}
		{
			\begin{minipage}{0.75\textwidth}
		}
		{
			\begin{minipage}{\textwidth}
		}
			\centering
			\begin{minipage}{0.49\textwidth}
				\centering
				\begin{minipage}{0.49\textwidth}
					\centering
					\begin{tikzpicture}
						\node (pic) at (0,0) {\includegraphics[width=\textwidth]{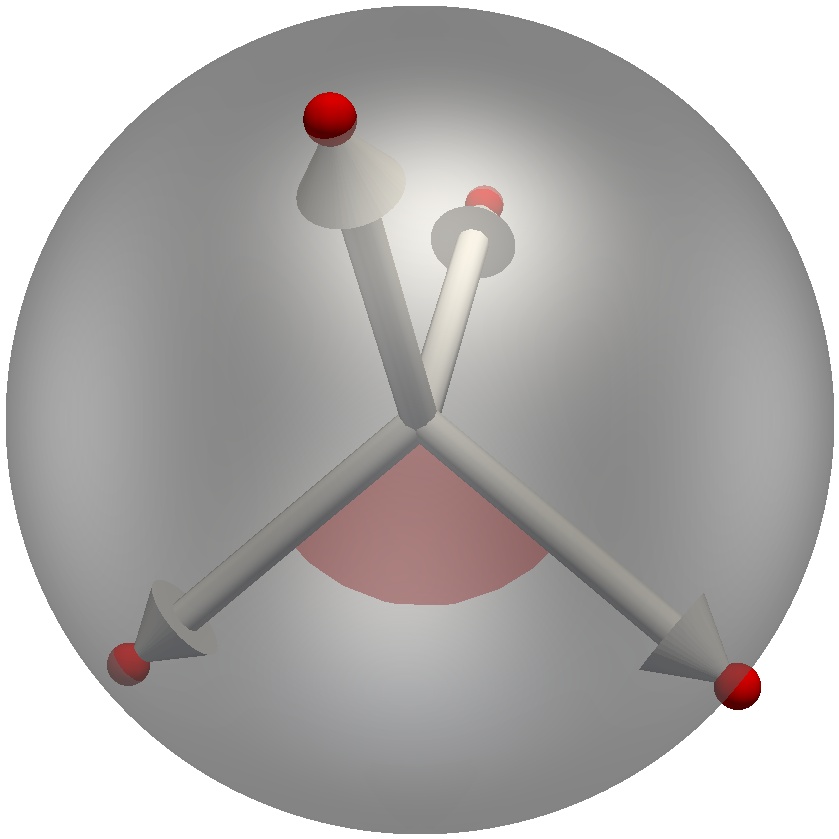}};
						\draw (-1.15,-1.0) node[anchor=east]   {\scriptsize $\boldsymbol{d}_i$};
						\draw (1.25,-1.1) node[anchor=west] {\scriptsize $\boldsymbol{d}_j$};
						\draw (0.0,-0.25) node[white,anchor=north] {\scriptsize $\gamma_{ij}$};
					\end{tikzpicture}
				\end{minipage}
			\end{minipage}
			\begin{minipage}{0.49\textwidth}
				\centering
				\input{pics/director_response/avgDefectAngle.tex}
			\end{minipage}
		\end{minipage}
	}
    \caption{Schematic drawing for defect angles $\gamma_{ij}$ between defect positions $\boldsymbol{d}_i$ and $\boldsymbol{d}_j$ (left) and averaged defect angle $\langle\gamma\rangle = \frac{1}{6} \sum_{i<j} \arccos(\langle\nicefrac{\boldsymbol{d}_i}{\|\boldsymbol{d}_i\|},\nicefrac{\boldsymbol{d}_j}{\|\boldsymbol{d}_j\|}\rangle)$ against time for the \landauDeGennesHelfrich\ model (right). $112^\circ$ is the equilibrium average for the considered ellipsoid at the initial and final configuration, $109.5^\circ$ and $120^\circ$ correspond to the tetrahedral and planar defect configuration on a sphere. The surface evolution passes the sphere configuration at $t=40$ and the final ellipsoidal shape is already reached at $t = 80$ (marked as vertical dashed lines).}
    \label{fig:director_response:angle}
\end{figure}

\subsection{Surface Response}

\begin{figure}[!h]
    \centering
    \ifthenelse{\boolean{useSnapshots}}
    {
		\ifthenelse{\boolean{forArxiv}}
		{
			\includegraphics[width=0.6\textwidth]{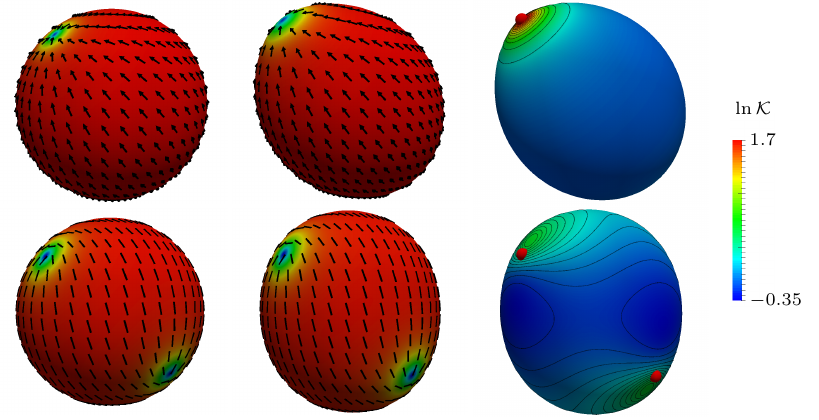}
		}
		{
			\includegraphics[width=0.8\textwidth]{snapshots/fig_3.png}
		}
	}
    {
		\ifthenelse{\boolean{forArxiv}}
		{
			\begin{minipage}{0.75\textwidth}
		}
		{
			\begin{minipage}{\textwidth}
		}
				\centering
				\begin{minipage}{0.75\textwidth}
					\centering
					\def\picwidth{0.32\textwidth}
					\begin{minipage}{\picwidth}
						\includegraphics[width=\textwidth]{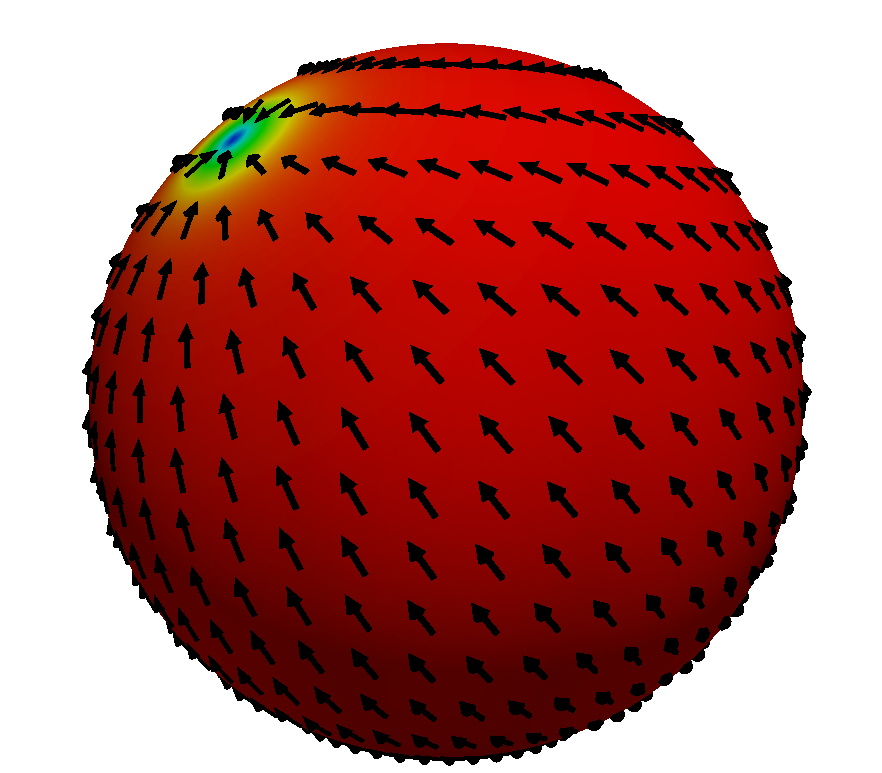}
					\end{minipage}
					\begin{minipage}{\picwidth}
						\includegraphics[width=\textwidth]{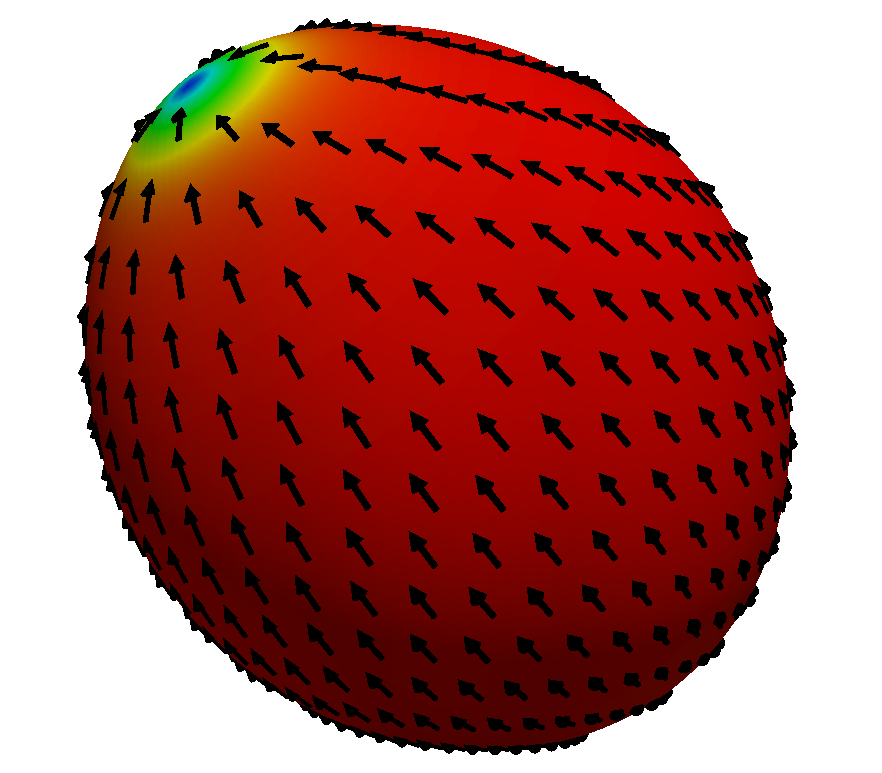}
					\end{minipage}
					\begin{minipage}{\picwidth}
						\includegraphics[width=\textwidth]{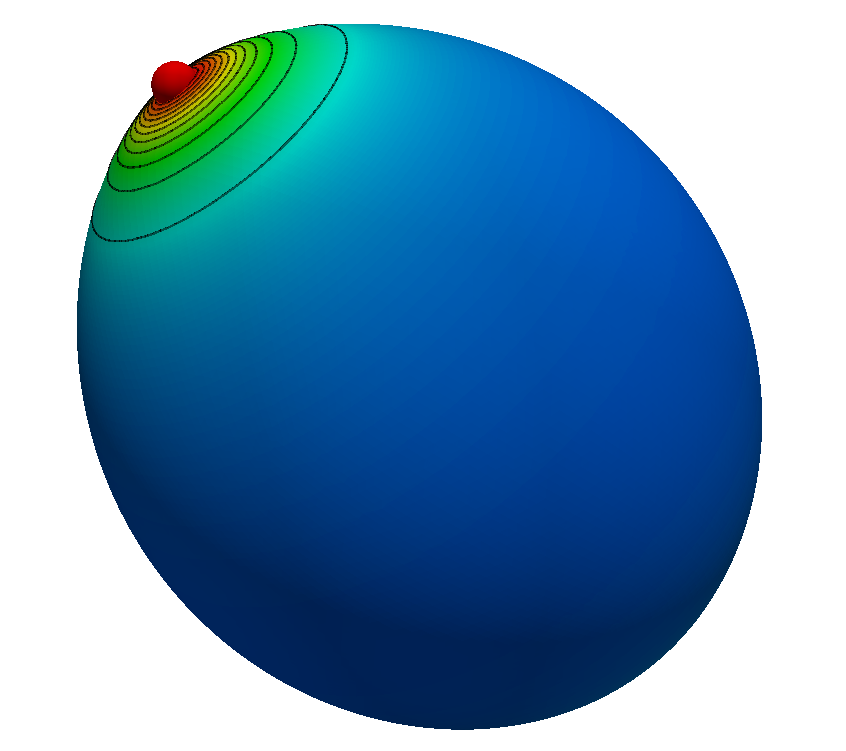}
					\end{minipage}
					\begin{minipage}{\picwidth}
						\includegraphics[width=\textwidth]{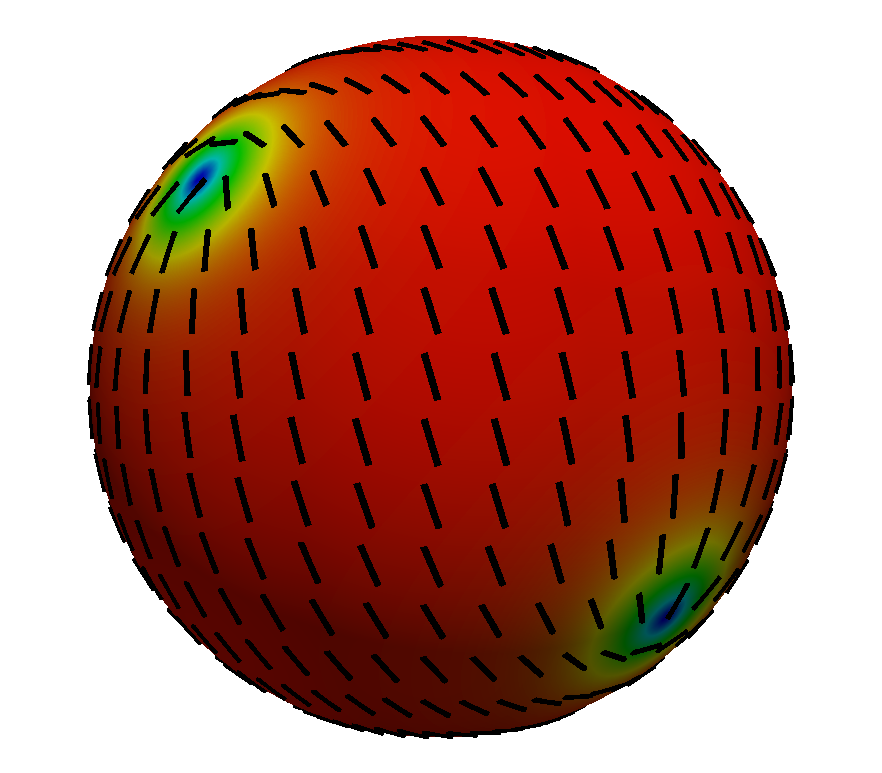}
					\end{minipage}
					\begin{minipage}{\picwidth}
						\includegraphics[width=\textwidth]{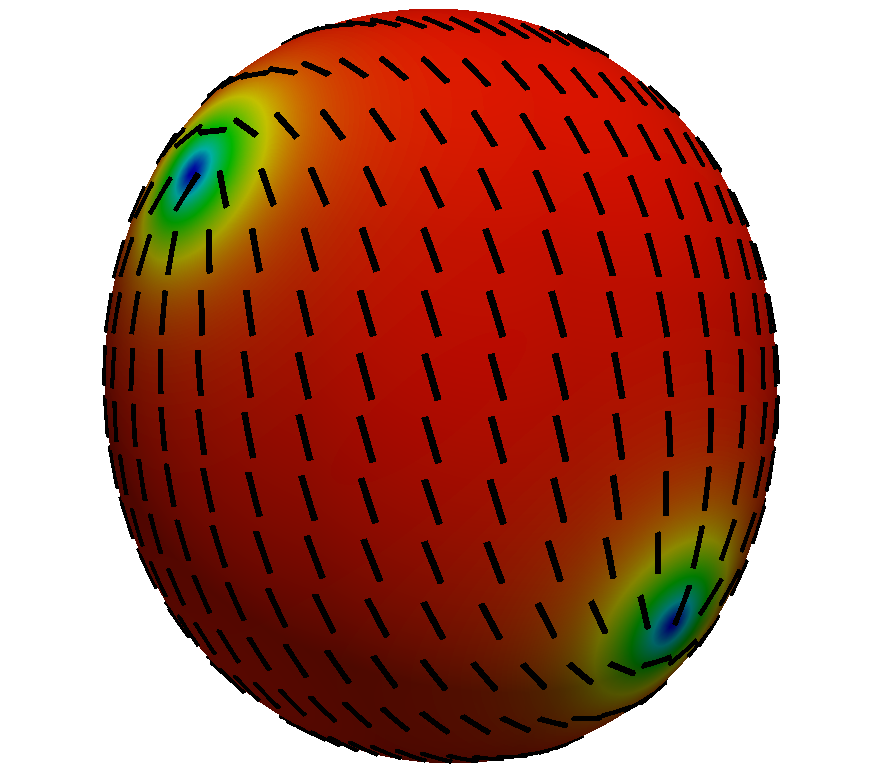}
					\end{minipage}
					\begin{minipage}{\picwidth}
						\includegraphics[width=\textwidth]{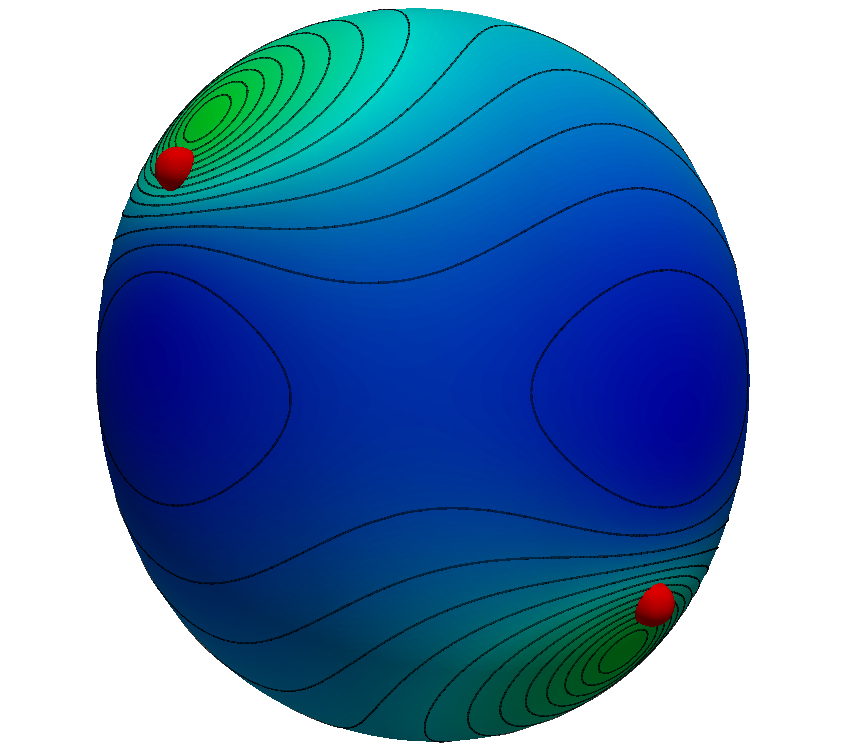}
					\end{minipage}
				\end{minipage}
				\begin{minipage}{0.11\textwidth}
					\insertColorbarVertical{$\operatorname{ln}\gaussc$}{$-0.35$}{}{}{$1.7$}
				\end{minipage}
			\end{minipage}
	}
	\caption{Top row: Initial condition for the \frankOseenHelfrich\ model as minimal configuration on the sphere (left) and reached steady state solution (center) with respective Gaussian curvature (right). Middle row: Initial condition for the \landauDeGennesHelfrich\ model as minimal configuration on the sphere (left) and reached steady state solution (center) with respective Gaussian curvature (right). The contour lines in the Gaussian curvature images indicate constant curvature lines and the red spheres are the positions of the defects.}
    \label{fig:surface_response:results}
\end{figure}
\begin{figure}[!h]
    \centering
    \ifthenelse{\boolean{useSnapshots}}
    {
		\ifthenelse{\boolean{forArxiv}}
		{
			\includegraphics[width=0.6375\textwidth]{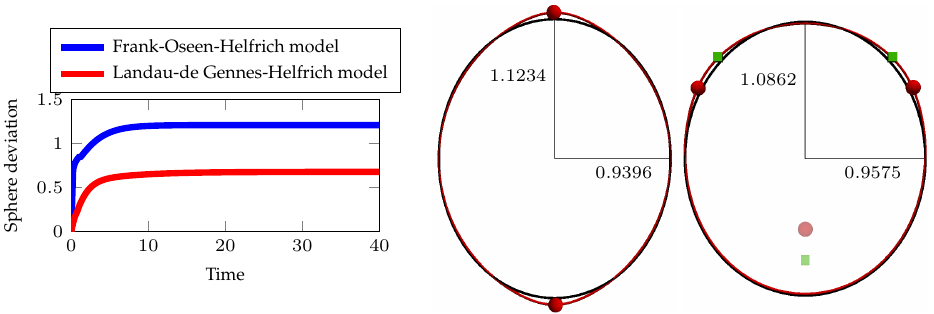}
		}
		{
			\includegraphics[width=0.85\textwidth]{snapshots/fig_4.png}
		}
    }
    {
		\ifthenelse{\boolean{forArxiv}}
		{
			\begin{minipage}{0.75\textwidth}
		}
		{
			\begin{minipage}{\textwidth}
		}
			\centering
			\begin{minipage}{0.45\textwidth}
				\centering
%
%
\begin{tikzpicture}

\begin{axis}[%
width=0.7\textwidth,
height=0.3\textwidth,
at={(0,0)},
scale only axis,
separate axis lines,
every outer x axis line/.append style={black},
every x tick label/.append style={font=\color{black}},
xmin=0,
xmax=40,
xtick={ 0, 10, 20, 30, 40},
xlabel={$\mbox{Time}$},
every outer y axis line/.append style={black},
every y tick label/.append style={font=\color{black}},
ymin=0,
ymax=1.5,
ytick={  0, 0.5,   1, 1.5},
ylabel={$\mbox{Sphere deviation}$},
axis background/.style={fill=white},
x label style={font=\scriptsize,at={(axis description cs:0.5,0.075)}},
y label style={font=\scriptsize,at={(axis description cs:0.1,0.5)}},
yticklabel style = {font=\scriptsize},
xticklabel style = {font=\scriptsize},
legend style={font=\scriptsize,at={(0.5,1.05)},anchor=south,legend cell align=left,align=left,draw=black}, 
]

\addlegendentry{$\mbox{\frankOseenHelfrich\ model}$};
\addlegendentry{$\mbox{\landauDeGennesHelfrich\ model}$};

\addplot [color=blue,solid,line width=2.5pt]
  table[row sep=crcr]{%
0	6.54815e-31\\
0.25	0.706673\\
0.5	0.782216\\
0.75	0.818016\\
1	0.846258\\
1.25	0.849727\\
1.5	0.876097\\
1.75	0.903043\\
2	0.927997\\
2.25	0.952694\\
2.5	0.975494\\
2.75	0.996796\\
3	1.01657\\
3.25	1.03474\\
3.5	1.05144\\
3.75	1.06669\\
4	1.08059\\
4.25	1.09323\\
4.5	1.1047\\
4.75	1.11509\\
5	1.12448\\
5.25	1.13297\\
5.5	1.14063\\
5.75	1.14754\\
6	1.15376\\
6.25	1.15936\\
6.5	1.1644\\
6.75	1.16893\\
7	1.173\\
7.25	1.17666\\
7.5	1.17994\\
7.75	1.18289\\
8	1.18553\\
8.25	1.18791\\
8.5	1.19004\\
8.75	1.19195\\
9	1.19366\\
9.25	1.1952\\
9.5	1.19658\\
9.75	1.19781\\
10	1.19892\\
10.25	1.19991\\
10.5	1.2008\\
10.75	1.20159\\
11	1.20231\\
11.25	1.20294\\
11.5	1.20352\\
11.75	1.20403\\
12	1.20449\\
12.25	1.2049\\
12.5	1.20527\\
12.75	1.2056\\
13	1.20589\\
13.25	1.20616\\
13.5	1.2064\\
13.75	1.20661\\
14	1.2068\\
14.25	1.20697\\
14.5	1.20712\\
14.75	1.20726\\
15	1.20738\\
15.25	1.20749\\
15.5	1.20759\\
15.75	1.20768\\
16	1.20776\\
16.25	1.20783\\
16.5	1.20789\\
16.75	1.20795\\
17	1.208\\
17.25	1.20804\\
17.5	1.20808\\
17.75	1.20812\\
18	1.20815\\
18.25	1.20818\\
18.5	1.20821\\
18.75	1.20823\\
19	1.20825\\
19.25	1.20827\\
19.5	1.20829\\
19.75	1.20831\\
20	1.20832\\
20.25	1.20833\\
20.5	1.20834\\
20.75	1.20835\\
21	1.20836\\
21.25	1.20837\\
21.5	1.20838\\
21.75	1.20838\\
22	1.20839\\
22.25	1.20839\\
22.5	1.2084\\
22.75	1.2084\\
23	1.20841\\
23.25	1.20841\\
23.5	1.20841\\
23.75	1.20842\\
24	1.20842\\
24.25	1.20842\\
24.5	1.20842\\
24.75	1.20843\\
25	1.20843\\
25.25	1.20843\\
25.5	1.20843\\
25.75	1.20843\\
26	1.20843\\
26.25	1.20843\\
26.5	1.20844\\
26.75	1.20844\\
27	1.20844\\
27.25	1.20844\\
27.5	1.20844\\
27.75	1.20844\\
28	1.20844\\
28.25	1.20844\\
28.5	1.20844\\
28.75	1.20844\\
29	1.20844\\
29.25	1.20844\\
29.5	1.20844\\
29.75	1.20844\\
30	1.20844\\
30.25	1.20845\\
30.5	1.20845\\
30.75	1.20845\\
31	1.20845\\
31.25	1.20845\\
31.5	1.20845\\
31.75	1.20845\\
32	1.20845\\
32.25	1.20845\\
32.5	1.20845\\
32.75	1.20845\\
33	1.20845\\
33.25	1.20845\\
33.5	1.20845\\
33.75	1.20845\\
34	1.20845\\
34.25	1.20845\\
34.5	1.20845\\
34.75	1.20845\\
35	1.20845\\
35.25	1.20845\\
35.5	1.20845\\
35.75	1.20845\\
36	1.20845\\
36.25	1.20845\\
36.5	1.20845\\
36.75	1.20845\\
37	1.20845\\
37.25	1.20845\\
37.5	1.20845\\
37.75	1.20845\\
38	1.20845\\
38.25	1.20846\\
38.5	1.20846\\
38.75	1.20846\\
39	1.20846\\
39.25	1.20846\\
39.5	1.20846\\
39.75	1.20846\\
40	1.20846\\
};

\addplot [color=red,solid,line width=2.5pt]
  table[row sep=crcr]{%
0	6.66784e-31\\
0.25	0.100742\\
0.5	0.176264\\
0.75	0.221647\\
1	0.281838\\
1.25	0.334013\\
1.5	0.377066\\
1.75	0.415117\\
2	0.447866\\
2.25	0.475765\\
2.5	0.499416\\
2.75	0.519418\\
3	0.536309\\
3.25	0.550576\\
3.5	0.56265\\
3.75	0.572905\\
4	0.581655\\
4.25	0.589162\\
4.5	0.595647\\
4.75	0.601288\\
5	0.606234\\
5.25	0.610606\\
5.5	0.614502\\
5.75	0.618003\\
6	0.621172\\
6.25	0.624063\\
6.5	0.626719\\
6.75	0.629174\\
7	0.631456\\
7.25	0.633588\\
7.5	0.635589\\
7.75	0.637474\\
8	0.639256\\
8.25	0.640945\\
8.5	0.64255\\
8.75	0.644077\\
9	0.645533\\
9.25	0.646922\\
9.5	0.64825\\
9.75	0.64952\\
10	0.650735\\
10.25	0.651899\\
10.5	0.653013\\
10.75	0.65408\\
11	0.655103\\
11.25	0.656082\\
11.5	0.657021\\
11.75	0.657921\\
12	0.658784\\
12.25	0.65961\\
12.5	0.660401\\
12.75	0.66116\\
13	0.661886\\
13.25	0.662582\\
13.5	0.663249\\
13.75	0.663887\\
14	0.664498\\
14.25	0.665083\\
14.5	0.665643\\
14.75	0.666179\\
15	0.666692\\
15.25	0.667183\\
15.5	0.667653\\
15.75	0.668102\\
16	0.668532\\
16.25	0.668943\\
16.5	0.669336\\
16.75	0.669712\\
17	0.670071\\
17.25	0.670415\\
17.5	0.670744\\
17.75	0.671058\\
18	0.671358\\
18.25	0.671645\\
18.5	0.67192\\
18.75	0.672182\\
19	0.672432\\
19.25	0.672672\\
19.5	0.672901\\
19.75	0.673119\\
20	0.673328\\
20.25	0.673528\\
20.5	0.673719\\
20.75	0.673901\\
21	0.674075\\
21.25	0.674241\\
21.5	0.6744\\
21.75	0.674552\\
22	0.674697\\
22.25	0.674835\\
22.5	0.674968\\
22.75	0.675094\\
23	0.675215\\
23.25	0.67533\\
23.5	0.67544\\
23.75	0.675545\\
24	0.675645\\
24.25	0.675741\\
24.5	0.675833\\
24.75	0.675921\\
25	0.676004\\
25.25	0.676084\\
25.5	0.67616\\
25.75	0.676233\\
26	0.676303\\
26.25	0.676369\\
26.5	0.676433\\
26.75	0.676493\\
27	0.676551\\
27.25	0.676606\\
27.5	0.676659\\
27.75	0.67671\\
28	0.676758\\
28.25	0.676804\\
28.5	0.676848\\
28.75	0.67689\\
29	0.67693\\
29.25	0.676969\\
29.5	0.677006\\
29.75	0.677041\\
30	0.677074\\
30.25	0.677106\\
30.5	0.677137\\
30.75	0.677166\\
31	0.677194\\
31.25	0.677221\\
31.5	0.677246\\
31.75	0.677271\\
32	0.677294\\
32.25	0.677317\\
32.5	0.677338\\
32.75	0.677358\\
33	0.677378\\
33.25	0.677397\\
33.5	0.677415\\
33.75	0.677432\\
34	0.677448\\
34.25	0.677464\\
34.5	0.677479\\
34.75	0.677493\\
35	0.677507\\
35.25	0.67752\\
35.5	0.677533\\
35.75	0.677545\\
36	0.677557\\
36.25	0.677568\\
36.5	0.677578\\
36.75	0.677589\\
37	0.677598\\
37.25	0.677608\\
37.5	0.677617\\
37.75	0.677625\\
38	0.677634\\
38.25	0.677642\\
38.5	0.677649\\
38.75	0.677657\\
39	0.677664\\
39.25	0.67767\\
39.5	0.677677\\
39.75	0.677683\\
40	0.677689\\
};

\end{axis}
\end{tikzpicture}%
			\end{minipage}
			\begin{minipage}{0.25\textwidth}
				\centering
				\begin{tikzpicture}
					\node (pic) at (0,0) {\includegraphics[width=\textwidth]{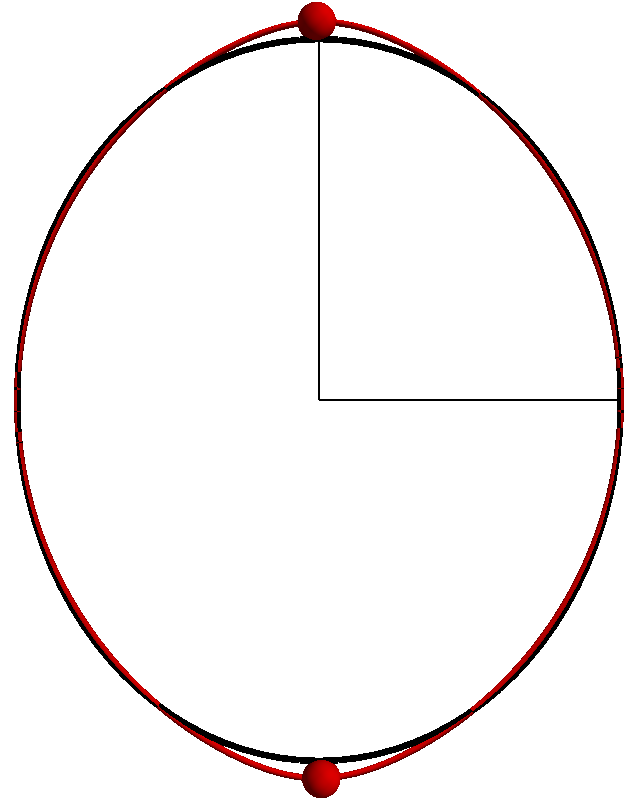}};
					\draw (0.95625,0.0) node[anchor=north] {\scriptsize $0.9396$};
					\draw (0.0,1.14375) node[anchor=east] {\scriptsize $1.1234$};
				\end{tikzpicture}
			\end{minipage}
			\begin{minipage}{0.25\textwidth}
				\centering
				\begin{tikzpicture}
					\node (pic) at (0,0) {\includegraphics[width=\textwidth]{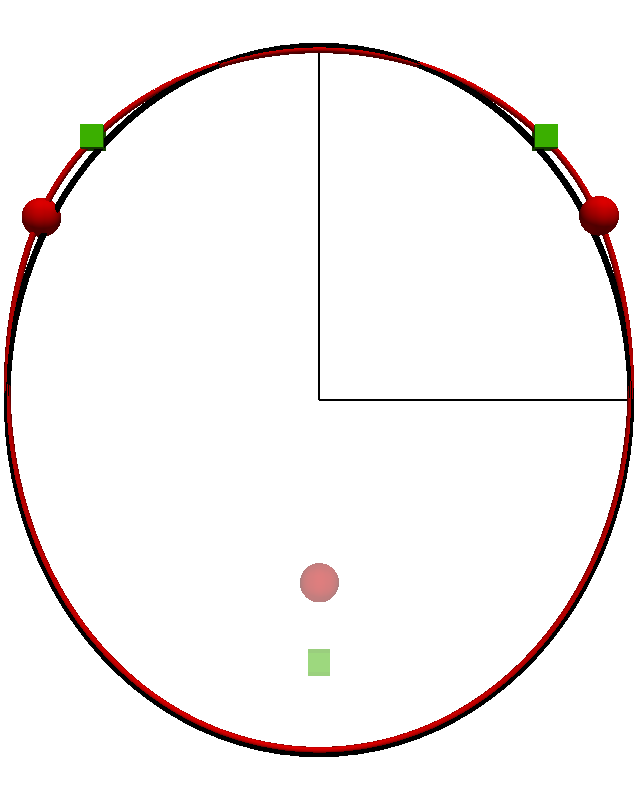}};
					\draw (0.9375,0.0) node[anchor=north] {\scriptsize $0.9575$};
					\draw (0.0,1.0875) node[anchor=east] {\scriptsize $1.0862$};
				\end{tikzpicture}
			\end{minipage}
		\end{minipage}
	}
    \caption{Left: Deviation from a sphere geometry during time evolution. The sphere deviation $\sigma_{\mathcal{S}^2}$ is considered to be $\sigma_{S^2} = \int_\surf\left(\meanc-\meanc_{S^2}\right)^2\mu$, where $\meanc_{S^2}$ denotes the mean curvature of the unit sphere. Right: Slice of final shape of the \frankOseenHelfrich\ and \landauDeGennesHelfrich\ model (red), reference is an ellipse (black). For the nematic case the slice is through the upper defect pair. Red spheres are the defect positions (light red sphere indicates the defect pair that does not belong to the slice), green squares are the positions of maximum Gaussian curvature (light green square for corresponding position which does not belong to the slice). Ellipses are fitted such that the surface area corresponds to the surface area of the final shape.}
    \label{fig:surface_response:schematic}
\end{figure}

Next, we consider stationary director fields and Q tensor fields and investigate the response of the surface for the respective fields. 
In the models \eqref{eq:gradientflow:p:vnor}-\eqref{eq:gradientflow:p:dirf} and \eqref{eq:gradientflow:q:vnor}-\eqref{eq:gradientflow:q:qten} this can be achieved by using $\kinConstP/\kinConst\to\infty$ and $\kinConstQ/\kinConst\to\infty$, respectively. 
Roughly speaking, the time scale of the response of the surface is infinitesimal small compared to the relaxation time scales of the director field/Q tensor field.
We start with the minimal energy configurations on a sphere for the director field $\dirf$ and the Q tensor field $\qten$, which were obtained as steady state solutions of eqs. \eqref{eq:gradientflow:p:dirf} and \eqref{eq:gradientflow:q:qten} on the unit sphere with the initial conditions from the prior section. 
\autoref{fig:surface_response:results} shows these initial conditions and the reached steady state solutions after shape relaxation. 
In both cases the shape deviates. 
The deviation from a sphere is computed according to $\sigma_{S^2} = \int_\surf\left(\meanc-\meanc_{S^2}\right)^2\mu$, where $\meanc_{S^2}$ denotes the mean curvature of the unit sphere.
The deviation from a sphere for the \frankOseenHelfrich\ model is approximately two times larger as for the \landauDeGennesHelfrich\ model, see \autoref{fig:surface_response:schematic}. 
This results from the larger number of defect, four instead of two, but also the Helfrich-like contribution, \cf\ \eqref{eq:gradientflow:q:vnor} resulting from the variation of the \landauDeGennes\ energy \wrt\ to the surface, which effectively increases bending stiffness and therefore keeps the surface more spherical. 
However, in both cases the surface is most distorted in the vicinity of the defects, where regions of high Gaussian curvature evolve, see \autoref{fig:surface_response:results}. 
A logarithmic scale for Gaussian curvature is used to highlight the differences between both cases. For polar order the distortion results in a symmetric ellipsoidal-like shape, with the two $+1$ defects located at the high curvature points on the long axis. 
For nematic order the final shape becomes asymmetric. 
However, in \autoref{fig:surface_response:schematic} we approximate both shapes by a symmetric ellipsoid with the same surface area. 
The distortions in the vicinity of the defects are visible and the different values for long- and short-axis of the approximating ellipsoids are marked. 
\autoref{fig:surface_response:schematic} further highlights the differences of defect positions and points of maximal Gaussian curvature for the nematic case. 
They are still close to each other but do not coincide. 
Within the vicinity of the poles on the long axis of the approximating ellipsoids defects and points of maximal curvature are symmetrically arranged on one line, which is rotated by $90^\circ$ for the two poles. 
Additional forces which contribute to shape relaxation are due to alignment. 
In regions of low Gaussian curvature the director fields are almost perfectly aligned. 

\subsection{Fully Coupled System}

After investigating either the directors response or the surfaces response we now combine both mechanisms by considering the whole system of equations \eqref{eq:gradientflow:p:vnor}-\eqref{eq:gradientflow:p:dirf} and \eqref{eq:gradientflow:q:vnor}-\eqref{eq:gradientflow:q:qten} in an appropriate parameter setting, \cf\ \autoref{tab:parameters}. 
As initial conditions we consider shapes for which director fields / Q tensor fields with more than the minimal number of topologically required defects are energetically favourable. 
For the \frankOseenHelfrich\ energy we make use of a so-called nonic surface, which was considered in \cite{Nestleretal_JNS_2018}.
The parametrization of this surface reads
\begin{align}
    \label{eq:nonic:parametrization}
    \para(\theta, \varphi) &:= \para_{S^2}(\theta, \varphi) + f_{C,r}(\cos\theta)\mathbf{e}_x - B\sin\theta\sin\varphi\mathbf{e}_y
\end{align}
with standard parametrization angles $\theta$, $\varphi$, standard parametrization of the unit sphere $\para_{S^2}(\theta, \varphi)$, unit vectors in $x$-,$y$-direction $\mathbf{e}_x$,$\mathbf{e}_y$ in $\R^3$, parameters $C=0.75$, $B=\nicefrac{7}{20}C$ and $r=0.95$ and 
\begin{align*}
    f_{C,r}(z) &:= \frac{1}{4}Cz^2\left(\left(z+1\right)^2\left(4-3z\right) + r\left(z-1\right)^2\left(4+3z\right)\right) \formPeriod
\end{align*}
According to \cite{Nestleretal_JNS_2018} there exist two stable defect configurations, one with two $+1$ defects and one with four defects, three with $+1$ and one with $-1$ topological charge, located at points with positive and negative Gaussian curvature extreme value, respectively. 
The latter is used as initial condition in the present setting and can be obtained as steady state solution on the stationary nonic surface with
\begin{align*}
    \dirf|_{t=0} &= \frac{1}{\|\proj_\surf\mathbf{e}_x\|}\proj_\surf\mathbf{e}_x
\end{align*}
as initial condition. 
While being the ground state on the initial surface during evolution regions with high curvature become shallower such that the four defect solution is no longer favourable and two defects (a $\pm1$ defect pair) annihilate shortly after $t=1$. 
The following evolution is again towards a symmetric ellipsoidal-like shape with the two remaining $+1$ defects located at the high Gaussian curvature points at the long axis. 
The results are shown in \autoref{fig:results:pField}. 
The energy decay and dissipation rate for the evolution is shown in \autoref{fig:results:pField} (bottom line), with the drop in the energy curves and the peak in the dissipation rate corresponding to the defect annihilation.

\begin{figure}[!h]
    \centering
    \ifthenelse{\boolean{useSnapshots}}
    {
		\ifthenelse{\boolean{forArxiv}}
		{
			\includegraphics[width=0.675\textwidth]{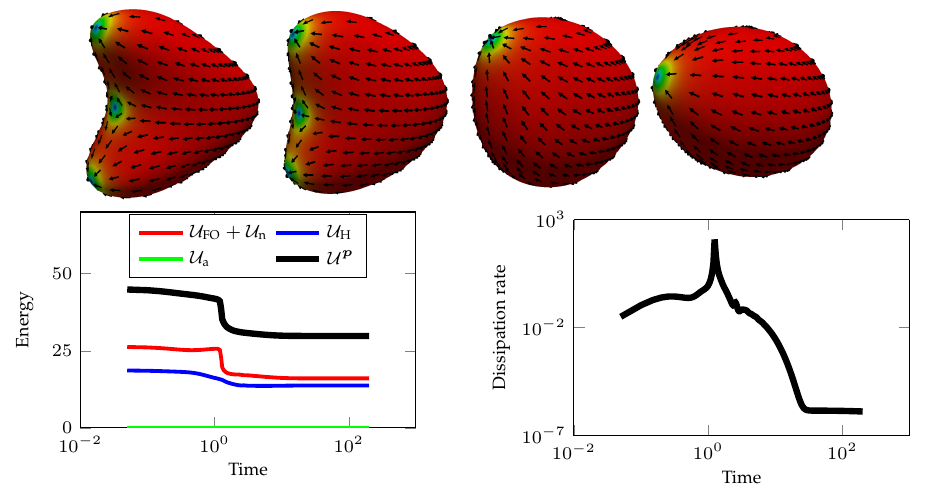}
		}
		{
			\includegraphics[width=0.9\textwidth]{snapshots/fig_5.png}
		}
    }
    {
		\ifthenelse{\boolean{forArxiv}}
		{
			\begin{minipage}{0.75\textwidth}
		}
		{
			\begin{minipage}{\textwidth}
		}
				\centering
				\begin{minipage}{\textwidth}
					\centering
					\def\picwidth{0.19\textwidth}
					\includegraphics[width=\picwidth]{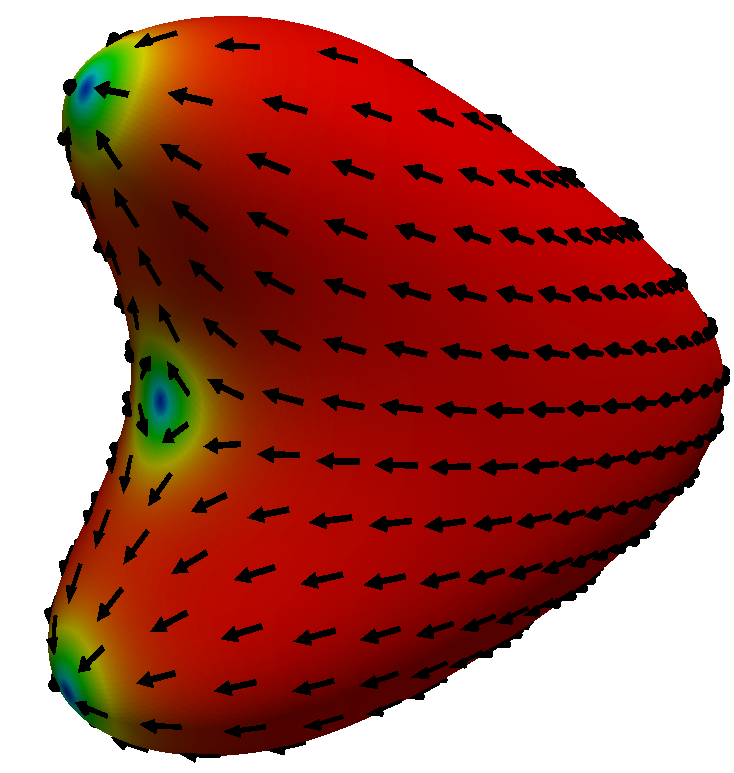}
					\includegraphics[width=\picwidth]{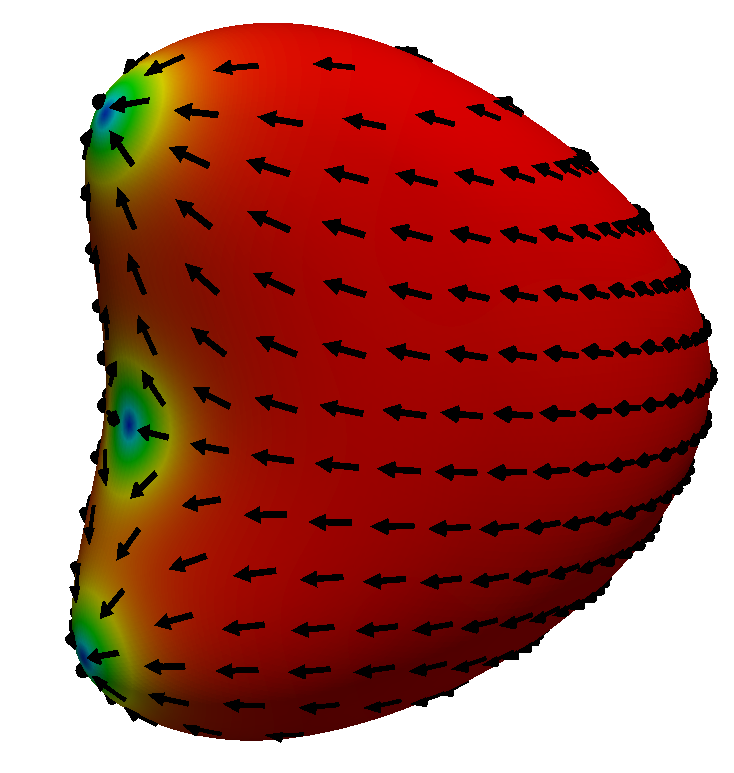}
					\includegraphics[width=\picwidth]{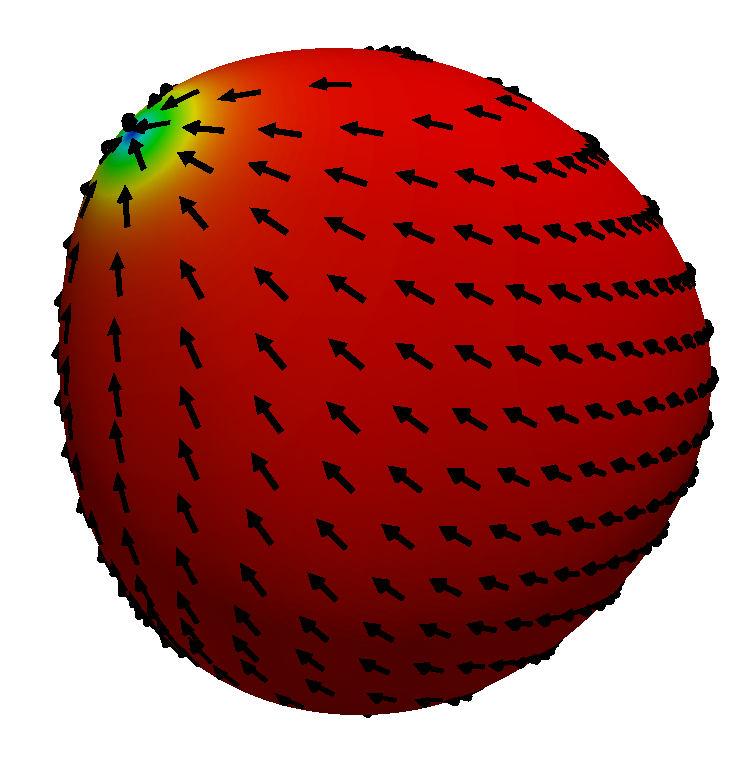}
					\includegraphics[width=\picwidth]{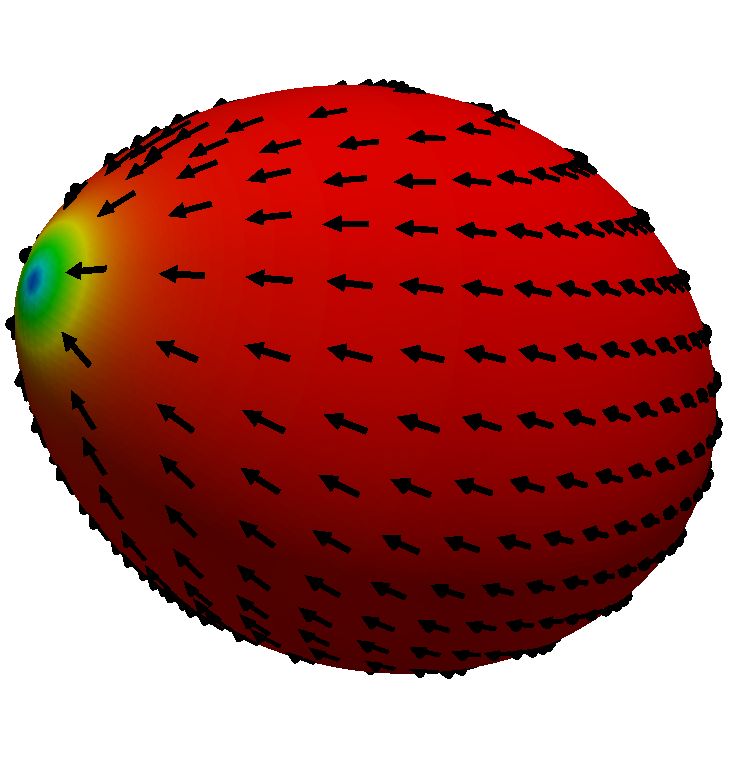}
				\end{minipage}
				\begin{minipage}{0.49\textwidth}
					\centering
					\input{pics/full/pField/energy.tex}
				\end{minipage}
				\begin{minipage}{0.49\textwidth}
					\centering
					\input{pics/full/pField/dissipationRate.tex}
				\end{minipage}
			\end{minipage}
    }
    \caption{Top: Evolution of the $\dirf$ field of the \frankOseenHelfrich\ model for $t=0, 1, 2.5, 100$ for the nonic surface initial condition (left to right). Bottom: Energies against time (left) and dissipation rate against time (right). See also Supplementary Video 1.}
    \label{fig:results:pField}
\end{figure}

\begin{figure}[!h]
    \centering
    \ifthenelse{\boolean{useSnapshots}}
    {
		\ifthenelse{\boolean{forArxiv}}
		{
			\includegraphics[width=0.675\textwidth]{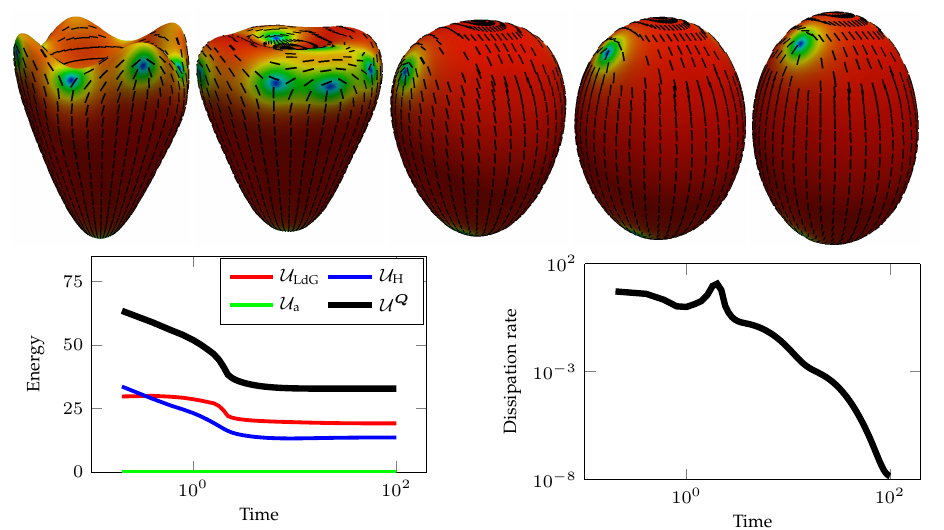}
		}
		{
			\includegraphics[width=0.9\textwidth]{snapshots/fig_6.png}
		}
    }
    {
		\ifthenelse{\boolean{forArxiv}}
		{
			\begin{minipage}{0.75\textwidth}
		}
		{
			\begin{minipage}{\textwidth}
		}
				\centering
				\begin{minipage}{\textwidth}
					\centering
					\def\picheight{0.24\textwidth}
					\includegraphics[height=\picheight]{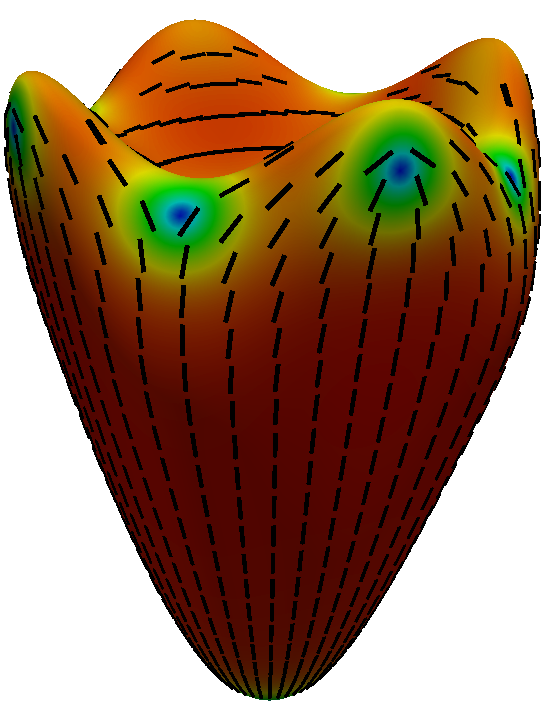}
					\includegraphics[height=\picheight]{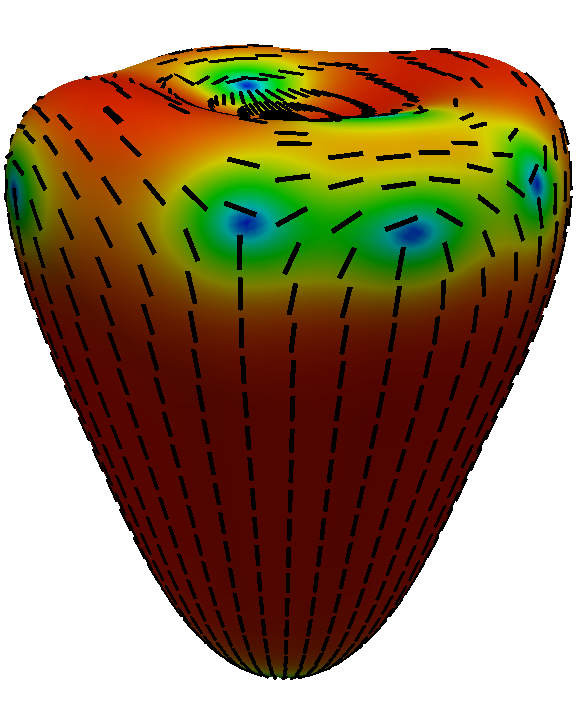}
					\includegraphics[height=\picheight]{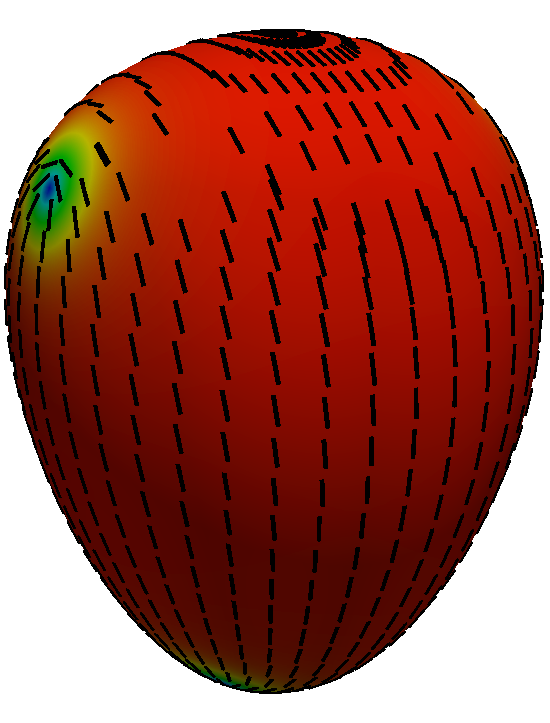}
					\includegraphics[height=\picheight]{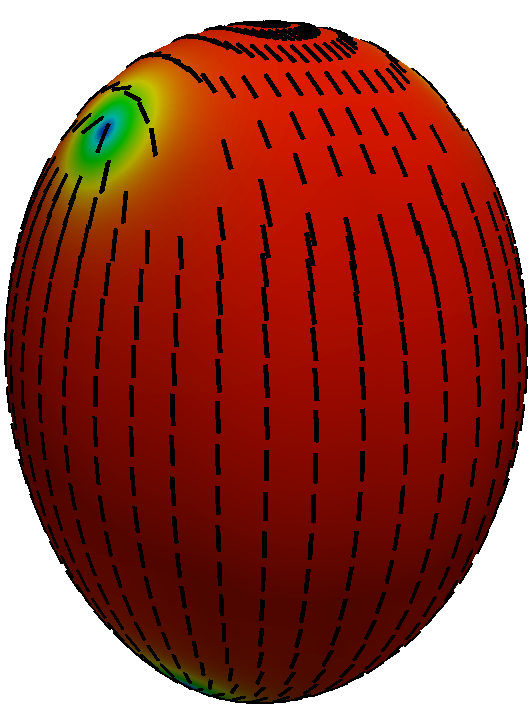}
					\includegraphics[height=\picheight]{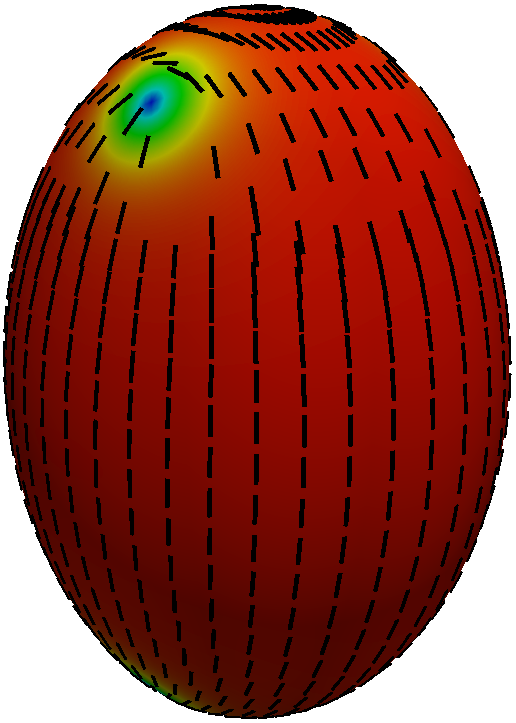}
				\end{minipage}
				\begin{minipage}{0.49\textwidth}
					\centering
					\input{pics/full/qTensor/energy.tex}
				\end{minipage}
				\begin{minipage}{0.49\textwidth}
					\centering
					\input{pics/full/qTensor/dissipationRate.tex}
				\end{minipage}
			\end{minipage}
    }
	\caption{Top: Evolution of the q tensor field of the \landauDeGennesHelfrich\ model for $t=0, 1.5, 5, 10, 100$ for the bi-nonic surface initial condition (left to right). Bottom: Energies against time (left) and dissipation rate against time (right). See also Supplementary Video 2.}
    \label{fig:results:qTensor}
\end{figure}

To construct a similar situation for the \landauDeGennesHelfrich\ model is more complex. 
A morphology has to be found for which more than four $+\nicefrac{1}{2}$ defects are energetically favorable. 
We modify the parametrization by using
\begin{align}
    \label{eq:binonic:parametrization}
    \para(\theta, \varphi) &:= \para_{S^2}(\theta, \varphi) + f_{C,r}(\cos\theta)\mathbf{e}_x + f_{C,r}(\sin\theta\sin\varphi)\mathbf{e}_x
\end{align}
and call this surface bi-nonic. 
All parameters remain unchanged except of $C=1.1$. 
As initial condition we use the steady state solution on the stationary bi-nonic surface with
\begin{align*}
    \qten|_{t=0} &= \orderp\left(\left(\normal\times\dirf_{\qten}\right)\otimes\left(\normal\times\dirf_{\qten}\right) - \frac{1}{2}\proj_\surf\right) \formComma \quad
    \dirf_{\qten} = \frac{1}{\|\proj_\surf\mathbf{e}_x\|}\proj_\surf\mathbf{e}_x \formPeriod
\end{align*}
which consists of eight $+\nicefrac{1}{2}$ defects and four $-\nicefrac{1}{2}$ defects. 
The positive ones are located in the vicinity of peaks and valleys, while the negative ones can be found around saddle points of the surface. 
As for the setting for the \frankOseenHelfrich\ energy in \cite{Nestleretal_JNS_2018} it is demonstrated that this solution is stable and has a lower energy than any configuration with less defects. 
The geometry, the initial condition and the evolution of both are shown in \autoref{fig:results:qTensor}. 
The results show a similar behavior as for the \frankOseenHelfrich\ model. 
Regions of high curvature become shallower over time such that geometric forces -- pushing or pulling defects -- become weaker. 
This results in annihilation of defects (four $\pm\nicefrac{1}{2}$ defect pairs), which can again be observed in the evolution of the energy as well as dissipation rate as step or peak, respectively, \cf\ \autoref{fig:results:qTensor} (bottom row).
The step and peak are not as sharp as for the \frankOseenHelfrich\ model, which results from the different topological charge of the annihilated defects and the large energy of defects associated with normalization energy $\normalizationEnergy$.

Most interestingly is the evolution after the annihilation, if only four $+\nicefrac{1}{2}$ defects remain. 
In contrast with the tetrahedral ground state on a sphere or a rotationally symmetric ellipsoid, here the four $+\nicefrac{1}{2}$ defects arrange in a planar position on a surface which is asymmetric. 
In \autoref{fig:results:schematic} we compare the deviations from a sphere for both the \frankOseenHelfrich\ and the \landauDeGennesHelfrich\ model during evolution. 
The increase for the \landauDeGennesHelfrich\ model corresponds with the defect rearrangement and the formation of an asymmetric shape. 
In the polar case, we approximate the final shape by an ellipsoid with the same surface area to highlight the distortions in the vicinity of the defects. 
In the nematic case we mark only the distances from the origin to the surface to highlight the asymmetry of the final configuration.
Defects are visualized as red spheres and the maximal value of the Gaussian curvature are marked as green squares.  

\begin{figure}[!h]
    \centering
    \ifthenelse{\boolean{useSnapshots}}
    {
		\ifthenelse{\boolean{forArxiv}}
		{
			\includegraphics[width=0.6\textwidth]{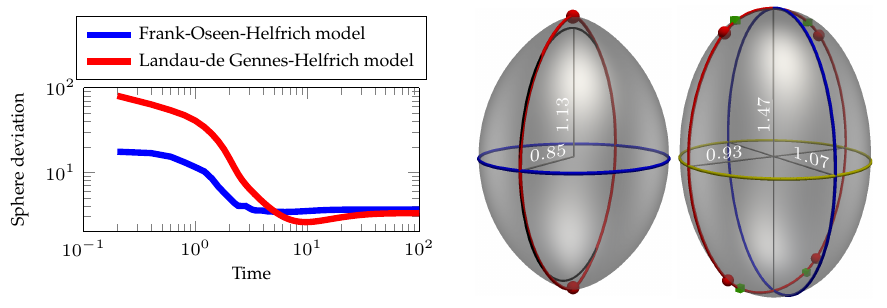}
		}
		{
			\includegraphics[width=0.8\textwidth]{snapshots/fig_7.png}
		}
    }
    {
		\ifthenelse{\boolean{forArxiv}}
		{
			\begin{minipage}{0.75\textwidth}
		}
		{
			\begin{minipage}{\textwidth}
		}
				\centering
				\begin{minipage}{0.49\textwidth}
				\centering
					\input{pics/full/sphereDeviation.tex}
				\end{minipage}
				\begin{minipage}{0.2\textwidth}
					\centering
					\begin{tikzpicture}
						\node (pic) at (0,0) {\includegraphics[width=\textwidth]{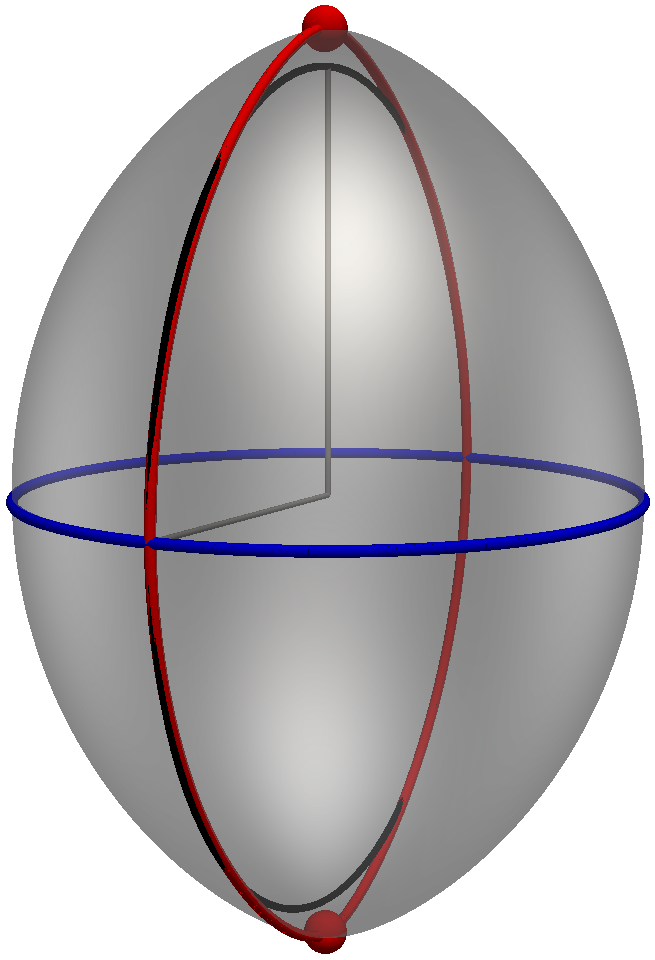}};
						\draw (-0.35,-0.21) node[white,anchor=south,rotate=0,xslant=0.0,yslant=0.2]   {\scriptsize $0.85$};
						\draw (0.04,0.5) node[white,anchor=south,rotate=90]   {\scriptsize $1.13$};
					\end{tikzpicture}
				\end{minipage}
				\begin{minipage}{0.2\textwidth}
					\centering
					\begin{tikzpicture}
						\node (pic) at (0,0) {\includegraphics[width=\textwidth]{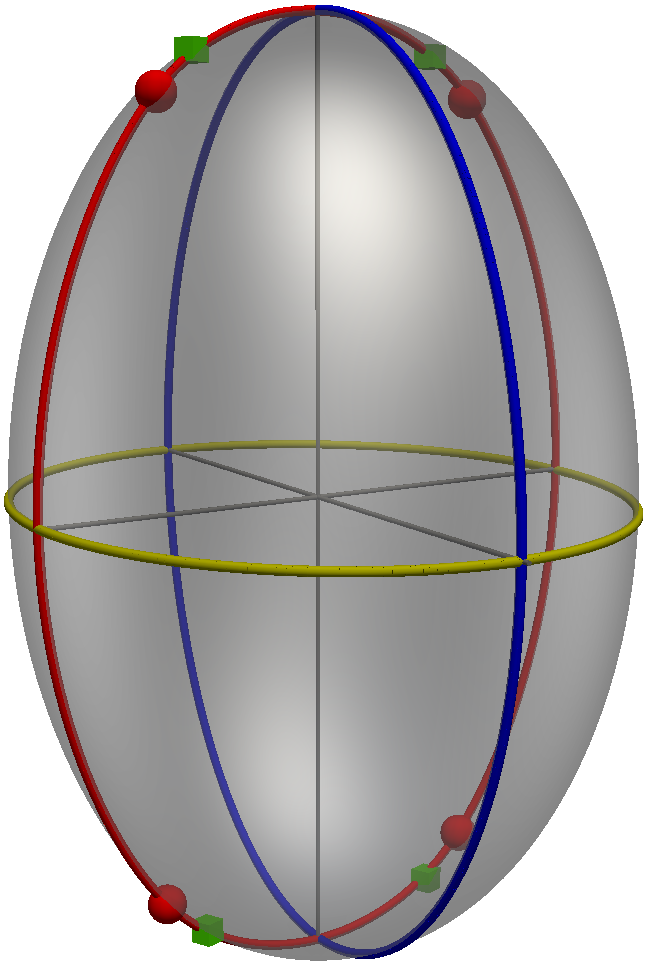}};
						\draw (-0.7,-0.21) node[white,anchor=south,rotate=0,xslant=0.0,yslant=0.12]   {\scriptsize $0.93$};
						\draw (0.04,0.5) node[white,anchor=south,rotate=90]   {\scriptsize $1.47$};
						\draw (0.5,-0.28) node[white,anchor=south,xslant=0.0,yslant=-0.35] {\scriptsize $1.07$};
					\end{tikzpicture}
				\end{minipage}
			\end{minipage}
	}
    \caption{Left: Deviation from a sphere geometry during time evolution. The sphere deviation $\sigma_{\mathcal{S}^2}$ is considered to be $\sigma_{S^2} = \int_\surf\left(\meanc-\meanc_{S^2}\right)^2\mu$, where $\meanc_{S^2}$ denotes the mean curvature of the unit sphere. Right: Final shape of the \frankOseenHelfrich\ and \landauDeGennesHelfrich\ model. The values for the axes (distances from the origin) indicate the symmetry (asymmetry) for the polar (nematic) case. Red spheres are the defect positions, green squares are the positions of maximum Gaussian curvature. In the polar case both coincide, therefore only the red spheres are shown. The values for the long and short axis correspond to the fitted symmetric ellipsoid (black line) in the polar case. The distortion at the defect positions is clearly visible. In the nematic case the values correspond to the asymmetric form of the final shape. The defects and points of maximal curvature are located on the same plane (red line).}
    \label{fig:results:schematic}
\end{figure}

The fully coupled system does not only annihilate all topologically unnecessary defects by smoothing the surface in the case of nematic order it also leads to asymmetric shapes and defect arrangements which differ from ground state configurations on rotationally symmetric ellipsoids. 

\subsection{Discussion}

\begin{figure}[!h]
    \centering
    \ifthenelse{\boolean{useSnapshots}}
    {
		\ifthenelse{\boolean{forArxiv}}
		{
			\includegraphics[width=0.675\textwidth]{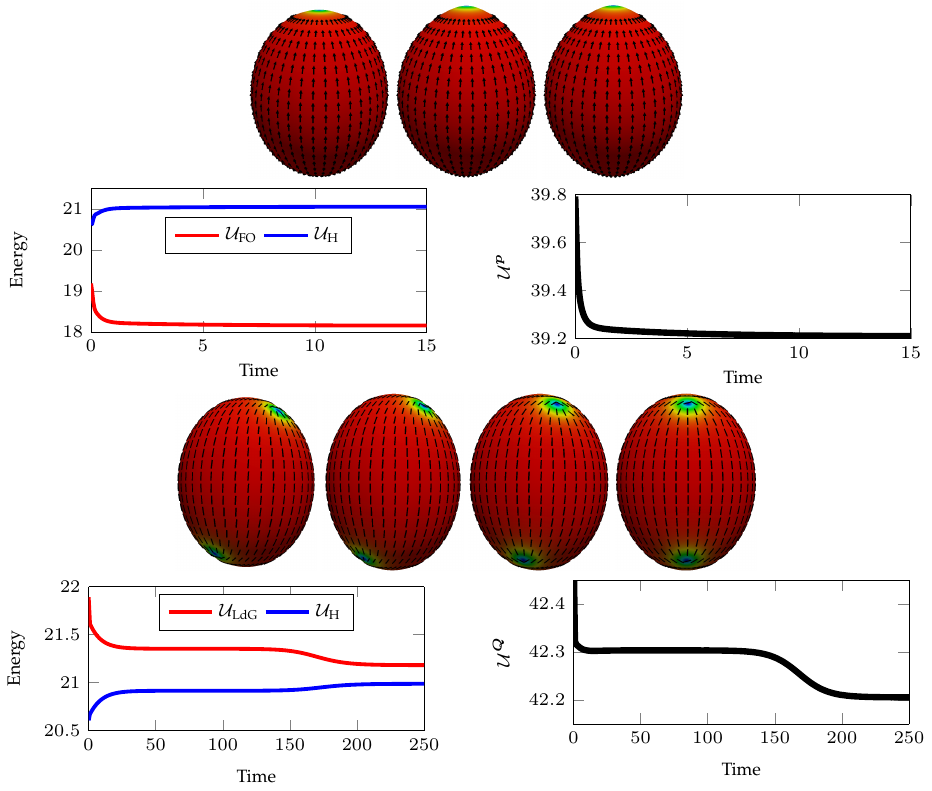}
		}
		{
			\includegraphics[width=0.9\textwidth]{snapshots/fig_8.png}
		}
    }
    {
		\ifthenelse{\boolean{forArxiv}}
		{
			\begin{minipage}{0.75\textwidth}
		}
		{
			\begin{minipage}{\textwidth}
		}
				\centering
				\begin{minipage}{0.6\textwidth}
					\centering
					\def\picwidth{0.24\textwidth}
					\includegraphics[width=\picwidth]{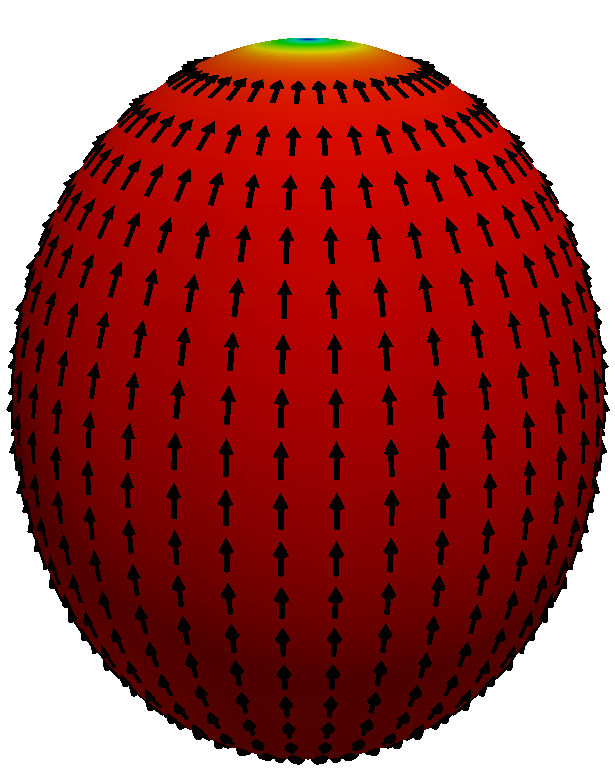}
					\includegraphics[width=\picwidth]{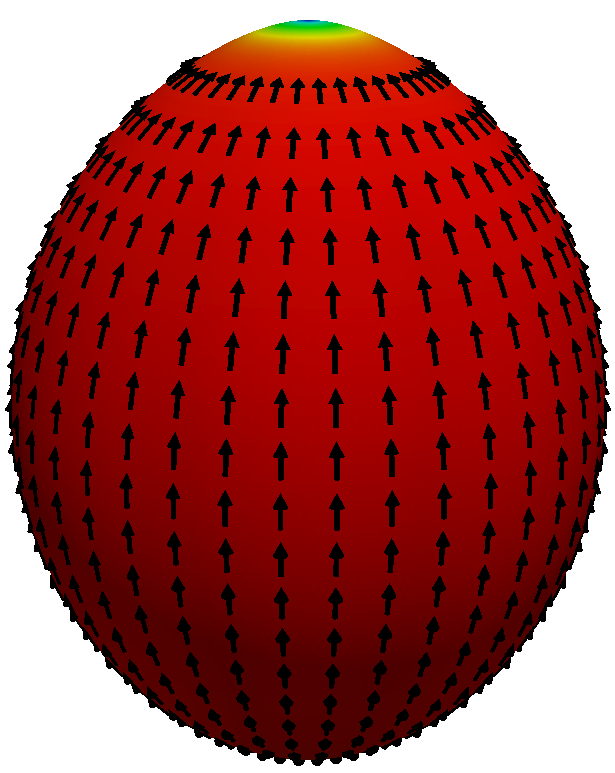}
					\includegraphics[width=\picwidth]{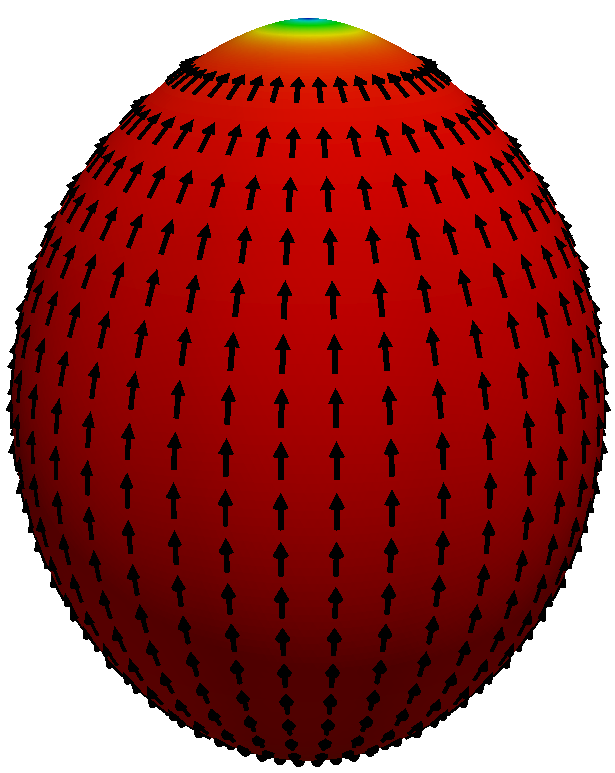}
				\end{minipage}
				\begin{minipage}{0.49\textwidth}
					\centering
					\input{pics/full/pField2/frankOseenAndHelfrichEnergy.tex}
				\end{minipage}
				\begin{minipage}{0.49\textwidth}
					\centering
%
%
\begin{tikzpicture}

\begin{axis}[%
width=0.7\textwidth,
height=0.3\textwidth,
at={(0,0)},
scale only axis,
separate axis lines,
every outer x axis line/.append style={black},
every x tick label/.append style={font=\color{black}},
xmin=0,
xmax=15,
xtick = {0, 5, 10, 15},
xlabel={$\mbox{Time}$},
every outer y axis line/.append style={black},
every y tick label/.append style={font=\color{black}},
ymin=39.2,
ymax=39.8,
ytick = {39.2, 39.4, 39.6, 39.8},
ylabel={$\PotEP$},
axis background/.style={fill=white},
x label style={font=\scriptsize,at={(axis description cs:0.5,0.1)}},
y label style={font=\scriptsize,at={(axis description cs:0.05,0.5)}},
yticklabel style = {font=\scriptsize},
xticklabel style = {font=\scriptsize},
]
\addplot [color=black,solid,line width=2.5pt,forget plot]
  table[row sep=crcr]{%
0	39.7928926\\
0.05	39.6307075\\
0.1	39.4850128\\
0.15	39.4160004\\
0.2	39.3723356\\
0.25	39.3422753\\
0.3	39.3211624\\
0.35	39.3052482\\
0.4	39.2928715\\
0.45	39.2832212\\
0.5	39.2756491\\
0.55	39.269496\\
0.6	39.264471\\
0.65	39.2605342\\
0.7	39.2572351\\
0.75	39.2545117\\
0.8	39.2522362\\
0.85	39.2502258\\
0.9	39.2487193\\
0.95	39.2471833\\
1	39.2459508\\
1.05	39.2448432\\
1.1	39.243852\\
1.15	39.2431193\\
1.2	39.2422873\\
1.25	39.2415274\\
1.3	39.2409209\\
1.35	39.2403486\\
1.4	39.2397614\\
1.45	39.2392901\\
1.5	39.2387954\\
1.55	39.2382676\\
1.6	39.2378874\\
1.65	39.237455\\
1.7	39.2370807\\
1.75	39.2367249\\
1.8	39.2363078\\
1.85	39.2359196\\
1.9	39.2355604\\
1.95	39.2352103\\
2	39.2349796\\
2.05	39.2346782\\
2.1	39.2342763\\
2.15	39.233984\\
2.2	39.2337112\\
2.25	39.2333481\\
2.3	39.2330947\\
2.35	39.2327411\\
2.4	39.2325072\\
2.45	39.2322731\\
2.5	39.2319489\\
2.55	39.2317245\\
2.6	39.2314099\\
2.65	39.2312052\\
2.7	39.2309904\\
2.75	39.2306954\\
2.8	39.2303904\\
2.85	39.2301053\\
2.9	39.2299101\\
2.95	39.2296348\\
3	39.2293494\\
3.05	39.229174\\
3.1	39.2289985\\
3.15	39.2287229\\
3.2	39.2285573\\
3.25	39.2283016\\
3.3	39.2280459\\
3.35	39.2277801\\
3.4	39.2275242\\
3.45	39.2274783\\
3.5	39.2272323\\
3.55	39.2269863\\
3.6	39.2267402\\
3.65	39.2266041\\
3.7	39.226368\\
3.75	39.2261318\\
3.8	39.2260055\\
3.85	39.2257792\\
3.9	39.2256529\\
3.95	39.2253365\\
4	39.2251201\\
4.05	39.2249936\\
4.1	39.2247871\\
4.15	39.2246706\\
4.2	39.224464\\
4.25	39.2242573\\
4.3	39.2241507\\
4.35	39.223944\\
4.4	39.2238472\\
4.45	39.2235504\\
4.5	39.2234536\\
4.55	39.2233568\\
4.6	39.2230699\\
4.65	39.2229729\\
4.7	39.222786\\
4.75	39.222609\\
4.8	39.2224319\\
4.85	39.2223549\\
4.9	39.2221678\\
4.95	39.2219906\\
5	39.2218135\\
5.05	39.2217463\\
5.1	39.221479\\
5.15	39.2214118\\
5.2	39.2212445\\
5.25	39.2211771\\
5.3	39.2210098\\
5.35	39.2208624\\
5.4	39.220805\\
5.45	39.2205475\\
5.5	39.22039\\
5.55	39.2203425\\
5.6	39.220185\\
5.65	39.2200374\\
5.7	39.2199898\\
5.75	39.2198422\\
5.8	39.2196945\\
5.85	39.2195669\\
5.9	39.2194192\\
5.95	39.2193714\\
6	39.2192437\\
6.05	39.2191059\\
6.1	39.2189681\\
6.15	39.2188402\\
6.2	39.2187024\\
6.25	39.2186745\\
6.3	39.2185466\\
6.35	39.2185187\\
6.4	39.2184007\\
6.45	39.2182727\\
6.5	39.2181447\\
6.55	39.2180267\\
6.6	39.2179086\\
6.65	39.2177906\\
6.7	39.2176725\\
6.75	39.2176644\\
6.8	39.2175462\\
6.85	39.2174281\\
6.9	39.2173199\\
6.95	39.2172017\\
7	39.2172035\\
7.05	39.2170852\\
7.1	39.2169869\\
7.15	39.2167787\\
7.2	39.2167804\\
7.25	39.216672\\
7.3	39.2165737\\
7.35	39.2165653\\
7.4	39.216477\\
7.45	39.2163686\\
7.5	39.2161801\\
7.55	39.2161917\\
7.6	39.2160832\\
7.65	39.2159948\\
7.7	39.2159963\\
7.75	39.2158178\\
7.8	39.2157193\\
7.85	39.2157307\\
7.9	39.2156422\\
7.95	39.2156536\\
8	39.215475\\
8.05	39.2153864\\
8.1	39.2153978\\
8.15	39.2152091\\
8.2	39.2152305\\
8.25	39.2151618\\
8.3	39.2150731\\
8.35	39.2149944\\
8.4	39.2149157\\
8.45	39.214937\\
8.5	39.2147582\\
8.55	39.2147795\\
8.6	39.2147107\\
8.65	39.2146419\\
8.7	39.2145631\\
8.75	39.2144943\\
8.8	39.2144254\\
8.85	39.2144466\\
8.9	39.2142777\\
8.95	39.2143089\\
9	39.21414\\
9.05	39.2141711\\
9.1	39.2141122\\
9.15	39.2140333\\
9.2	39.2139743\\
9.25	39.2139154\\
9.3	39.2138464\\
9.35	39.2137775\\
9.4	39.2137185\\
9.45	39.2136595\\
9.5	39.2137005\\
9.55	39.2135415\\
9.6	39.2135824\\
9.65	39.2134134\\
9.7	39.2134643\\
9.75	39.2133053\\
9.8	39.2133562\\
9.85	39.2132871\\
9.9	39.213238\\
9.95	39.2131789\\
10	39.2131298\\
10.05	39.2130707\\
10.1	39.2131216\\
10.15	39.2129724\\
10.2	39.2130233\\
10.25	39.2129641\\
10.3	39.2129149\\
10.35	39.2128657\\
10.4	39.2128165\\
10.45	39.2127673\\
10.5	39.2127181\\
10.55	39.2127689\\
10.6	39.2126297\\
10.65	39.2125704\\
10.7	39.2126312\\
10.75	39.2124819\\
10.8	39.2125427\\
10.85	39.2125034\\
10.9	39.2123541\\
10.95	39.2124048\\
11	39.2123655\\
11.05	39.2123262\\
11.1	39.2122869\\
11.15	39.2122376\\
11.2	39.2121983\\
11.25	39.2121589\\
11.3	39.2121196\\
11.35	39.2121802\\
11.4	39.2120409\\
11.45	39.2120015\\
11.5	39.2120621\\
11.55	39.2119227\\
11.6	39.2118934\\
11.65	39.211954\\
11.7	39.2119146\\
11.75	39.2117752\\
11.8	39.2118457\\
11.85	39.2118063\\
11.9	39.2117769\\
11.95	39.2116374\\
12	39.211698\\
12.05	39.2116686\\
12.1	39.2116391\\
12.15	39.2116096\\
12.2	39.2115802\\
12.25	39.2115407\\
12.3	39.2115012\\
12.35	39.2115717\\
12.4	39.2114422\\
12.45	39.2114127\\
12.5	39.2113832\\
12.55	39.2114637\\
12.6	39.2113342\\
12.65	39.2112947\\
12.7	39.2112652\\
12.75	39.2113356\\
12.8	39.2113161\\
12.85	39.2111865\\
12.9	39.211157\\
12.95	39.2112274\\
13	39.2111979\\
13.05	39.2111683\\
13.1	39.2111488\\
13.15	39.2111192\\
13.2	39.2110996\\
13.25	39.21107\\
13.3	39.2110504\\
13.35	39.2110208\\
13.4	39.2109912\\
13.45	39.2109716\\
13.5	39.210942\\
13.55	39.2109224\\
13.6	39.2110028\\
13.65	39.2109832\\
13.7	39.2108536\\
13.75	39.2108339\\
13.8	39.2108143\\
13.85	39.2108847\\
13.9	39.210865\\
13.95	39.2108454\\
14	39.2107157\\
14.05	39.2106961\\
14.1	39.2106764\\
14.15	39.2107568\\
14.2	39.2107371\\
14.25	39.2107174\\
14.3	39.2107077\\
14.35	39.2105781\\
14.4	39.2105584\\
14.45	39.2106387\\
14.5	39.210619\\
14.55	39.2105993\\
14.6	39.2105796\\
14.65	39.2105699\\
14.7	39.2105502\\
14.75	39.2105305\\
14.8	39.2105108\\
14.85	39.2105011\\
14.9	39.2104814\\
14.95	39.2104517\\
15	39.2104419\\
};
\end{axis}
\end{tikzpicture}%
				\end{minipage}
				\begin{minipage}{0.6\textwidth}
					\centering
					\def\picwidth{0.24\textwidth}
					\includegraphics[width=\picwidth]{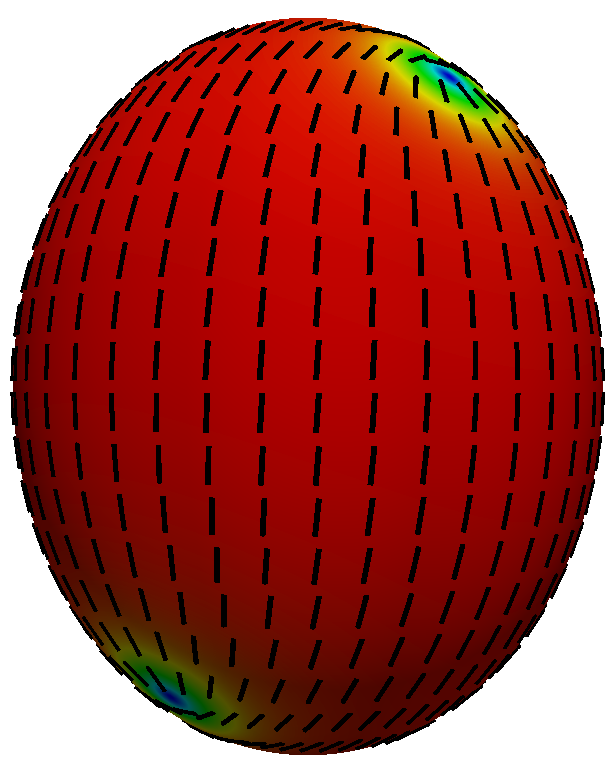}
					\includegraphics[width=\picwidth]{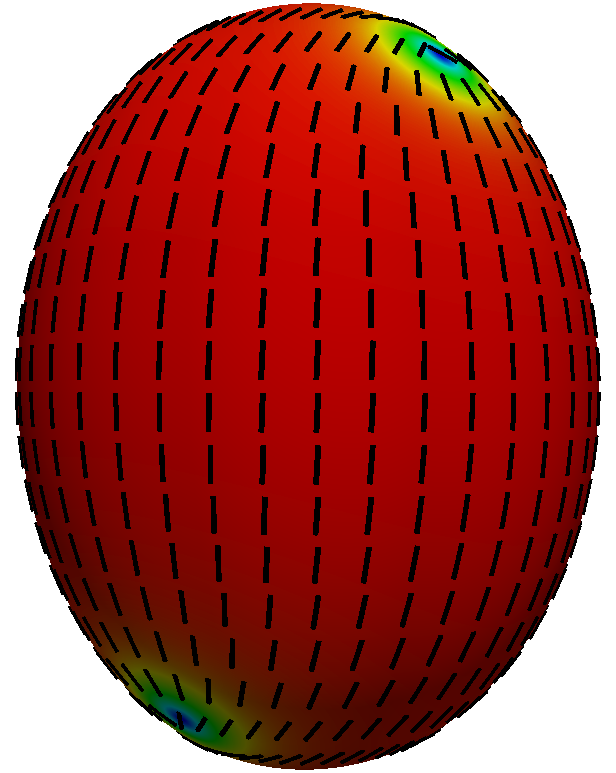}
					\includegraphics[width=\picwidth]{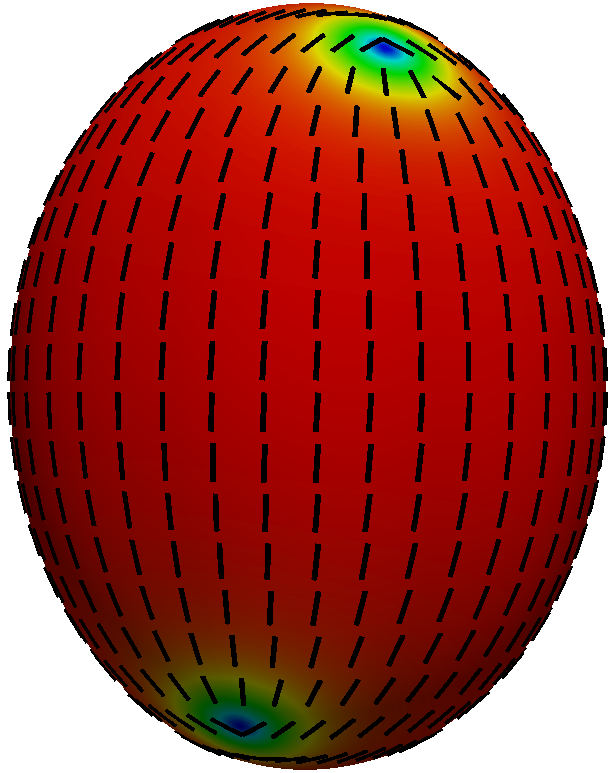}
					\includegraphics[width=\picwidth]{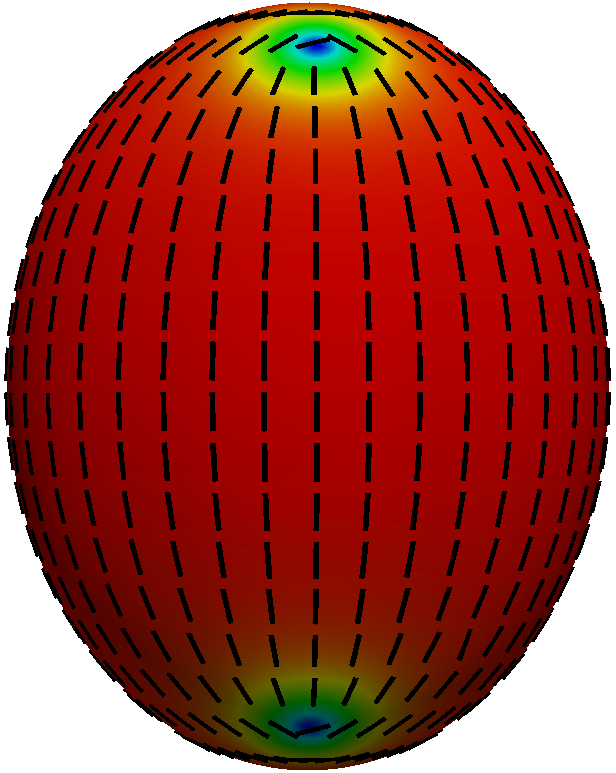}
				\end{minipage}
				\begin{minipage}{0.49\textwidth}
					\centering
%
%
\begin{tikzpicture}

\begin{axis}[%
width=0.7\textwidth,
height=0.3\textwidth,
at={(0,0)},
scale only axis,
separate axis lines,
every outer x axis line/.append style={black},
every x tick label/.append style={font=\color{black}},
xmin=0,
xmax=250,
xtick = {0, 50, 100, 150, 200, 250},
xlabel={$\mbox{Time}$},
every outer y axis line/.append style={black},
every y tick label/.append style={font=\color{black}},
ymin=20.5,
ymax=22,
ylabel={$\mbox{Energy}$},
ytick = {20.5, 21, 21.5, 22},
axis background/.style={fill=white},
x label style={font=\scriptsize,at={(axis description cs:0.5,0.05)}},
y label style={font=\scriptsize,at={(axis description cs:0.05,0.5)}},
yticklabel style = {font=\scriptsize},
xticklabel style = {font=\scriptsize},
legend style={font=\scriptsize,at={(0.5,0.95)},anchor=north,legend columns=2,legend cell align=left,align=left,draw=black}, 
]
\addplot [color=red,solid,line width=1.5pt]
  table[row sep=crcr]{%
0.1	21.8895\\
1.1	21.6122\\
2.1	21.581\\
3.1	21.5553\\
4.1	21.5316\\
5.1	21.5109\\
6.1	21.4924\\
7.1	21.476\\
8.1	21.4615\\
9.1	21.4487\\
10.1	21.4374\\
11.1	21.4274\\
12.1	21.4186\\
13.1	21.4107\\
14.1	21.4039\\
15.1	21.3978\\
16.1	21.3924\\
17.1	21.3877\\
18.1	21.3835\\
19.1	21.3798\\
20.1	21.3765\\
21.1	21.3736\\
22.1	21.3711\\
23.1	21.3688\\
24.1	21.3668\\
25.1	21.3651\\
26.1	21.3635\\
27.1	21.3621\\
28.1	21.3609\\
29.1	21.3599\\
30.1	21.3589\\
31.1	21.3581\\
32.1	21.3574\\
33.1	21.3567\\
34.1	21.3561\\
35.1	21.3556\\
36.1	21.3552\\
37.1	21.3548\\
38.1	21.3544\\
39.1	21.3541\\
40.1	21.3539\\
41.1	21.3536\\
42.1	21.3534\\
43.1	21.3532\\
44.1	21.3531\\
45.1	21.3529\\
46.1	21.3528\\
47.1	21.3527\\
48.1	21.3526\\
49.1	21.3525\\
50.1	21.3524\\
51.1	21.3523\\
52.1	21.3523\\
53.1	21.3522\\
54.1	21.3522\\
55.1	21.3521\\
56.1	21.3521\\
57.1	21.352\\
58.1	21.352\\
59.1	21.352\\
60.1	21.352\\
61.1	21.3519\\
62.1	21.3519\\
63.1	21.3519\\
64.1	21.3519\\
65.1	21.3519\\
66.1	21.3519\\
67.1	21.3519\\
68.1	21.3518\\
69.1	21.3518\\
70.1	21.3518\\
71.1	21.3518\\
72.1	21.3518\\
73.1	21.3518\\
74.1	21.3518\\
75.1	21.3518\\
76.1	21.3518\\
77.1	21.3518\\
78.1	21.3518\\
79.1	21.3518\\
80.1	21.3518\\
81.1	21.3518\\
82.1	21.3518\\
83.1	21.3518\\
84.1	21.3518\\
85.1	21.3517\\
86.1	21.3517\\
87.1	21.3517\\
88.1	21.3517\\
89.1	21.3517\\
90.1	21.3517\\
91.1	21.3517\\
92.1	21.3517\\
93.1	21.3517\\
94.1	21.3517\\
95.1	21.3516\\
96.1	21.3516\\
97.1	21.3516\\
98.1	21.3516\\
99.1	21.3516\\
100.1	21.3516\\
101.1	21.3515\\
102.1	21.3515\\
103.1	21.3515\\
104.1	21.3515\\
105.1	21.3514\\
106.1	21.3514\\
107.1	21.3513\\
108.1	21.3513\\
109.1	21.3512\\
110.1	21.3512\\
111.1	21.3511\\
112.1	21.3511\\
113.1	21.351\\
114.1	21.3509\\
115.1	21.3508\\
116.1	21.3507\\
117.1	21.3506\\
118.1	21.3505\\
119.1	21.3504\\
120.1	21.3503\\
121.1	21.3501\\
122.1	21.3499\\
123.1	21.3498\\
124.1	21.3496\\
125.1	21.3494\\
126.1	21.3491\\
127.1	21.3489\\
128.1	21.3486\\
129.1	21.3483\\
130.1	21.3479\\
131.1	21.3475\\
132.1	21.3471\\
133.1	21.3467\\
134.1	21.3462\\
135.1	21.3457\\
136.1	21.3451\\
137.1	21.3445\\
138.1	21.3438\\
139.1	21.343\\
140.1	21.3422\\
141.1	21.3413\\
142.1	21.3404\\
143.1	21.3393\\
144.1	21.3382\\
145.1	21.337\\
146.1	21.3357\\
147.1	21.3343\\
148.1	21.3327\\
149.1	21.3311\\
150.1	21.3293\\
151.1	21.3275\\
152.1	21.3255\\
153.1	21.3233\\
154.1	21.321\\
155.1	21.3186\\
156.1	21.3161\\
157.1	21.3134\\
158.1	21.3106\\
159.1	21.3076\\
160.1	21.3045\\
161.1	21.3013\\
162.1	21.298\\
163.1	21.2945\\
164.1	21.291\\
165.1	21.2873\\
166.1	21.2837\\
167.1	21.2799\\
168.1	21.2761\\
169.1	21.2723\\
170.1	21.2684\\
171.1	21.2646\\
172.1	21.2608\\
173.1	21.257\\
174.1	21.2533\\
175.1	21.2496\\
176.1	21.246\\
177.1	21.2425\\
178.1	21.2391\\
179.1	21.2359\\
180.1	21.2327\\
181.1	21.2296\\
182.1	21.2267\\
183.1	21.2239\\
183.9	21.2217\\
184.9	21.2191\\
185.9	21.2167\\
186.9	21.2144\\
187.9	21.2122\\
188.9	21.2101\\
189.9	21.2081\\
190.9	21.2063\\
191.9	21.2046\\
192.9	21.2029\\
193.9	21.2014\\
194.9	21.2\\
195.9	21.1986\\
196.9	21.1974\\
197.9	21.1962\\
198.9	21.1951\\
199.9	21.1941\\
200.9	21.1932\\
201.9	21.1923\\
202.9	21.1915\\
203.9	21.1907\\
204.9	21.19\\
205.9	21.1893\\
206.9	21.1887\\
207.9	21.1882\\
208.9	21.1876\\
209.9	21.1871\\
210.9	21.1867\\
211.9	21.1863\\
212.9	21.1859\\
213.9	21.1855\\
214.9	21.1852\\
215.9	21.1849\\
216.9	21.1846\\
217.9	21.1843\\
218.9	21.1841\\
219.9	21.1838\\
220.9	21.1836\\
221.9	21.1834\\
222.9	21.1833\\
223.9	21.1831\\
224.9	21.183\\
225.9	21.1828\\
226.9	21.1827\\
227.9	21.1826\\
228.9	21.1824\\
229.9	21.1823\\
230.9	21.1823\\
231.9	21.1822\\
232.9	21.1821\\
233.9	21.182\\
234.9	21.1819\\
235.9	21.1819\\
236.9	21.1818\\
237.9	21.1818\\
238.9	21.1817\\
239.9	21.1817\\
240.9	21.1816\\
241.9	21.1816\\
242.9	21.1816\\
243.9	21.1815\\
244.9	21.1815\\
245.9	21.1815\\
246.9	21.1814\\
247.9	21.1814\\
248.9	21.1814\\
249.9	21.1814\\
};
\addlegendentry{$\landaudegennesEnergy$};

\addplot [color=blue,solid,line width=1.5pt]
  table[row sep=crcr]{%
0.1	20.606\\
1.1	20.673\\
2.1	20.6973\\
3.1	20.7204\\
4.1	20.7417\\
5.1	20.7602\\
6.1	20.777\\
7.1	20.792\\
8.1	20.8055\\
9.1	20.8176\\
10.1	20.8284\\
11.1	20.838\\
12.1	20.8466\\
13.1	20.8542\\
14.1	20.8611\\
15.1	20.8671\\
16.1	20.8725\\
17.1	20.8773\\
18.1	20.8815\\
19.1	20.8853\\
20.1	20.8886\\
21.1	20.8916\\
22.1	20.8942\\
23.1	20.8965\\
24.1	20.8986\\
25.1	20.9004\\
26.1	20.902\\
27.1	20.9034\\
28.1	20.9047\\
29.1	20.9058\\
30.1	20.9068\\
31.1	20.9076\\
32.1	20.9084\\
33.1	20.9091\\
34.1	20.9097\\
35.1	20.9102\\
36.1	20.9107\\
37.1	20.9111\\
38.1	20.9114\\
39.1	20.9118\\
40.1	20.9121\\
41.1	20.9123\\
42.1	20.9125\\
43.1	20.9127\\
44.1	20.9129\\
45.1	20.913\\
46.1	20.9132\\
47.1	20.9133\\
48.1	20.9134\\
49.1	20.9135\\
50.1	20.9136\\
51.1	20.9137\\
52.1	20.9137\\
53.1	20.9138\\
54.1	20.9138\\
55.1	20.9139\\
56.1	20.9139\\
57.1	20.9139\\
58.1	20.914\\
59.1	20.914\\
60.1	20.914\\
61.1	20.9141\\
62.1	20.9141\\
63.1	20.9141\\
64.1	20.9141\\
65.1	20.9141\\
66.1	20.9141\\
67.1	20.9141\\
68.1	20.9141\\
69.1	20.9142\\
70.1	20.9142\\
71.1	20.9142\\
72.1	20.9142\\
73.1	20.9142\\
74.1	20.9142\\
75.1	20.9142\\
76.1	20.9142\\
77.1	20.9142\\
78.1	20.9142\\
79.1	20.9142\\
80.1	20.9142\\
81.1	20.9142\\
82.1	20.9142\\
83.1	20.9142\\
84.1	20.9142\\
85.1	20.9142\\
86.1	20.9142\\
87.1	20.9142\\
88.1	20.9142\\
89.1	20.9142\\
90.1	20.9142\\
91.1	20.9142\\
92.1	20.9142\\
93.1	20.9142\\
94.1	20.9142\\
95.1	20.9142\\
96.1	20.9142\\
97.1	20.9142\\
98.1	20.9143\\
99.1	20.9143\\
100.1	20.9143\\
101.1	20.9143\\
102.1	20.9143\\
103.1	20.9143\\
104.1	20.9143\\
105.1	20.9143\\
106.1	20.9143\\
107.1	20.9143\\
108.1	20.9143\\
109.1	20.9144\\
110.1	20.9144\\
111.1	20.9144\\
112.1	20.9144\\
113.1	20.9144\\
114.1	20.9145\\
115.1	20.9145\\
116.1	20.9145\\
117.1	20.9145\\
118.1	20.9146\\
119.1	20.9146\\
120.1	20.9147\\
121.1	20.9147\\
122.1	20.9148\\
123.1	20.9148\\
124.1	20.9149\\
125.1	20.9149\\
126.1	20.915\\
127.1	20.9151\\
128.1	20.9152\\
129.1	20.9153\\
130.1	20.9154\\
131.1	20.9155\\
132.1	20.9157\\
133.1	20.9158\\
134.1	20.916\\
135.1	20.9161\\
136.1	20.9163\\
137.1	20.9165\\
138.1	20.9168\\
139.1	20.917\\
140.1	20.9173\\
141.1	20.9176\\
142.1	20.9179\\
143.1	20.9182\\
144.1	20.9186\\
145.1	20.919\\
146.1	20.9195\\
147.1	20.9199\\
148.1	20.9205\\
149.1	20.921\\
150.1	20.9216\\
151.1	20.9223\\
152.1	20.923\\
153.1	20.9237\\
154.1	20.9245\\
155.1	20.9254\\
156.1	20.9263\\
157.1	20.9272\\
158.1	20.9282\\
159.1	20.9293\\
160.1	20.9305\\
161.1	20.9316\\
162.1	20.9329\\
163.1	20.9342\\
164.1	20.9355\\
165.1	20.9369\\
166.1	20.9383\\
167.1	20.9398\\
168.1	20.9413\\
169.1	20.9428\\
170.1	20.9443\\
171.1	20.9459\\
172.1	20.9475\\
173.1	20.949\\
174.1	20.9506\\
175.1	20.9522\\
176.1	20.9537\\
177.1	20.9552\\
178.1	20.9567\\
179.1	20.9582\\
180.1	20.9596\\
181.1	20.961\\
182.1	20.9624\\
183.1	20.9637\\
183.9	20.9647\\
184.9	20.966\\
185.9	20.9672\\
186.9	20.9683\\
187.9	20.9694\\
188.9	20.9704\\
189.9	20.9714\\
190.9	20.9724\\
191.9	20.9733\\
192.9	20.9741\\
193.9	20.9749\\
194.9	20.9757\\
195.9	20.9764\\
196.9	20.9771\\
197.9	20.9777\\
198.9	20.9783\\
199.9	20.9789\\
200.9	20.9795\\
201.9	20.98\\
202.9	20.9804\\
203.9	20.9809\\
204.9	20.9813\\
205.9	20.9817\\
206.9	20.982\\
207.9	20.9824\\
208.9	20.9827\\
209.9	20.983\\
210.9	20.9833\\
211.9	20.9835\\
212.9	20.9838\\
213.9	20.984\\
214.9	20.9842\\
215.9	20.9844\\
216.9	20.9846\\
217.9	20.9847\\
218.9	20.9849\\
219.9	20.985\\
220.9	20.9852\\
221.9	20.9853\\
222.9	20.9854\\
223.9	20.9855\\
224.9	20.9856\\
225.9	20.9857\\
226.9	20.9858\\
227.9	20.9859\\
228.9	20.9859\\
229.9	20.986\\
230.9	20.9861\\
231.9	20.9861\\
232.9	20.9862\\
233.9	20.9862\\
234.9	20.9863\\
235.9	20.9863\\
236.9	20.9863\\
237.9	20.9864\\
238.9	20.9864\\
239.9	20.9864\\
240.9	20.9865\\
241.9	20.9865\\
242.9	20.9865\\
243.9	20.9865\\
244.9	20.9866\\
245.9	20.9866\\
246.9	20.9866\\
247.9	20.9866\\
248.9	20.9866\\
249.9	20.9866\\
};
\addlegendentry{$\helfrichEnergy$};

\end{axis}
\end{tikzpicture}%
				\end{minipage}
				\begin{minipage}{0.49\textwidth}
					\centering
%
%
\begin{tikzpicture}

\begin{axis}[%
width=0.7\textwidth,
height=0.3\textwidth,
at={(0,0)},
scale only axis,
separate axis lines,
every outer x axis line/.append style={black},
every x tick label/.append style={font=\color{black}},
xmin=0,
xmax=250,
xtick = {0, 50, 100, 150, 200, 250},
xlabel={$\mbox{Time}$},
every outer y axis line/.append style={black},
every y tick label/.append style={font=\color{black}},
ymin=42.15,
ymax=42.45,
ytick = {42.2, 42.3, 42.4},
ylabel={$\PotEQ$},
axis background/.style={fill=white},
x label style={font=\scriptsize,at={(axis description cs:0.5,0.05)}},
y label style={font=\scriptsize,at={(axis description cs:0.05,0.5)}},
yticklabel style = {font=\scriptsize},
xticklabel style = {font=\scriptsize},
]
\addplot [color=black,solid,line width=2.5pt,forget plot]
  table[row sep=crcr]{%
0.1	42.5209056\\
1.1	42.3228872\\
2.1	42.316022\\
3.1	42.3134012\\
4.1	42.3109947\\
5.1	42.3087888\\
6.1	42.3070841\\
7.1	42.3056802\\
8.1	42.3046768\\
9.1	42.3039737\\
10.1	42.3034709\\
11.1	42.3030684\\
12.1	42.3028661\\
13.1	42.3025641\\
14.1	42.3026623\\
15.1	42.3025606\\
16.1	42.3025591\\
17.1	42.3026578\\
18.1	42.3026567\\
19.1	42.3027556\\
20.1	42.3027547\\
21.1	42.3028539\\
22.1	42.3029531\\
23.1	42.3029525\\
24.1	42.3030519\\
25.1	42.3031514\\
26.1	42.3031509\\
27.1	42.3031505\\
28.1	42.3032502\\
29.1	42.3033499\\
30.1	42.3033496\\
31.1	42.3033493\\
32.1	42.3034491\\
33.1	42.3034489\\
34.1	42.3034487\\
35.1	42.3034486\\
36.1	42.3035485\\
37.1	42.3035483\\
38.1	42.3034482\\
39.1	42.3035481\\
40.1	42.3036481\\
41.1	42.303548\\
42.1	42.3035479\\
43.1	42.3035479\\
44.1	42.3036478\\
45.1	42.3035478\\
46.1	42.3036477\\
47.1	42.3036477\\
48.1	42.3036477\\
49.1	42.3036476\\
50.1	42.3036476\\
51.1	42.3036476\\
52.1	42.3036476\\
53.1	42.3036476\\
54.1	42.3036476\\
55.1	42.3036475\\
56.1	42.3036475\\
57.1	42.3035475\\
58.1	42.3036475\\
59.1	42.3036475\\
60.1	42.3036475\\
61.1	42.3036475\\
62.1	42.3036475\\
63.1	42.3036475\\
64.1	42.3036475\\
65.1	42.3036475\\
66.1	42.3036475\\
67.1	42.3036475\\
68.1	42.3035475\\
69.1	42.3036475\\
70.1	42.3036475\\
71.1	42.3036475\\
72.1	42.3036475\\
73.1	42.3036475\\
74.1	42.3036475\\
75.1	42.3036475\\
76.1	42.3036475\\
77.1	42.3036475\\
78.1	42.3036475\\
79.1	42.3036475\\
80.1	42.3036475\\
81.1	42.3036475\\
82.1	42.3036475\\
83.1	42.3036475\\
84.1	42.3036475\\
85.1	42.3035475\\
86.1	42.3035475\\
87.1	42.3035475\\
88.1	42.3035475\\
89.1	42.3035475\\
90.1	42.3035475\\
91.1	42.3035475\\
92.1	42.3035475\\
93.1	42.3035475\\
94.1	42.3035475\\
95.1	42.3034475\\
96.1	42.3034475\\
97.1	42.3034475\\
98.1	42.3035475\\
99.1	42.3035475\\
100.1	42.3035475\\
101.1	42.3034475\\
102.1	42.3034475\\
103.1	42.3034475\\
104.1	42.3034475\\
105.1	42.3033475\\
106.1	42.3033475\\
107.1	42.3032475\\
108.1	42.3032475\\
109.1	42.3032475\\
110.1	42.3032475\\
111.1	42.3031476\\
112.1	42.3031476\\
113.1	42.3030476\\
114.1	42.3030476\\
115.1	42.3029476\\
116.1	42.3028476\\
117.1	42.3027476\\
118.1	42.3027477\\
119.1	42.3026477\\
120.1	42.3026477\\
121.1	42.3024477\\
122.1	42.3023478\\
123.1	42.3022478\\
124.1	42.3021478\\
125.1	42.3019479\\
126.1	42.3017479\\
127.1	42.3016479\\
128.1	42.301448\\
129.1	42.301248\\
130.1	42.3009481\\
131.1	42.3006481\\
132.1	42.3004482\\
133.1	42.3001483\\
134.1	42.2998484\\
135.1	42.2994484\\
136.1	42.2990485\\
137.1	42.2986486\\
138.1	42.2982487\\
139.1	42.2976489\\
140.1	42.297149\\
141.1	42.2965491\\
142.1	42.2959493\\
143.1	42.2951494\\
144.1	42.2944496\\
145.1	42.2936498\\
146.1	42.29285\\
147.1	42.2918502\\
148.1	42.2908504\\
149.1	42.2897506\\
150.1	42.2885509\\
151.1	42.2874511\\
152.1	42.2861514\\
153.1	42.2846517\\
154.1	42.283152\\
155.1	42.2816523\\
156.1	42.2800526\\
157.1	42.2782529\\
158.1	42.2764532\\
159.1	42.2745535\\
160.1	42.2726539\\
161.1	42.2705542\\
162.1	42.2685545\\
163.1	42.2663548\\
164.1	42.2641551\\
165.1	42.2618554\\
166.1	42.2596556\\
167.1	42.2573559\\
168.1	42.2550561\\
169.1	42.2527563\\
170.1	42.2503565\\
171.1	42.2481566\\
172.1	42.2459568\\
173.1	42.2436569\\
174.1	42.2415569\\
175.1	42.239457\\
176.1	42.237357\\
177.1	42.235357\\
178.1	42.233457\\
179.1	42.231757\\
180.1	42.2299569\\
181.1	42.2282568\\
182.1	42.2267567\\
183.1	42.2252566\\
183.9	42.2240566\\
184.9	42.2227564\\
185.9	42.2215563\\
186.9	42.2203562\\
187.9	42.2192561\\
188.9	42.2181559\\
189.9	42.2171558\\
190.9	42.2163557\\
191.9	42.2155556\\
192.9	42.2146554\\
193.9	42.2139553\\
194.9	42.2133552\\
195.9	42.2126551\\
196.9	42.212155\\
197.9	42.2115549\\
198.9	42.2110548\\
199.9	42.2106547\\
200.9	42.2103546\\
201.9	42.2099545\\
202.9	42.2095544\\
203.9	42.2092544\\
204.9	42.2089543\\
205.9	42.2086542\\
206.9	42.2083542\\
207.9	42.2082541\\
208.9	42.207954\\
209.9	42.207754\\
210.9	42.2076539\\
211.9	42.2074539\\
212.9	42.2073538\\
213.9	42.2071538\\
214.9	42.2070538\\
215.9	42.2069537\\
216.9	42.2068537\\
217.9	42.2066537\\
218.9	42.2066536\\
219.9	42.2064536\\
220.9	42.2064536\\
221.9	42.2063536\\
222.9	42.2063535\\
223.9	42.2062535\\
224.9	42.2062535\\
225.9	42.2061535\\
226.9	42.2061535\\
227.9	42.2061534\\
228.9	42.2059534\\
229.9	42.2059534\\
230.9	42.2060534\\
231.9	42.2059534\\
232.9	42.2059534\\
233.9	42.2058534\\
234.9	42.2058534\\
235.9	42.2058534\\
236.9	42.2057534\\
237.9	42.2058533\\
238.9	42.2057533\\
239.9	42.2057533\\
240.9	42.2057533\\
241.9	42.2057533\\
242.9	42.2057533\\
243.9	42.2056533\\
244.9	42.2057533\\
245.9	42.2057533\\
246.9	42.2056533\\
247.9	42.2056533\\
248.9	42.2056533\\
249.9	42.2056533\\
};
\end{axis}
\end{tikzpicture}%
				\end{minipage}
			\end{minipage}
	}
    \caption{Top: Evolution of the director field of the \frankOseenHelfrich\ model for $t=0, 1, 15$. Bottom: Evolution of the Q tensor field of the \landauDeGennesHelfrich\ model for $t=0, 100, 170, 250$. In both cases the same initial condition as in \autoref{fig:director_response:results} are used and energies against time (left) and dissipation rates against time (right) are plotted.}
    \label{fig:results:qtensor2}
\end{figure}

In order to better understand the observed behaviour in the previous section we run the full models with the equilibrium configurations on an ellipsoid as initial conditions.
\autoref{fig:results:qtensor2} shows the corresponding evolution together with their energy plots.
We observe the same behaviour. 
In the polar case the ellipsoid with the two $+1$ defects located at the poles of the long axis remains symmetric but deforms at the positions of the defects leading to high curvature values at these points. 
In the nematic case the shape becomes asymmetric and the four $+\nicefrac{1}{2}$ defects arrange in a planar configuration. 
To observe this behaviour requires long time simulations and an accurate numerical scheme.
The energy plots clearly show that the increase in \helfrich\ energy $\helfrichEnergy$ is overcompensated by the decrease in the \frankOseen\ energy $\frankOseenEnergy$ and the \landauDeGennes\ energy $\landaudegennesEnergy$.
Both shapes significantly differ from the proposed equilibrium shapes in \cite{Parketal_EPL_1992}, where $n-$atic order on deformable surfaces is considered, but with a much simpler and purely intrinsic model. 
More recently, it has been demonstrated that besides these intrinsic curvature terms also extrinsic curvature terms, \ie\ curvatures related to the geometry of the embedding space, are relevant \cite{Napolietal_PRE_2012,Napolietal_PRL_2012,Mbangaetal_PRL_2012,Nguyenetal_SM_2013,Segattietal_PRE_2014,Napolietal_PRE_2016,Segattietal_M3AS_2016,Koningetal_PRE_2016,Duanetal_PRE_2017,Nestleretal_JNS_2018,Nitschkeetal_PRSA_2018,Napolietal_PRE_2018}. 
It has been demonstrated that the intrinsic geometry tends to confine topological defects to regions of maximal Gaussian curvature, while extrinsic couplings tend to orient the director field along minimal curvature lines.
Extrinsic curvature has also been shown to expel defects from regions of maximum curvature above a critical coupling threshold \cite{Mbangaetal_PRL_2012}, to modify the defect arrangement from tetrahedral to planar \cite{Nguyenetal_SM_2013}, to change the phase diagram allowing for coexistence of nematic and isotropic phases in curved two-dimensional liquid crystals \cite{Nitschkeetal_PRSA_2018} and to modify the critical geometry deformation parameter, \cf\ $C$, in the considered nonic and bi-nonic surfaces which lead to energetically favourable solutions with more defects than topologically necessary \cite{Nestleretal_JNS_2018}. 
The observed arrangements thus add to these phenomena and can also be explained by an interplay of intrinsic and extrinsic curvature effects.

Mathematically the extrinsic contributions result from the proposed anchoring conditions on the boundary of the thin film \cite{Nestleretal_JNS_2018,Nitschkeetal_PRSA_2018}.
If these terms are neglected the \frankOseenHelfrich\ and the elastic part of the \landauDeGennesHelfrich\ energy have to be modified to
\begin{align} 
        \frankOseenEnergy = \frac{K}{2}\int_{\surf} \left\| \GradSurf\dirf \right\|^2 \mu, \qquad
        \elasticEnergy = \frac{L}{2} \int_{\surf} \left\| \GradSurf\qten \right\|^2 \mu \formPeriod
\end{align}
All other energies in eqs. \eqref{eq:energies1}-\eqref{eq:energies7} remain. 
This results in the intrinsic evolution laws for $ \dirf $
\begin{align}
    \label{eq:gradientflow:p:vnor:intrinsic}
    \kinConst\vnor &= -\alpha\left(\laplaceBeltrami\meanc + \meanc\left( \frac{\meanc^2}{2} - 2\gaussc \right)\right)
             +\frac{\pnorm}{4}\meanc\left( \left\| \dirf \right\|^2 - 1 \right)^2 \notag\\
          &\quad  -K\left( \DivSurf\left( \left( \levicivita{\shop}{\dirf} - \levicivita{\shop\dirf}{\dirf} \right)\dirf \right) 
                    + \left\langle \sigma^{\text{E}}_{\surf}, \proj_{\qspace}\shop \right\rangle\right) 
                    + \frac{\areaPen}{\areaZero^2}\left(\area-\areaZero\right)\meanc \\
    \label{eq:gradientflow:p:dirf:intrinsic}
    \kinConstP\dot{\dirf} &= K\laplaceBochner\dirf - \pnorm\left( \left\| \dirf \right\|^2 - 1 \right)\dirf
\end{align}
and $ \qten $
\begin{align}
    \label{eq:gradientflow:q:vnor:intrinsic}
    \kinConst\vnor &= -\alpha\left(\laplaceBeltrami\meanc + \meanc\left( \frac{\meanc^2}{2} - 2\gaussc \right)\right) + \meanc\left( a'\Tr\qten^2 + c\Tr\qten^4 + C_1 \right) \notag\\
        &\quad- 2L\DivSurf\left( \qten\levicivita{\shop}{\qten} - \levicivita{\qten\shop}{\qten} \right) - L\left\langle (\GradSurf\qten)^{T_{(123)}}:\GradSurf\qten,\proj_{\qspace}\shop \right\rangle 
         + \frac{\areaPen}{\areaZero^2}\left(\area-\areaZero\right)\meanc \\
    \label{eq:gradientflow:q:qten:intrinsic}
    \kinConstQ\dot{\qten} &= L\laplaceBochner\qten - 2\left( a'+c\Tr\qten^2 \right)\qten \formPeriod
\end{align}
These models can be solved with the same numerical approach, just setting various terms to zero.
\autoref{fig:results:intrinsic} shows the resulting equilibrium shapes for both models in comparison with their full models. 
The equilibrium shapes correspond qualitatively with those in \cite{Parketal_EPL_1992}. 
Defects are located at maximal Gaussian curvature points, the distortion of the surface in the vicinity of the defects is small and a tetrahedral defect arrangement is most favourable for the nematic case. 
The comparison between the equilibrium shapes for the intrinsic and the full model including also extrinsic curvature contributions further highlights the differences.

\begin{figure}[!h]
    \centering
    \ifthenelse{\boolean{useSnapshots}}
    {
		\ifthenelse{\boolean{forArxiv}}
		{
			\includegraphics[width=0.675\textwidth]{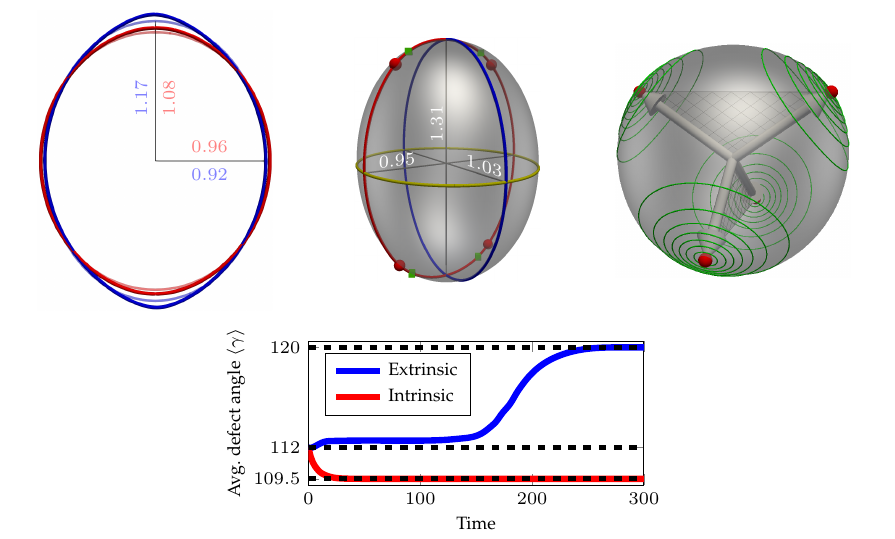}
		}
		{
			\includegraphics[width=0.9\textwidth]{snapshots/fig_9.png}
		}
    }
    {
		\ifthenelse{\boolean{forArxiv}}
		{
			\begin{minipage}{0.75\textwidth}
		}
		{
			\begin{minipage}{\textwidth}
		}
				\centering
				\begin{minipage}{\textwidth}
					\centering
					\hfill
					\begin{minipage}{0.24\textwidth}
						\centering
						\begin{tikzpicture}
							\node (pic) at (0,0) {\includegraphics[width=\textwidth]{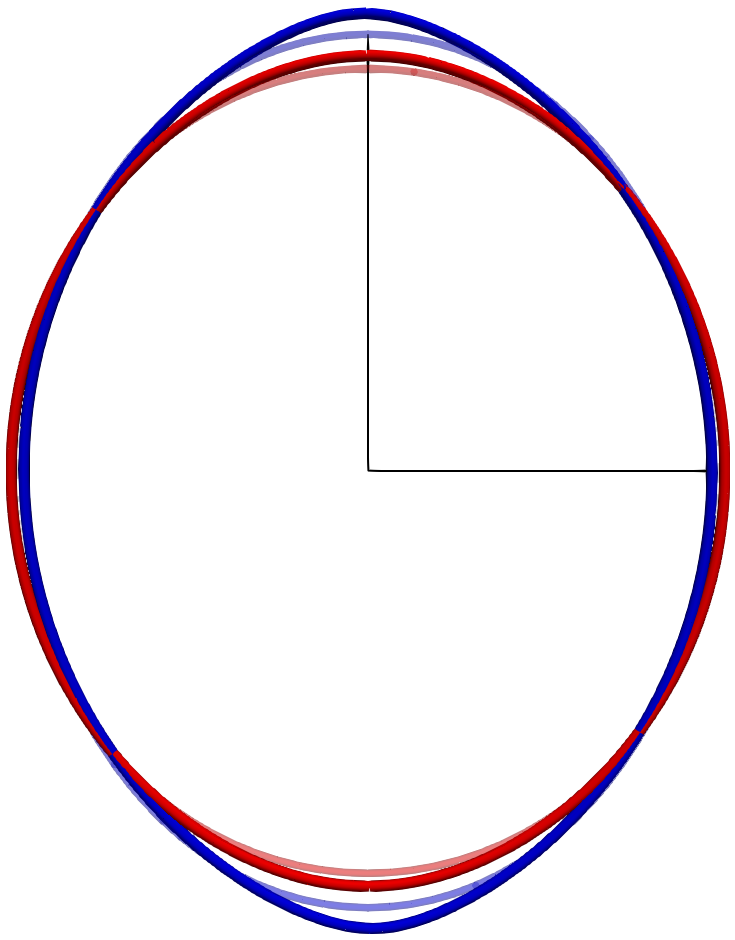}};
							\draw (0.75,0.0) node[red!50,anchor=south] {\scriptsize $0.96$};
							\draw (0.0,0.85) node[red!50,anchor=north,rotate=90] {\scriptsize $1.08$};
							\draw (0.75,0.0) node[blue!50,anchor=north] {\scriptsize $0.92$};
							\draw (0.0,0.85) node[blue!50,anchor=south,rotate=90] {\scriptsize $1.17$};
						\end{tikzpicture}
					\end{minipage}
					\hfill
					\begin{minipage}{0.19\textwidth}
						\centering
						\begin{tikzpicture}
							\node (pic) at (0,0) {\includegraphics[width=\textwidth]{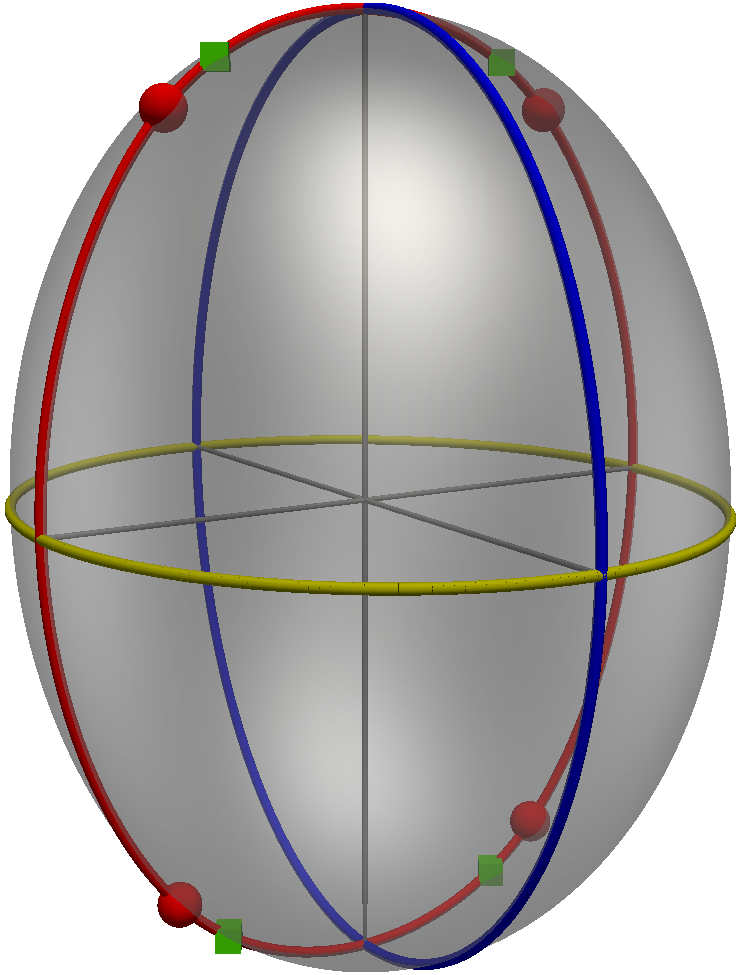}};
							\draw (-0.7,-0.2) node[white,anchor=south,rotate=0,xslant=0.0,yslant=0.12]   {\scriptsize $0.95$};
							\draw (0.04,0.5) node[white,anchor=south,rotate=90]   {\scriptsize $1.31$};
							\draw (0.5,-0.27) node[white,anchor=south,xslant=0.0,yslant=-0.35] {\scriptsize $1.03$};
						\end{tikzpicture}
					\end{minipage}
					\hfill
					\begin{minipage}{0.24\textwidth}
						\centering
						\includegraphics[width=\textwidth]{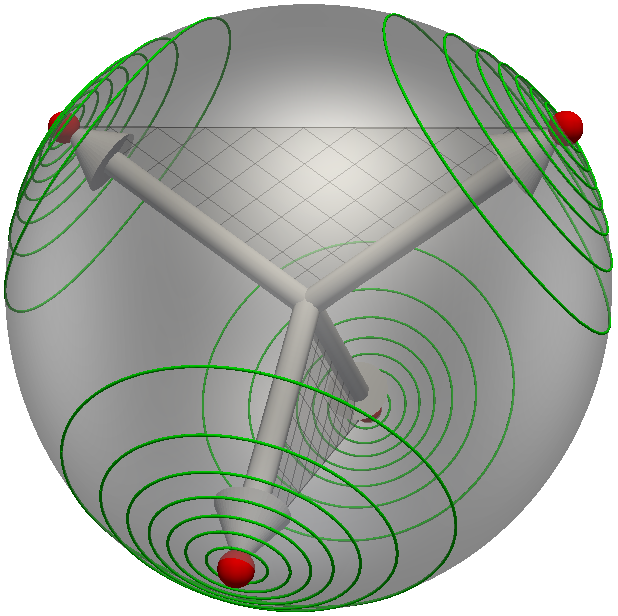}
					\end{minipage}
					\hfill
					\hfill
				\end{minipage}
				\begin{minipage}{0.49\textwidth}
					\centering
%
%
\begin{tikzpicture}

\begin{axis}[%
width=0.7\textwidth,
height=0.3\textwidth,
at={(0,0)},
scale only axis,
separate axis lines,
every outer x axis line/.append style={black},
every x tick label/.append style={font=\color{black}},
xmin=0,
xmax=300,
xtick={ 0, 100, 200, 300},
xlabel={$\mbox{Time}$},
every outer y axis line/.append style={black},
every y tick label/.append style={font=\color{black}},
ymin=109,
ymax=120.5,
ytick={ 109.5, 112, 120},
ylabel={$\mbox{Avg. defect angle }\langle\gamma\rangle$},
axis background/.style={fill=white},
x label style={font=\scriptsize,at={(axis description cs:0.5,0.1)}},
y label style={font=\scriptsize,at={(axis description cs:0.05,0.5)}},
yticklabel style = {font=\scriptsize},
xticklabel style = {font=\scriptsize},
legend style={font=\scriptsize,at={(0.05,0.92)},anchor=north west,legend columns=1,legend cell align=left,align=left,draw=black}, 
]
\addplot [color=blue,solid,line width=2.5pt]
  table[row sep=crcr]{%
0	112\\
2	111.96265803329\\
4	112.016226973265\\
6	112.109279648101\\
8	112.202175315885\\
10	112.296735811231\\
12	112.383836826135\\
14	112.449131554401\\
16	112.488868023442\\
18	112.509889094277\\
20	112.519498272864\\
22	112.524847280044\\
24	112.529596778914\\
26	112.533916374828\\
28	112.537823888023\\
30	112.541337138734\\
32	112.544473947196\\
34	112.547252133645\\
36	112.549689518317\\
38	112.551803921447\\
40	112.55361316327\\
42	112.555135064023\\
44	112.55638744394\\
46	112.557388123258\\
48	112.558154922211\\
50	112.558705661036\\
52	112.559058159968\\
54	112.559230239243\\
56	112.559239719095\\
58	112.559104419762\\
60	112.558842161477\\
62	112.558470764477\\
64	112.558008048998\\
66	112.557471835275\\
68	112.556879943543\\
70	112.556250194038\\
72	112.555600406996\\
74	112.554948402651\\
76	112.554312001241\\
78	112.553709023\\
80	112.553157288164\\
82	112.552674616968\\
84	112.552278829648\\
86	112.551987746439\\
88	112.551819187578\\
90	112.551790973299\\
92	112.551920923838\\
94	112.552233514058\\
96	112.552779837326\\
98	112.553617641636\\
100	112.554804674984\\
102	112.556398685364\\
104	112.55845742077\\
106	112.561038629197\\
108	112.564200058639\\
110	112.567999457091\\
112	112.572494572547\\
114	112.577765304879\\
116	112.583980161462\\
118	112.591329801551\\
120	112.600004884399\\
122	112.610196069258\\
124	112.622094015382\\
126	112.635777781666\\
128	112.650880025574\\
130	112.666927168695\\
132	112.683569015762\\
134	112.700578754649\\
136	112.717734937717\\
138	112.734864923936\\
140	112.75291862437\\
142	112.77396850217\\
144	112.800135827103\\
146	112.833541868934\\
148	112.876307897427\\
150	112.93055518235\\
152	112.998415177045\\
154	113.08225355717\\
156	113.184670220698\\
158	113.306788830996\\
160	113.443787378692\\
162	113.589740857646\\
164	113.740257947361\\
166	113.897259868163\\
168	114.086384318021\\
170	114.326523178554\\
172	114.585868588338\\
174	114.82146510478\\
176	115.025486440874\\
178	115.225168627009\\
180	115.447737769831\\
182	115.717492706985\\
184	116.026695652485\\
186	116.335639034905\\
188	116.612915569533\\
190	116.86587908118\\
192	117.108111553477\\
194	117.339346495782\\
196	117.555953603281\\
198	117.756563572129\\
200	117.942068099453\\
202	118.113457186773\\
204	118.271720835606\\
206	118.41784904747\\
208	118.552831823884\\
210	118.677659166365\\
212	118.793256235827\\
214	118.900288830765\\
216	118.999357909067\\
218	119.091064428623\\
220	119.176009347323\\
222	119.254793623056\\
224	119.328018213711\\
226	119.396212969141\\
228	119.459623307046\\
230	119.518423537091\\
232	119.572787968939\\
234	119.622890912254\\
236	119.668906676699\\
238	119.711009571939\\
240	119.749373907637\\
242	119.784173993457\\
244	119.815584139063\\
246	119.843778654118\\
248	119.868931848286\\
250	119.891218031232\\
252	119.910811512618\\
254	119.927886602108\\
256	119.942617609367\\
258	119.955178844057\\
260	119.965744615844\\
262	119.97448923439\\
264	119.981587009359\\
266	119.987212250415\\
268	119.991539267222\\
270	119.994742369444\\
272	119.996995866743\\
274	119.998474068785\\
276	119.999351285233\\
278	119.99980182575\\
280	120\\
282	120.000093667924\\
284	120.000124890565\\
286	120.000109279244\\
288	120.000062445282\\
290	120\\
292	119.999937554718\\
294	119.999890720756\\
296	119.999875109435\\
298	119.999906332076\\
300	120\\
};
\addlegendentry{$\mbox{Extrinsic}$};

\addplot [color=red,solid,line width=2.5pt]
  table[row sep=crcr]{%
0	112\\
2	111.281913294097\\
4	110.812121541723\\
6	110.500749909834\\
8	110.264153240116\\
10	110.066094919439\\
12	109.939362912265\\
14	109.843741607047\\
16	109.764508936954\\
18	109.700107341895\\
20	109.64897412804\\
22	109.60942852551\\
24	109.579671688371\\
26	109.557899636952\\
28	109.54230839158\\
30	109.531093972583\\
32	109.522687341688\\
34	109.516459226221\\
36	109.512015294911\\
38	109.508961216485\\
40	109.506902659669\\
42	109.505445293191\\
44	109.504257303517\\
46	109.503256948065\\
48	109.502425001993\\
50	109.501742240458\\
52	109.501189438618\\
54	109.50074737163\\
56	109.500396814652\\
58	109.50011854284\\
60	109.499893666152\\
62	109.49971099493\\
64	109.499567039903\\
66	109.499458646597\\
68	109.499382660541\\
70	109.499335927261\\
72	109.499315292285\\
74	109.499317601141\\
76	109.499339699357\\
78	109.499378432459\\
80	109.499430645976\\
82	109.499493185434\\
84	109.499562896362\\
86	109.499636624286\\
88	109.499711214735\\
90	109.499783513235\\
92	109.499850365315\\
94	109.499908616502\\
96	109.499955112322\\
98	109.499986698305\\
100	109.500000219977\\
150	109.5\\
200	109.5\\
250	109.5\\
300	109.5\\
};
\addlegendentry{$\mbox{Intrinsic}$};

\addplot [color=black,dashed,line width=2.0pt,forget plot]
  table[row sep=crcr]{%
0	120\\
300	120\\
};
\addplot [color=black,dashed,line width=2.0pt,forget plot]
  table[row sep=crcr]{%
0	109.5\\
300	109.5\\
};
\addplot [color=black,dashed,line width=2.0pt,forget plot]
  table[row sep=crcr]{%
0	112\\
300	112\\
};
\end{axis}
\end{tikzpicture}%
				\end{minipage}
			\end{minipage}
	}
    \caption{Top left: Sliced equilibrium shapes obtained with the extrinsic (blue) and the intrinsic (red) \frankOseenHelfrich\ model \eqref{eq:gradientflow:p:vnor},\eqref{eq:gradientflow:p:dirf} and \eqref{eq:gradientflow:p:vnor:intrinsic},\eqref{eq:gradientflow:p:dirf:intrinsic}. The fitted ellipsoids as well as their axes parameters are marked in light blue and in light red, respectively. In both cases the defect positions correspond with the maximum curvature points. However, the surface is stronger distorted in these points in the extrinsic case. On the other side in regions of minimal curvature (equator on short axis) the extrinsic case tends to be more flat. Top center: Equilibrium shape obtained with the extrinsic \landauDeGennesHelfrich\ model \eqref{eq:gradientflow:q:vnor}, \eqref{eq:gradientflow:q:qten}. Red spheres are the defect positions and green squares are the positions of maximum Gaussian curvature. The distances of the surface from the origin are marked to highlight the asymmetry. All defects and maximal curvature points are located on one plane (red curve). Top right: Equilibrium shape obtained with the intrinsic \landauDeGennesHelfrich\ model \eqref{eq:gradientflow:q:vnor:intrinsic}, \eqref{eq:gradientflow:q:qten:intrinsic}. The Gaussian curvature is shown as green contour lines. The defect positions correspond with the maximum Gaussian curvature points. They are arranged in a tetrahedral configuration. Initial conditions correspond to \autoref{fig:results:qtensor2}. Bottom: Averaged defect angle $\langle\gamma\rangle$ (see \autoref{fig:director_response:angle}) against time for the extrinsic and the intrinsic \landauDeGennesHelfrich\ model. $112^\circ$ is the equilibrium average for the considered ellipsoid at the initial configuration, $109.5^\circ$ and $120^\circ$ correspond to the tetrahedral and planar defect configuration, respectively.}
    \label{fig:results:intrinsic}
\end{figure}

Our results add to the importance of intrinsic and extrinsic contributions. 
The derived thermodynamically consistent models for liquid crystals on deformable surfaces and the considered numerical approach to solve these highly nonlinear equations show the tight interplay of these curvature contributions with defect arrangements and shape changes. 
The results go well beyond purely intrinsic energy minimization approaches. 
These phenomena do not only change the behaviour quantitatively, they qualitatively lead to new solutions. The symmetric tetrahedral arrangement, which is the basis for proposed self-assembly processes of colloidal particles with tetrahedral structure \cite{Nelson_NL_2002}, is not stable if the surface is deformable. The tendency towards asymmetric shapes will be further enhanced if more realistic liquid crystal models, beyond the one-constant approximation, are used \cite{Shinetal_PRL_2008}. This will not only influence applications in soft condensed matter based on liquid crystals on flexible curved substrates, but also has implications in biophysics, concerning morphological changes during development and the design of bio-inspired materials that are capable of self-organisation.

\ifthenelse{\boolean{forArxiv}}
{
}
{
	\begin{description}[leftmargin=0cm]
	\item[Data accessibility statement.] This work does not have any experimental data. Datasets and code are made available at https:/$\!\!$/gitlab.mn.tu-dresden.de/sourcecode/surface-lc \cite{data}.
	\item[Authors' contributions.] IN derived the models. SR implemented the models and performed finite element simulations. IN, SR and AV conceived and interpreted the numerical experiments. All authors contributed to a critical discussion of the derived model and numerical data and participated in writing the manuscript, which AV finalized. AV coordinated the project.
	\item[Competing interests statement.] We have no competing interests.
	\item[Funding.] AV acknowledges financial support from DFG through Vo899/19 and FOR3013. We further acknowledge computing resources provided by JSC under grant HDR06 and ZIH/TU Dresden.
	\end{description}
}

\appendix

\section*{Appendix}
We here provide all necessary details concerning the model and the numerical approach. While an analytical investigation of the model and the numerical approach is beyond the scope of this paper, we thereby restrict ourselves to experimental studies. When convenient, we use Ricci calculus \cite{Schouten_1954} with lowercase Latin indices.
Please note that for ease of terminology we exert parenthesis of symmetrizing,
\eg\ $\tensor{\rtenC}{^{(i}_j^{k)}_l} = \frac{1}{2}(\tensor{\rtenC}{^{i}_j^{k}_l}+ \tensor{\rtenC}{^{k}_j^{i}_l})$ for a mixed 4-tensor 
$\rten\in\tangentT{^1_1^1_1}\surf$.
To avoid confusion these parenthesis are only used pairwise and concerning always neighboring indices in the same height of indexing.
Another feature concerning the symmetry is that we write one index below the other directly for mixed symmetric 2-tensors,
\eg\ $\shopC^i_j := \tensor{\shopC}{^i_j} = \tensor{\shopC}{_j^i}$, where the latter identity holds, since $\shopC^{ij}=\shopC^{(ij)}$.

\section{Derivations}
\label{sec:app:one}

\subsection{Shape Variations}
\label{sec:shapevariations}

In the following we investigate some surface quantities and operators of the surface under perturbation.
For consistence with notation of differential geometry 
we write $\partial_i := \frac{\dup\phantom{y^i}}{\dup y^i}$ for the derivative along local coordinates.
This is motivated in the sense that we can use $\partial_i$ as well as partial derivative, \ie\ for a function 
$f=f(\xvar)$ and given $\xvar$ we find a function $f_{\xvar}$ such that $f(\xvar(t,y^1,y^2),y^1,y^2))=f_{\xvar}(t,y^1,y^2)$ and it holds
$\partial_i f= \partial_i f_{\xvar} = \frac{\partial f_{\xvar}}{\partial y^i}$. 
Moreover, we observe that the surface variation and $\partial_i$ commute, 
\ie\ $\dsurf\circ\partial_i = \partial_i\circ\dsurf$.
For vector fields $\dirf = \dirfC^i\partial_i\para\in\tangentT{^1}\surf$ and 2-tensor fields $\rten = \rtenC^{ij}\partial_i\para\otimes\partial_j\para\in\tangentT{^2}\surf$
the surface variation $\dsurf$ gives
\begin{align}
    \label{eq_codsurfdirf}
    \proj_{\surf}\dsurf\dirf &= \left( \dsurf\dirfC^i \right)\partial_i\para + \dirfC^i\proj_{\surf}\partial_i\left( \dxvar\normal \right) = \left\{ \dsurf\dirfC^i \right\} - \shop\dirf\dxvar \formComma \\
    \proj_{\surf}\dsurf\rten &= \left( \dsurf\rtenC^{ij} \right)\partial_i\para\otimes\partial_j\para + \rtenC^{ij} \proj_{\surf}\left[ \partial_i\left( \dxvar\normal \right)\otimes\partial_j\para + \partial_i\para\otimes\partial_j\left( \dxvar\normal \right)\right] \notag \\
    \label{eq_codsurf2ten}
    &= \left( \dsurf\rtenC^{ij} \right) \partial_i\para\otimes\partial_j\para - \left( \shop\rten + \rten\shop \right)\dxvar \formPeriod
\end{align}
Obviously, symmetric 2-tensor fields are closed under $\proj_{\surf}\dsurf:\tangentT{^2}\surf\rightarrow\tangentT{^2}\surf$.
To show the closeness \wrt\ Q tensor fields $ \qten\in\qspace $, we observe that $g_{ij}\dsurf\qtenC^{ij} = \dsurf\Tr\qten - \qtenC^{ij}\dsurf g_{ij} = 2\left\langle \shop,\qten \right\rangle\dxvar$, see eq. \eqref{eq_dsurfgij} below.
Therefore, with eq. \eqref{eq_codsurf2ten} and $\Tr[\partial_i\para\otimes\partial_j\para] = g_{ij}$ we obtain
$\Tr\proj_{\surf}\dsurf\qten = \left( \dsurf\qtenC^{ij}- \left[ \shop\qten + \qten\shop \right]^{ij} \right)g_{ij} = 2\left\langle \shop,\qten \right\rangle\dxvar - 2\left\langle \shop,\qten \right\rangle\dxvar = 0$ and it finally holds
\begin{align}
    \label{eq_codsurfqten}
    \forall\qten\in\qspace:\quad\quad\proj_{\qspace}\dsurf\qten &= \left\{ \dsurf\qtenC^{ij} \right\} - \left( \shop\qten + \qten\shop \right)\dxvar
                                                                    \in\qspace\formPeriod
\end{align}
For the components of the metric tensor we get
\begin{align}
    \label{eq_dsurfgij}
    \dsurf g_{ij}  &= \dsurf\left\langle \partial_i\para,\partial_j\para \right\rangle = \left\langle \partial_i\left( \dxvar\normal \right), \partial_j\para \right\rangle + \left\langle \partial_i\para, \partial_j\left( \dxvar\normal \right) \right\rangle = -2\left\langle \normal, \partial_i\partial_j\para \right\rangle\dxvar \notag \\
    &= -2\shopC_{ij} \dxvar
\end{align} 
and consequently for the components of the inverse metric tensor
\begin{align}
    \dsurf g^{ij} &= 2\shopC^{ij} \dxvar \formComma \label{eq_dsurfgijinv}
\end{align}
which follows from $\dsurf g_{ij} = \dsurf ( g_{ik}g_{jl} g^{kl}) = 2\dsurf g_{ij} +  g_{ik}g_{jl}\dsurf g^{kl}$ by using the product rule.
Using the identity $\mathcal{D}\det\g = \det\g\Tr\{\mathcal{D}g_{ij}\}$ for derivatives
$\mathcal{D}:\tangent^{0}\surf \rightarrow \tangent^{0}\surf$ acting on scalar-valued functions, we obtain
\begin{align*}
    \dsurf\sqrt{\det\g}	&= \frac{1}{2\sqrt{\det\g}}\dsurf\det\g = \frac{\sqrt{\det\g}}{2}\Tr\{\dsurf g_{ij}\} = - \sqrt{\det\g}\meanc\dxvar \formComma
\end{align*}
which finally results in
\begin{align}
    \label{eq:transport}
    \dsurf \int_{\surf} f \mu &= \int_{\surf} \dsurf f - f\meanc\dxvar\mu
\end{align}
for scalar-valued functions $ f $, similarly to the well-known transport theorem.
Next, we consider the Christoffel symbols which are needed for covariant differentiation. 
With eq. \eqref{eq_dsurfgij}, the surface variation of the first kind Christoffel symbols reads
\begin{align*}
    \dsurf\Gamma_{ijl} &=\frac{1}{2}\dsurf\left( \partial_i g_{jl} + \partial_j g_{il} - \partial_l g_{ij}\right) = \frac{1}{2}\left( \partial_i \dsurf g_{jl} + \partial_j \dsurf g_{il} - \partial_l \dsurf g_{ij}\right)\notag\\
    &= -\left( \dxvar_{|i}\shopC_{jl} + \dxvar_{|j}\shopC_{il} - \dxvar_{|l}\shopC_{ij}\right) - \dxvar\left( \shopC_{jl|i} + \shopC_{il|j} - \shopC_{ij|l} + 2 \Gamma_{ij}^{k}\shopC_{kl}\right)\notag\\
    &= \dxvar_{|l}\shopC_{ij} - \dxvar_{|i}\shopC_{jl} - \dxvar_{|j}\shopC_{il} -\dxvar\left( \shopC_{jl|i} + 2 \Gamma_{ij}^{k}\shopC_{kl} \right) \notag\\
    &= \dxvar_{|l}\shopC_{ij} - \dxvar_{|j}\shopC_{il} - \left( \dxvar\shopC_{jl} \right)_{|i} - 2\dxvar\Gamma_{ij}^{k}\shopC_{kl}  \formComma
\end{align*}
where we have used that the shape operator is curl-free, \ie\ $\shopC_{il|j} - \shopC_{ij|l} = 0$.
Furthermore, with eq. \eqref{eq_dsurfgijinv}, it holds for the second kind Christoffel symbols that 
\begin{align*}
    \dsurf\Gamma_{ij}^{k} &= \dsurf\left( g^{kl}\Gamma_{ijl} \right) = 2\shopC^{kl}\Gamma_{ijl}\dxvar + g^{kl}\dsurf\Gamma_{ijl} =\dxvar^{|k}\shopC_{ij}  - \dxvar_{|j}\shopC_{i}^k - \left(\shopC_j^k\dxvar\right)_{|i} \formPeriod
\end{align*}
With these evaluations we obtain for an independent vector field $\dirf\in\tangent^{1}\surf$, \ie\ $\proj_{\surf}\dsurf\dirf = 0$ or $\dsurf\dirfC^{i} = [ \shop\dirf ]^{i}\dxvar$, respectively, see eq. \eqref{eq_codsurfdirf}, that
\begin{align*}
    \dsurf\tensor{\dirfC}{^i_{|k}} &=\partial_k\dsurf\dirfC^{i} + \Gamma_{kl}^i\dsurf\dirfC^{l} + \dirfC^{l}\dsurf\Gamma_{kl}^i = \left( \shopC^i_l\dirfC^l\dxvar \right)_{|k} - \dirfC^{l}\left( \left(\shopC_l^i\dxvar\right)_{|k} + \dxvar_{|l}\shopC_{k}^i - \dxvar^{|i}\shopC_{kl}\right)\notag\\
    &= \shopC^i_l\tensor{\dirfC}{^l_{|k}}\dxvar + \dirfC^{l}\left( \dxvar^{|i}\shopC_{kl} - \dxvar_{|l}\shopC_{k}^i\right)
\end{align*}
for the mixed components of $\GradSurf\dirf$. 
In a fully contravariant sense this reads 
\begin{align*}
    \dsurf\dirfC^{i|j} &= g^{jk}\dsurf\tensor{\dirfC}{^i_{|k}} + \tensor{\dirfC}{^i_{|k}}\dsurf g^{jk} = \left( \dirfC^{l|j}\shopC^{i}_{l} + 2 \dirfC^{i|l}\shopC^{j}_{l} \right)\dxvar + \dirfC^{l}\left( \dxvar^{|i}\shopC^j_l - \dxvar_{|l}\shopC^{ij}\right) \formPeriod
\end{align*}
In the embedding space we explicitly obtain
\begin{align*}
    \proj_{\surf}\dsurf\GradSurf\dirf &= \left( \dsurf\dirfC^{i|j} \right)\partial_i\para\otimes\partial_j\para  + \dirfC^{i|j}\left( \partial_i\normal\otimes\partial_j\para + \partial_i\para\otimes\partial_j\normal\right)\dxvar
\end{align*} 
for $\GradSurf\dirf = \dirfC^{i|j}\partial_i\para\otimes\partial_j\para\in\tangentT{^2}\surf$, which finally gives
\begin{align}
    \label{eq_distgraddirf}
    \forall\dirf\in\tangent^{1}\surf \text{ with } \proj_{\surf}\dsurf\dirf = 0:\quad
        \proj_{\surf}\dsurf\GradSurf\dirf &= \left( \GradSurf\dirf \right)\shop\dxvar + \GradSurf\dxvar\otimes\shop\dirf - \left( \levicivita{\dirf}{\dxvar} \right)\shop
\end{align}
where $\partial_i\normal = -\shopC^k_i\partial_k\para$ was used.
Similarly, for an independent Q tensor field $\qten\in\qspace$, \ie\ $\proj_{\qspace}\dsurf\qten = 0$
or $\dsurf\qtenC^{ij} = \left[ \shop\qten + \qten\shop \right]^{ij}\dxvar$, respectively, see \eqref{eq_codsurfqten}, we obtain
\begin{align*}
    \dsurf\tensor{\qtenC}{^{ij}_{|k}} &=\partial_k\dsurf\qtenC^{ij} + 2\Gamma_{kl}^{(i}\dsurf\qtenC^{j)l} + 2\qtenC^{l(j}\dsurf\Gamma_{kl}^{i)} \notag\\
    &= 2\left( \shopC^{(i}_l\qtenC^{j)l}\dxvar \right)_{|k} - 2\qtenC^{l(j}\left( \left(\shopC_l^{i)}\dxvar\right)_{|k} + \dxvar_{|l}\shopC_{k}^{i)} - \dxvar^{|i)}\shopC_{kl}\right)\notag\\
    &= 2\shopC^{(i}_l\tensor{\qtenC}{^{j)l}_{|k}}\dxvar + 2\qtenC^{l(j}\left( \dxvar^{|i)}\shopC_{kl} - \dxvar_{|l}\shopC_{k}^{i)}\right)
\end{align*}
for the mixed components of $\GradSurf\qten$.
In a fully contravariant sense this reads
\begin{align*}
    \dsurf\qtenC^{ij|k} &= g^{kl}\dsurf\tensor{\qtenC}{^{ij}_{|l}} + \tensor{\qtenC}{^{ij}_{|l}}\dsurf g^{kl} = 2 \left( \shopC^{(i}_{l}\qtenC^{j)l|k} + \qtenC^{ij|l}\shopC^{k}_{l} \right)\dxvar + 2\qtenC^{l(j}\left( \dxvar^{|i)}\shopC_{l}^k - \dxvar_{|l}\shopC^{i)k}\right) \formPeriod
\end{align*}
In the embedding space we explicitly obtain
\begin{align*}
    \proj_{\surf}\dsurf\GradSurf\qten &= \left( \dsurf\qtenC^{ij|k} \right)\partial_i\para\otimes\partial_j\para\otimes\partial_k\para \\
    &\quad  + \qtenC^{ij|k}\left( \partial_i\normal\otimes\partial_j\para\otimes\partial_k \para + \partial_i\para\otimes\partial_j\normal\otimes\partial_k\para + \partial_i\para\otimes\partial_j\para\otimes\partial_k\normal\right)\dxvar
\end{align*}
for $\GradSurf\qten = \qtenC^{ij|k}\partial_i\para\otimes\partial_j\para\otimes\partial_k\para\in\tangentT{^3}\surf$, which finally gives 
\begin{align}
    \label{eq_distgradqten}
    \forall\qten\in\qspace \text{ with } \proj_{\qspace}\dsurf\qten = 0:\quad
        \left[\proj_{\surf}\dsurf\GradSurf\qten\right]^{ijk}
            &= \shopC^{k}_{l}\left( \qtenC^{ij|l}\dxvar + 2\qtenC^{l(i}\dxvar^{|j)} \right) - 2\shopC^{k(i}\qtenC^{j)l}\dxvar_{|l}
\end{align}
where again $ \partial_i\normal = -\shopC^k_i\partial_k\para $ was used.

To determine the behavior of extrinsic curvature quantities under surface variations we first consider the normal vector $\normal$.
The normal part of $\dsurf\normal$ vanishes and the tangential part is determined by
$\left\langle \dsurf\normal, \partial_i\para \right\rangle = -\left\langle \normal, \partial_i\dsurf\para \right\rangle = -\left\langle \normal, \dxvar\partial_i\normal + \normal\partial_i\dxvar \right\rangle = -\partial_i\dxvar$ and thus
\begin{align}
    \label{eq:dsurNormal}
    \dsurf\normal &= -\GradSurf\dxvar \formPeriod
\end{align}
For the covariant components of the shape operator we obtain
\begin{align*}
    \dsurf\shopC_{ij} &= -\dsurf\left\langle \partial_i\para, \partial_j\normal \right\rangle = -\left\langle \partial_i(\dxvar\normal), \partial_j\normal \right\rangle + \left\langle \partial_i\para, \partial_j\GradSurf\dxvar \right\rangle \notag \\
    &=-\left\langle \partial_i\normal, \partial_j\normal \right\rangle\dxvar + \left\langle \partial_i\para, (\partial_j\dxvar^{|k})\partial_k\para + \dxvar^{|k}\partial_j\partial_k\para \right\rangle \notag \\
    &=-\left[ \shop^2 \right]_{ij}\dxvar + g_{ik}\partial_j\dxvar^{|k} + \Gamma_{jki}\dxvar^{|k} = -\left[ \shop^2 \right]_{ij}\dxvar + \dxvar_{|i|j} = \left[ \GradSurf^2\dxvar - \shop^2\dxvar \right]_{ij}
\end{align*}
and for the contravariant components
\begin{align*}
    \dsurf\shopC^{ij} &= \left[ \GradSurf^2\dxvar + 3\shop^2\dxvar \right]^{ij}
\end{align*}
where $\dsurf\shopC^{ij} = \dsurf(g^{ik}g^{jl}\shopC_{kl}) $ and eq. \eqref{eq_dsurfgijinv} was used.
Furthermore, for the embedded tensor $\shop = \shopC^{ij}\partial_i\para\otimes\partial_j\para$ we get
\begin{align}
    \label{eq_dsurfshop}
    \proj_{\surf}\dsurf\shop &= \GradSurf^2\dxvar + 3\shop^2\dxvar + \shopC^{ij}\left( \partial_i\normal\otimes\partial_j\para + \partial_i\para\otimes\partial_j\normal \right)\dxvar = \GradSurf^2\dxvar + \shop^2\dxvar \formPeriod
\end{align}
Finally, consider the mean curvature $\meanc$, for which $\dsurf\meanc = \shopC_{ij}\dsurf g^{ij} + g^{ij}\dsurf\shopC_{ij}$ can be observed, 
we obtain
\begin{align}
    \label{eq:p:eq_distmeanc}
    \dsurf\meanc &= \laplaceBeltrami\dxvar + \left\| \shopC \right\|^2\dxvar = \laplaceBeltrami\dxvar + \left( \meanc^2 - 2\gaussc \right)\dxvar \formPeriod
\end{align}

\subsection{Functional Derivatives}
\label{sec:funcderivatives}
    
With the tools from above we are able to compute the functional derivatives of free energies depending on independent scalar field $\xvar\in\tangentT{^0}\surf$ and vector field $\dirf\in\tangentT{^1}\surf$, \ie\ $\proj_{\surf}\dsurf\dirf=0$ or Q tensor field $\qten\in\qspace$, \ie\ $\proj_{\qspace}\dsurf\qten=0$.
Note that we consider a surface without boundary, which drastically reduces the complexity of the derivations.
    
\subsubsection{Helfrich Energy}
With eqs. \eqref{eq:transport} and \eqref{eq:p:eq_distmeanc} the \helfrich\ energy $\helfrichEnergy$ yields, \cf\ \cite[pp.\ 280--282]{Willmore_book_1996},
\begin{align*}
    \dsurf\helfrichEnergy &=\alpha\int_{\surf}\meanc\dsurf\meanc - \frac{\meanc^3}{2}\dxvar\mu = \alpha\int_{\surf} \meanc\laplaceBeltrami\dxvar + \meanc\left( \frac{\meanc^2}{2} - 2\gaussc \right)\dxvar\mu \formPeriod
\end{align*}
Thus, by using integration by parts we get 
\begin{align}\label{eq:dhelfrichdxi}
    \frac{\delta\helfrichEnergy}{\dxvar} &= \alpha\left(\laplaceBeltrami\meanc + \meanc\left( \frac{\meanc^2}{2} - 2\gaussc \right)\right)\formPeriod
\end{align}
    
\subsubsection{\frankOseen\ Energy}
The \frankOseen\ energy $\frankOseenEnergy$ is split into an intrinsic and an extrinsic part, \ie
\begin{align*}
    \PotEP_{\text{I}} =  \PotEP_{\text{I}}[\xvar, \dirf] = \frac{K}{2}\int_{\surf} \left\| \GradSurf\dirf \right\|^2\mu \formComma \qquad
    \PotEP_{\text{E}} = \PotEP_{\text{E}}[\xvar, \dirf] =\frac{K}{2}\int_{\surf} \left\| \shop\dirf \right\|^2\mu \formComma
\end{align*}
respectively. The intrinsic \frankOseen\ energy measures the distortion of the vector field $\dirf\in\tangent^{1}\surf$. With eqs. \eqref{eq:transport}, \eqref{eq_distgraddirf} and the fact that $\GradSurf\dirf\in\tangentT{^2}\surf$ is tangential by construction of the connection $\levicivitaPlain$, we get
\begin{align*}
    \dsurf\PotEP_{\text{I}} &= K\int_{\surf} \left\langle \GradSurf\dirf, \proj_{\surf}\dsurf\GradSurf\dirf \right\rangle - \frac{\meanc}{2}\left\| \GradSurf\dirf \right\|^2\dxvar\mu \\
    &= K\int_{\surf}\left\langle (\GradSurf\dirf)^{T}\GradSurf\dirf,\shop \right\rangle\dxvar + \left\langle \levicivita{\shop\dirf}{\dirf},\GradSurf\dxvar \right\rangle - \levicivita{\shop}{\dirf}\levicivita{\dirf}{\dxvar} - \frac{\meanc}{2}\left\| \GradSurf\dirf \right\|^2\dxvar\mu \\
    &= K\int_{\surf}\left\langle \levicivita{\shop\dirf}{\dirf} - \left(\levicivita{\shop}{\dirf}\right)\dirf, \GradSurf\dxvar \right\rangle + \left\langle (\GradSurf\dirf)^{T}\GradSurf\dirf, \proj_{\qspace}\shop \right\rangle\dxvar\mu
\end{align*}
and therefore
\begin{align*}
    \frac{\delta\PotEP_{\text{I}}}{\dxvar} = K\left( \DivSurf \left( \left( \levicivita{\shop}{\dirf} \right)\dirf - \levicivita{\shop\dirf}{\dirf} \right) + \left\langle (\GradSurf\dirf)^{T}\GradSurf\dirf, \proj_{\qspace}\shop \right\rangle \right), \qquad
    \frac{\delta\PotEP_{\text{I}}}{\delta\dirf} = -K\laplaceBochner\dirf\formPeriod
\end{align*}
Here, we have used $\levicivita{\shop}{\dirf} = \left\langle \GradSurf\dirf, \shop \right\rangle$, the orthogonal Q tensor projection $\proj_{\qspace}\shop = \shop - \frac{\meanc}{2}\g$ of the shape operator and the Bochner-Laplacian $\laplaceBochner\dirf=\GradSurf^*\GradSurf\dirf = \DivSurf\GradSurf\dirf$ of $\dirf$ \wrt\ the connection $\levicivitaPlain$. The extrinsic part is a leftover of the distortion energy in a thin film with homogeneous Neumann conditions at the upper and lower boundaries, see \cite{Nestleretal_JNS_2018}. 
We obtain by using eqs. \eqref{eq:transport} and \eqref{eq_dsurfshop}
\begin{align*}
    \dsurf\PotEP_{\text{E}} &= K\int_{\surf} \left\langle \shop\dirf, \shop\proj_{\surf}\dsurf\dirf + \left( \proj_{\surf}\dsurf\shop \right)\dirf \right\rangle - \frac{\meanc}{2}\left\| \shop\dirf \right\|^2\dxvar\mu \\
    &= K\int_{\surf} \left\langle \shop\dirf, \left( \GradSurf^2\dxvar \right)\dirf \right\rangle + \left\langle \shop\dirf,\shop^2\dirf \right\rangle\dxvar - \frac{\meanc}{2}\left\| \shop\dirf \right\|^2\dxvar\mu \\
    &= K\int_{\surf} -\left\langle \left( \DivSurf\dirf \right)\shop\dirf + \levicivita{\dirf}{\left(\shop\dirf\right)},\GradSurf\dxvar \right\rangle + \left\langle \shop^2\proj_{\qspace}\shop,\dirf\otimes\dirf \right\rangle\dxvar\mu
\end{align*}
and thus
\begin{align*}
    \frac{\delta\PotEP_{\text{E}}}{\dxvar} = K\left( \DivSurf \left( \left( \DivSurf\dirf \right)\shop\dirf + \levicivita{\dirf}{\left(\shop\dirf\right)} \right) + \left\langle \shop^2\proj_{\qspace}\shop,\dirf\otimes\dirf \right\rangle \right), \qquad
    \frac{\delta\PotEP_{\text{E}}}{\delta\dirf} = K\shop^2\dirf \formPeriod
\end{align*}
The functional derivatives of the intrinsic and the extrinsic \frankOseen\ energies fit well together.
With $\mathcal{L}_{\dirf}(\shop\dirf) =  \levicivita{\dirf}{\left(\shop\dirf\right)} -  \levicivita{\shop\dirf}{\dirf}$, the fact that the orthogonal projection in the product $\shop^2\proj_{\qspace}\shop$ is permutable, \ie\ $\shop^2\proj_{\qspace}\shop = \shop\left( \proj_{\qspace}\shop \right)\shop$, and Proposition \autoref{prop_divterm} we obtain for the sum of these energies
\begin{align*}
    \frac{\delta\frankOseenEnergy}{\dxvar} = \frac{\delta(\PotEP_{\text{I}}+\PotEP_{\text{E}})}{\dxvar} &= K\Big( \DivSurf\left( \mathcal{L}_{\dirf}(\shop\dirf) + (\DivSurf\dirf)\shop\dirf + (\levicivita{\shop}{\dirf})\dirf\right)\notag \\
    &\quad\quad\quad+ \left\langle (\GradSurf\dirf)^{T}\GradSurf\dirf, \proj_{\qspace}\shop \right\rangle + \left\langle \dirf\otimes\dirf, \shop^2\proj_{\qspace}\shop \right\rangle\Big) \\
    &= K\left( \DivSurf\left( \left( \levicivita{\dirf}{\meanc} + 2\levicivita{\shop}{\dirf} \right)\dirf \right) + \left\langle \sigma^{\text{E}}_{\surf}, \proj_{\qspace}\shop \right\rangle\right)
\end{align*}
with $\sigma^{\text{E}}_{\surf} = (\GradSurf\dirf)^{T}\GradSurf\dirf + \shop\dirf\otimes\shop\dirf\in\tangentT{^2}\surf$ the extrinsic surface Ericksen stress tensor, see \cite{Nitschkeetal_PRF_2019}.
    
\subsubsection{Normalization Penalty Energy}
The normalization penalty energy $\normalizationEnergy$
forces $\dirf$ to be approximately normalized, with $\lim_{\pnorm\rightarrow\infty}\left\| \dirf \right\| = 1$ almost everywhere.
With eq. \eqref{eq:transport} we get
\begin{align*}
    \dsurf\normalizationEnergy  &= \pnorm\int_{\surf} \left( \left\| \dirf \right\|^2 - 1 \right)\left\langle \dirf, \proj_{\surf}\dsurf\dirf   \right\rangle - \frac{\meanc}{4}\left( \left\| \dirf \right\|^2 - 1 \right)^2\dxvar \mu = -\frac{\pnorm}{4} \int_{\surf} \meanc\left( \left\| \dirf \right\|^2 - 1 \right)^2\dxvar \mu
\end{align*}
and therefore
\begin{align*}
    \frac{\delta\normalizationEnergy}{\dxvar} = -\frac{\pnorm}{4}\meanc\left( \left\| \dirf \right\|^2 - 1 \right)^2, \qquad
    \frac{\delta\normalizationEnergy}{\delta\dirf} = \pnorm\left( \left\| \dirf \right\|^2 - 1 \right)\dirf \formPeriod
\end{align*}
   
\subsubsection{Elastic Part of the \landauDeGennes\ Energy}
Consider the elastic part of the \landauDeGennes\ energy, \ie\ eq. \eqref{eq_elenrgy}, in the equivalent representation
$\elasticEnergy = \frac{L}{2} \int_{\surf} \left\| \GradSurf\qten \right\|^2 + \left\langle \shop^2, \rten \right\rangle \mu$, 
where $\rten = (\|\qten\|^2 + \nicefrac{\orderp^2}{2})\Id_{\surf} + 2\orderp\qten$.
Similarly to the \frankOseen\ energy above, we call the first summand the intrinsic part and the second one the extrinsic part. 
Considering the intrinsic part, by eq. \eqref{eq_distgradqten} and symmetry arguments, variation of the integrand yields
\begin{align*}
    \frac{1}{2}\dsurf\left\| \GradSurf\qten \right\|^2 &= \left\langle \GradSurf\qten,\proj_{\surf}\dsurf\GradSurf\qten \right\rangle 
            = \qtenC_{ij|k} \left(\shopC^{k}_{l}\left( \qtenC^{ij|l}\dxvar + 2\qtenC^{l(i}\dxvar^{|j)} \right) - 2\shopC^{k(i}\qtenC^{j)l}\dxvar_{|l}\right)\\
            &= \qtenC_{ij|k} \left(\shopC^{k}_{l}\left( \qtenC^{ij|l}\dxvar + 2\qtenC^{li}\dxvar^{|j} \right) - 2\shopC^{ki}\qtenC^{jl}\dxvar_{|l}\right) \\
    &= \left\langle (\GradSurf\qten)^{T_{(123)}}:\GradSurf\qten,\shop \right\rangle\dxvar + 2\left\langle \levicivita{\qten\shop}{\qten} - \qten\levicivita{\shop}{\qten} , \GradSurf\dxvar \right\rangle \formPeriod
\end{align*}
Hence, with eq. \eqref{eq:transport} we obtain
\begin{align*}
    \frac{1}{2}\dsurf\int_{\surf}\left\| \GradSurf\qten \right\|^2\mu 
        &= \frac{1}{2}\int_{\surf} \dsurf\left\| \GradSurf\qten \right\|^2 - \left\langle (\GradSurf\qten)^{T_{(123)}}:\GradSurf\qten, \Id_{\surf}\right\rangle\meanc\dxvar\mu\\
    &= \int_{\surf}\left\langle (\GradSurf\qten)^{T_{(123)}}:\GradSurf\qten,\proj_{\qspace}\shop \right\rangle\dxvar + 2\left\langle \levicivita{\qten\shop}{\qten} - \qten\levicivita{\shop}{\qten} , \GradSurf\dxvar \right\rangle \\
    &=\int_{\surf} \left( \left\langle (\GradSurf\qten)^{T_{(123)}}:\GradSurf\qten,\proj_{\qspace}\shop \right\rangle + 2\DivSurf\left[ \qten\levicivita{\shop}{\qten} - \levicivita{\qten\shop}{\qten} \right]\right)\dxvar\mu\formPeriod
\end{align*}
For the extrinsic part, we observe that the 2-tensor $\rten\in\tangentT{^2}\surf$ is also a surface independent tensor quantity, \ie\ $\proj_{\surf}\dsurf\rten = 0$,
as $\proj_{\surf}\dsurf\qten = 0$ and $\proj_{\surf}\dsurf\Id_{\surf} = \proj_{\surf}\dsurf( g^{ij}\partial_i\para\otimes\partial_j\para) = 0$. Therefore, by using the product rule, the symmetric behaviors and eq. \eqref{eq_dsurfshop}, we obtain $\frac{1}{2}\dsurf\langle \shop^2,\rten \rangle = \langle \shop\dsurf\shop, \rten \rangle = \langle \rten\shop, \dsurf\shop \rangle
     = \langle \rten\shop,  \GradSurf^2\dxvar + \shop^2\dxvar \rangle = \langle \rten\shop,  \GradSurf^2\dxvar \rangle + ( \meanc^2-\gaussc )\langle \rten , \shop \rangle \dxvar - \meanc\gaussc\Tr\rten\dxvar$, 
where we have used $\shop^3 = (\meanc^2-\gaussc)\shop - \meanc\gaussc\Id_{\surf}$.
Furthermore, with eq. \eqref{eq:transport} and $\shop^2 = \meanc\shop - \gaussc\Id_{\surf}$, this results in
\begin{align}
    \frac{1}{2}\dsurf\int_{\surf} \left\langle \shop^2,\rten \right\rangle \mu &= \frac{1}{2}\int_{\surf} \dsurf\left\langle \shop^2,\rten \right\rangle - \meanc^2\left\langle \rten, \shop \right\rangle\dxvar + \meanc\gaussc\Tr\rten\dxvar\mu \notag\\
    &= \int_{\surf} \left( \DivSurf\Div_{\surf,2}\left( \rten\shop \right) + \left( \frac{\meanc^2}{2} - \gaussc \right)\left\langle \shop, \rten \right\rangle - \frac{\meanc\gaussc}{2}\Tr\rten\right)\dxvar\mu \notag\\
    &=  \int_{\surf} \left( \DivSurf\left( \levicivita{\shop}{\rten} + \rten\GradSurf\meanc \right) + \left( \frac{\meanc^2}{2} - \gaussc \right)\left\langle \shop, \rten \right\rangle - \frac{\meanc\gaussc}{2}\Tr\rten\right)\dxvar\mu \formPeriod \label{eq_tmp000}
\end{align}
Next, we reinsert $\rten = \left( \left\| \qten \right\|^2 + \frac{\orderp^2}{2} \right)\Id_{\surf} + 2\orderp\qten$ term by term, \ie
\begin{align*}
    \left[ \levicivita{\shop}{\rten} \right]_i &= \shopC^{kl}\left[ \left(\qtenC^{mn}\qtenC_{mn} + \frac{\orderp^2}{2}\right)g_{ik} + 2\orderp\qtenC_{ik}\right]_{|l} = 2 \shopC^{kl}\left( \qtenC^{mn}\qtenC_{mn|l}g_{ik} + \orderp\qtenC_{ik|l} \right)\notag\\
    &= 2 \left[ \qten:(\GradSurf\qten)\shop + \orderp\levicivita{\shop}{\qten} \right]_{i} \\
    \rten\GradSurf\meanc &= \left(  \Tr\qten^2 + \frac{\orderp^2}{2} \right)\GradSurf\meanc + 2\orderp\qten\GradSurf\meanc \\
    \left\langle \shop, \rten \right\rangle &= \left(  \Tr\qten^2 + \frac{\orderp^2}{2} \right)\meanc + 2\orderp\left\langle \shop, \qten \right\rangle\\
    \Tr\rten &= 2 \Tr\qten^2 + \orderp^2 \formPeriod
\end{align*}
Therefore, eq. \eqref{eq_tmp000} results in
\begin{align*}
  \frac{1}{2}\dsurf\int_{\surf} \left\langle \shop^2,\rten \right\rangle \mu
         &= \int_{\surf} \bigg( 2\DivSurf\left( \qten:(\GradSurf\qten)\shop + \orderp\levicivita{\shop}{\qten} + \frac{1}{2}\Tr\qten^2\GradSurf\meanc + \orderp\qten\GradSurf\meanc \right)\\ 
         &\quad\quad  +\frac{\orderp^2}{2}\laplaceBeltrami\meanc + \orderp\left\| \shop \right\|^2\left\langle \shop, \qten \right\rangle + \meanc\left( \frac{\meanc^2}{2} - 2\gaussc \right)\left( \Tr\qten^2 + \frac{\orderp^2}{2} \right)\bigg)\dxvar\mu\formPeriod
\end{align*}
Finally, we sum up the results for the intrinsic and extrinsic part and get
\begin{align*}
    \frac{\delta\elasticEnergy}{\dxvar} &= L\bigg[2\DivSurf\left( \qten\levicivita{\shop}{\qten} - \levicivita{\qten\shop}{\qten} + \qten:(\GradSurf\qten)\shop + \orderp\levicivita{\shop}{\qten} + \frac{1}{2}\Tr\qten^2\GradSurf\meanc + \orderp\qten\GradSurf\meanc \right)\\
    &\quad\quad + \orderp\left\| \shop \right\|^2\left\langle \shop, \qten \right\rangle + \meanc\left( \frac{\meanc^2}{2} - 2\gaussc \right)\Tr\qten^2 +  \left\langle (\GradSurf\qten)^{T_{(123)}}:\GradSurf\qten,\proj_{\qspace}\shop \right\rangle\\
    &\quad\quad + \frac{\orderp^2}{2}\left( \laplaceBeltrami\meanc + \meanc\left( \frac{\meanc^2}{2} - 2\gaussc \right) \right)\bigg]\\
    \frac{\delta\elasticEnergy}{\delta\qten} &= L\left[ -\laplaceBochner\qten + \left\| \shop \right\|^2\qten + \orderp\meanc\proj_{\qspace}\shop \right]\formPeriod
\end{align*}

\subsubsection{Bulk Energy Part of the \landauDeGennes\ Energy}
For the bulk energy part of the \landauDeGennes\ energy $\bulkEnergy$, by using $\dsurf\Tr\qten^2 = 2\left\langle \qten , \dsurf\qten \right\rangle = 0$ and 
$\dsurf\Tr\qten^4 = \Tr\qten^2\dsurf\Tr\qten^2 = 0$ it can be easily seen that the energy density is independent \wrt\ perturbation. 
Hence, eq. \eqref{eq:transport} yields
\begin{align*}
    \dsurf\bulkEnergy &= -\int_{\surf} \meanc\left(a'\Tr\qten^2 + c\Tr\qten^4 + C_1\right)\dxvar\mu
\end{align*} 
and therefore
\begin{align*}
    \frac{\delta\bulkEnergy}{\dxvar} = -\meanc\left(a'\Tr\qten^2 + c\Tr\qten^4 + C_1\right), \qquad
    \frac{\delta\bulkEnergy}{\delta\qten} = 2\left( a' + c\Tr\qten^2 \right)\qten \formPeriod
\end{align*}
      
\subsubsection{Surface Area Penalization Energy}
Considering eq. \eqref{eq:transport} and the surface area penalization energy $\areaEnergy$ with $\area:=\int_{\surf}\mu$. It follows that
\begin{align*}
    \dsurf\areaEnergy &= \frac{\areaPen}{\areaZero^2}\left(\area-\areaZero\right)\dsurf\area = -\frac{\areaPen}{\areaZero^2}\left(\area-\areaZero\right)\int_{\surf}\meanc\dxvar\mu \formPeriod
\end{align*}
Thus, the functional derivative of the penalization energy is given by
\begin{align*}
    \frac{\delta\areaEnergy}{\dxvar} &= -\frac{\areaPen}{\areaZero^2}\left(\area-\areaZero\right)\meanc\formPeriod
\end{align*}

\section{Operators}
\label{sec:opkunde}

\begin{proposition}\label{prop_divterm}
    It holds for $ \dirf\in\tangent^{1}\surf $
    \begin{align*}
        \DivSurf\left( \mathcal{L}_{\dirf}(\shop\dirf) + (\DivSurf\dirf)\shop\dirf + (\levicivita{\shop}{\dirf})\dirf\right) &= \DivSurf\left( \left( \Tr\Lie_{\dirf}\shop \right)\dirf \right) \formComma
    \end{align*}
    where $\Lie_{\dirf}\shop = \levicivita{\dirf}{\shop} + \left( \GradSurf\dirf \right)^{T}\shop + \shop\GradSurf\dirf$
    is the Lie derivative of the fully covariant shape operator $\shop\in\tangentT{_2}\surf$ in direction of 
    $\dirf\in\tangent^{1}\surf$, \ie\ $\Tr\Lie_{\dirf}\shop = \levicivita{\dirf}{\meanc} + 2\levicivita{\shop}{\dirf}$.
\end{proposition}
\begin{proof}
    With Levi-Civita tensor $\LC=\left\{ \sqrt{\det\g}\ \epsilon_{ij} \right\} \cong \mu$, Hodge dual vector field $*\dirf=-\LC\dirf$ and the function curl $\BigRotSurf f = -\LC\GradSurf f$, subtracting the right-hand side from the left-hand side gives
    \begin{align*}
        \DivSurf\left( \mathcal{L}_{\dirf}(\shop\dirf) + (\DivSurf\dirf)\shop\dirf  - \left( \levicivita{\dirf}{\meanc} \right)\dirf - (\levicivita{\shop}{\dirf})\dirf \right) &= \DivSurf\BigRotSurf\left\langle \shop\dirf, *\dirf \right\rangle =0 \formComma
    \end{align*}
    since $ \DivSurf\circ\BigRotSurf=0 $ on functions and
    \begin{align*}
        \left[\BigRotSurf\left\langle \shop\dirf, *\dirf \right\rangle\right]^i &= \LCC^{ij}\LCC_{km}\left( \shopC^k_l \dirfC^m \dirfC^l \right)_{|j}\\
        &= \left( \delta^i_k\delta^j_m - \delta^i_m\delta^j_k \right) \left( \shopC^k_l \dirfC^m \dirfC^l \right)_{|j} = \left( \shopC^i_l\dirfC^j\dirfC^l - \shopC^j_l\dirfC^i\dirfC^l \right)_{|j}\\
        &= \left[ \levicivita{\dirf}{\left( \shop\dirf \right)} + (\DivSurf\dirf)\shop\dirf - \left( \levicivita{\dirf}{\meanc} \right)\dirf - \levicivita{\shop\dirf}{\dirf} - (\levicivita{\shop}{\dirf})\dirf\right]^i\formComma
    \end{align*}
    where the Hodge compatibility $\GradSurf\LC=0$ was used and $\levicivita{\dirf}{\left( \shop\dirf \right)} - \levicivita{\shop\dirf}{\dirf} = \mathcal{L}_{\dirf}(\shop\dirf)$ is the Lie derivative of the contravariant vector field $\shop\dirf\in\tangent^{1}\surf$ in direction of $\dirf$.
\end{proof}
Additionally, the contravariant components of some of the used operators are considered, \ie 
$[ \laplaceBochner\qten]^{ij} = \tensor{\qtenC}{^{ij|k}_{|k}}$, $[ \qten\levicivita{\shop}{\qten}]^i = \qtenC^{ij}\shopC^{kl}\qtenC_{jk|l}$, 
$[ \levicivita{\qten\shop}{\qten}]^i = \qtenC_{jk}\shopC^{kl}\tensor{\qtenC}{^{ij}_{|l}}$, $[ \qten:(\GradSurf\qten)\shop]^i = \qtenC^{jk}\shopC^{li}\qtenC_{jk|l}$, $[ \levicivita{\shop}{\qten}]^i = \shopC_{jk}\qtenC^{ij|k}$ and $[ (\GradSurf\qten)^{T_{(123)}}:\GradSurf\qten]^{ij} = \qtenC^{kl|i}\tensor{\qtenC}{_{kl}^{|j}}$.

\section{Kinetic Considerations}
\label{sec:kinetic}

We consider kinetic properties, \ie\ time depending behaviours of some quantities such as material time derivative, rate of normal vector or energy dissipation.

\subsection{Time Derivative}
        
Since particle motion is considered only in normal direction, the Lagrangian and the transversal or intrinsic Eulerian observer coincide due to vanishing relative velocity.
For simplicity, we here only consider these equivalent observers and define the material time derivative covariantly by the tangential part of the time derivative in the embedding space and denote it by a dot. 
Thus, for scalar fields $f\in\tangentT{^0}\surf$, vector fields $\dirf\in\tangentT{^1}\surf$ and 2-tensor fields $\rten\in\tangentT{^2}\surf$ we state that
\begin{align}
    \dot{f} &:= \partial_t f\\
    \dot{\dirf} &:= \proj_{\surf}\left[\partial_t\dirf\right]
                  = g^{ik}\left< \partial_t\dirf,\partial_k\para\right>\partial_i\para\\
    \dot{\rten} &:= \proj_{\surf}\left[\partial_t\rten\right]
                  = g^{ik}g^{jl}\left< \partial_t\rten,\partial_k\para\otimes\partial_l\para\right>\partial_i\para\otimes\partial_j\para 
\end{align}
With chain rule, $\normal\bot\partial_k\para$ and \eqref{eq:paraDot} the covariant components of $\dot{\dirf}$ read
\begin{align}
    \left[\dot{\dirf}\right]_k 
        &= \left< \partial_t\dirf,\partial_k\para\right>
         = \left< \left(\partial_t\dirfC^i\right)\partial_i\para + \dirfC^i\partial_i\dot{\para}, \partial_k\para\right>
         = g_{ik}\partial_t\dirfC^i - \vnor\dirfC^i\shopC_{ik} \formComma
\end{align}
which eventually gives the time derivative of the vector filed $\dirf$ in terms of contravariant proxy rates, \ie
\begin{align}\label{eq:vecdot}
    \dot{\dirf} &= \left(\partial_t\dirfC^i\right)\partial_i\para - \vnor\shop\dirf \formPeriod
\end{align}
Analogously we obtain for the 2-tensor field $\rten$ that
\begin{align}\label{eq:twotendot}
    \dot{\rten} &= \left(\partial_t\rtenC^{ij}\right)\partial_i\para\otimes\partial_j\para - \vnor\left(\shop\rten + \rten\shop\right) \formPeriod
\end{align}
Similar to the surface variation $\dsurf$, the time derivative for 2-tensors is closed \wrt\ symmetry and trace-free behaviour.
This means, it holds for all $\qten\in\qspace$ that $\dot{\qten}\in\qspace$, since $\partial_t g_{ij} = -2\vnor\shopC_{ij}$ and thus $g_{ij}\partial_t\qtenC^{ij}= 2v\left<\shop,\qten\right> $.
    
\subsection{Normal Vector Rate}
Eq. \eqref{eq:paraDot}, $\normal\bot\partial_j\para$ and $\|\normal\|=1$ yield
\begin{align}
    \partial_t\normal
        &= \left<\partial_t \normal, \normal\right>\normal + g^{ij}\left<\partial_t \normal, \partial_j \para\right>\partial_i \para
         = -g^{ij}\left<\normal, \partial_j \dot{\para}\right>\partial_i \para
         = -g^{ij}\left<\normal, \left(\partial_j \vnor\right)\normal\right>\partial_i \para\notag\\
        &= -\GradSurf\vnor \formPeriod \label{eq:normalevolution}
\end{align}
A similar evolving geometric quantity has been considered in \cite{Huisken_JDG_1984} in the context of mean curvature flow.
    
\subsection{Dissipation of Energy}\label{sec:dissipationenergy}
With both $L^2$-gradient flows \eqref{eq:gradientflow:p:vnor}--\eqref{eq:gradientflow:q:qten}, their assumptions of independence and the theorem below the energy dissipation rates are given by 
\begin{align}
    \ddt\PotEP &= -\left(\kinConst\left\|\vnor\right\|^2_{L^2(\surf,\tangentT{^0}\surf)} + \kinConstP\left\|\dot{\dirf}\right\|^2_{L^2(\surf,\tangentT{^1}\surf)}\right)
               \le 0 \label{eq:dissennergyP}\\
    \ddt\PotEQ &= -\left(\kinConst\left\|\vnor\right\|^2_{L^2(\surf,\tangentT{^0}\surf)} + \kinConstQ\left\|\dot{\qten}\right\|^2_{L^2(\surf,\qspace)}\right)
               \le 0 \formComma\label{eq:dissennergyQ}
\end{align}
\ie\ the energy decrease in a $L^2$-manner, similar to classical geometric evolution equations of $L^2$-gradient flows structure and $L^2$-gradient flows on stationary surfaces, as considered in \cite{Nestleretal_JNS_2018,Nitschkeetal_PRSA_2018}.
\begin{theorem}
    Assume that a tensor field $\rten\in\tangentT{^n}\surf$ is independent of a surface $\surf=\surf[\xvar]$ in the sense of $\proj_{\surf}\dsurf\rten =0 $ or $\frac{\textup{d}\rten}{\textup{d}\xvar}= 0$, equivalently.
    Then it holds for an energy functional $\PotE=\PotE[\xvar,\rten]$ the chain rule
    \begin{align}\label{eq:generalenergydissipation}
        \ddt\PotE &= \nabla^*_{\!(\dot{\xvar},\dot{\rten})}\PotE
               := \left(\frac{\delta\PotE}{\dxvar}\right)^*\left[\dot{\xvar}\right]
                + \left(\frac{\delta\PotE}{\delta\rten}\right)^*\left[\dot{\rten}\right] \formComma
    \end{align}
    where the star $*$ denotes the dual element \wrt\ the $L^2$-inner product.
\end{theorem}
\begin{proof}
    By linearity of the integral, it is sufficient to show this statement on a subset    $\Sigma:=\para[\xvar](Y)\subseteq\surf$ with $Y=(y^1,y^2)\subset\R^2$.
    To use common differential calculus locally, we formulate the energy functional as
    \begin{align*}
        \PotE_{\Sigma}\left[\xvar,\rten\right]
            &= \int_{\Sigma} u\left(\xvar,\rten\right)\mu
             = \int_{Y}f\left(\xvar,\rten\right)\dup Y\formComma
    \end{align*}
    where $\dup Y= \textup{d}y^1\!\wedge\!\textup{d}y^2$, \ie\ the function $f$ equals energy density $u$ times area density function $\sqrt{\det\g}=\sqrt{\det\g(\xvar)}$.
    By definitions of functional and function derivative and the independence between $\xvar$ and $\rten$ the right-hand side of \eqref{eq:generalenergydissipation} reads
    \begin{align*}
        \nabla^*_{\!(\dot{\xvar},\dot{\rten})}\PotE_{\Sigma}
            &= \lim_{\epsilon_{\xvar}\rightarrow 0} \frac{\PotE_{\Sigma}\left[\xvar+ \epsilon_{\xvar}\dot{\xvar},\rten\right] - \PotE_{\Sigma}\left[\xvar,\rten\right]}{\epsilon_{\xvar}}
              +\lim_{\epsilon_{\rten}\rightarrow 0} \frac{\PotE_{\Sigma}\left[\xvar,\rten+ \epsilon_{\rten}\dot{\rten}\right] - \PotE_{\Sigma}\left[\xvar,\rten\right]}{\epsilon_{\xvar}} \\
            &=  \int_{Y} \lim_{\epsilon_{\xvar}\rightarrow 0} \frac{f\left(\xvar+ \epsilon_{\xvar}\dot{\xvar},\rten\right) - f\left(\xvar,\rten\right)}{\epsilon_{\xvar}}
              +\lim_{\epsilon_{\rten}\rightarrow 0} \frac{f\left(\xvar,\rten+ \epsilon_{\rten}\dot{\rten}\right) - f\left(\xvar,\rten\right)}{\epsilon_{\xvar}} \dup Y \\
            &= \int_{Y} \left(\frac{\partial f}{\partial \xvar}\right)\dot{\xvar} + \left(\frac{\partial f}{\partial \rten}\right)^* \left(\dot{\rten}\right) \dup Y
             = \int_{Y} \left(\frac{\dup f}{\dup \xvar}\right)\frac{\dup \xvar}{\dup t} + \left<\frac{\dup f}{\dup \rten}, \frac{\dup \rten}{\dup t}\right>_{\R^3} \dup Y\\
            &= \int_{Y} \ddt f\dup Y
             = \ddt\PotE_{\Sigma}\formComma
    \end{align*}
    since $\frac{\dup f}{\dup \rten}\in\tangentT{^n}\surf$.
    
    Actually, the circuit over the local $\R^3$-inner product is not necessary, if we reformulate $f$ such that it depends on the contravariant proxy functions of the tensor field $\rten$, partially. 
    To prevent confusing about interpretations of the differential expressions, let us consider the situation for $n=1$, where we set $\rten=\dirf\in\tangentT{^1}\surf$.
    Moreover, we define $\tilde{f}(\xvar,\dirfC^1,\dirfC^2):= f(\xvar,\dirf)$ with proxy functions $\dirfC^i=\dirfC^i(\xvar)=\dirfC^i(t,\xvar(t))$.
    Therefore, it holds
    \begin{align*}
        \frac{\dup \tilde{f}}{\dup \xvar}
            &= \frac{\partial \tilde{f}}{\partial \xvar} + \frac{\partial \tilde{f}}{\partial\dirfC^i}\frac{\partial\dirfC^i}{\partial\xvar}
        &&\text{and} 
        &\left(\frac{\dup \tilde{f}}{\dup \dirf}\right)^*
            &= \frac{\partial \tilde{f}}{\partial \dirfC^i}\dup y^i \formPeriod
    \end{align*}
    By definition of the variation $\dsurf$, we notice that 
    $\dsurf\dirfC^i = \frac{\partial\dirfC^i}{\partial\xvar}\dxvar$, \ie\ with arbitrariness of $\dxvar$, assumption $\proj_{\surf}\dsurf\dirf =0 $ and its representation
    \eqref{eq_codsurfdirf} the dependency of $\xi$ for the proxy functions is given by $\frac{\partial\dirfC^i}{\partial\xvar} = [\shop\dirf]^i$.
    Hence, the components of the time derivative \eqref{eq:vecdot} can be written as
    \begin{align*}
       \left[\dot{\dirf}\right]^i
            &= \frac{\dup \dirfC^i}{\dup t} - \frac{\partial\dirfC^i}{\partial\xvar}\frac{\dup \xvar}{\dup t} \formPeriod
    \end{align*}
    To sum up, we have
    \begin{align*}
        \left(\frac{\dup \tilde{f}}{\dup \xvar}\right)\dot{\xvar} + \left(\frac{\dup \tilde{f}}{\dup \dirf}\right)^* \left(\dot{\dirf}\right)
            &= \frac{\partial \tilde{f}}{\partial \xvar}\frac{\dup \xvar}{\dup t} + \frac{\partial \tilde{f}}{\partial \dirfC^i}\frac{\dup \dirfC^i}{\dup t}
             =\ddt\tilde{f}
             =\ddt f \formPeriod
    \end{align*}
    Note that for $n>1$ the argumentation would be exactly the same, but with $2^n$ contravariant proxy functions.
    For $n=0$ the situation is much simpler, since $\tilde{f}=f$. Note that in this case the projection $\proj_{\surf}$ stated in the assumption is just the scalar identity.
\end{proof}

\section{Numerical Tests}
\label{sec:numericaltests}

\subsection{Numerical Computation of Geometric Quantities}

To numerically solve eqs. \eqref{eq:gradientflow:p:vnor} and \eqref{eq:gradientflow:q:vnor} requires a \textit{sound approximation} of all geometric quantities. 
They all can be computed from the normal vector, which is obtained through $\partial_t\normal = -\GradSurf\vnor$, see eq. \eqref{eq:normalevolution}. 
This approach turns out to be much more accurate than the often used (weighted) element normal vector. 
We show the numerical convergence (by means of the timestep width $\tau$) of this equation for the  \helfrich\ and the surface area penalization energy for an ellipsoidal shape with axes parameters $(1,1,1.25)$ as initial geometry. 
The ellipsoid converges to a sphere, for which the normal vector is analytically known. 
Thus, we compare the computed normal vector according to eq. \eqref{eq:normalevolution} and the analytical one in the steady state regime according to the $H^1$ semi-norm. 
\autoref{fig:convergence_study} shows the results, where a linear convergence rate \wrt\ the timestep width $\tau$ can be observed. 
It is noted that we use a much smaller timestep width -- according to \autoref{tab:parameters} -- for the simulation as for the convergence evaluations in \autoref{fig:convergence_study} in order to reduce the temporal error made in the normal velocity equation. 

\begin{figure}[!h]
    \centering
    \ifthenelse{\boolean{useSnapshots}}
    {
		\ifthenelse{\boolean{forArxiv}}
		{
			\includegraphics[width=0.675\textwidth]{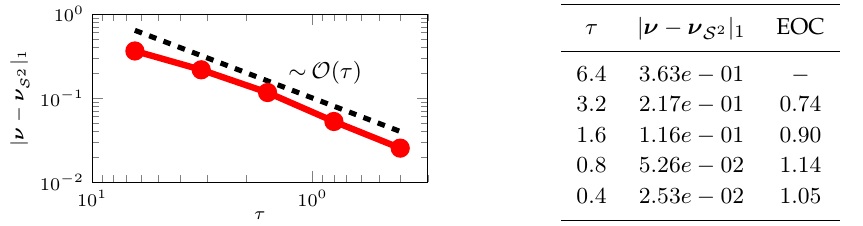}
		}
		{
			\includegraphics[width=0.9\textwidth]{snapshots/fig_10.png}
		}
    }
    {
		\ifthenelse{\boolean{forArxiv}}
		{
			\begin{minipage}{0.75\textwidth}
		}
		{
			\begin{minipage}{\textwidth}
		}
				\centering
				\begin{minipage}{0.49\textwidth}
					\centering
%
%
\begin{tikzpicture}

\begin{axis}[%
width=0.7\textwidth,
height=0.35\textwidth,
at={(0,0)},
scale only axis,
separate axis lines,
every outer x axis line/.append style={black},
every x tick label/.append style={font=\color{black}},
x dir=reverse,
xmode=log,
xmin=0.3,
xmax=10,
xminorticks=true,
xlabel={$\tau$},
every outer y axis line/.append style={black},
every y tick label/.append style={font=\color{black}},
ymode=log,
ymin=0.01,
ymax=1.0,
yminorticks=true,
ylabel={$|\normal-\normal_{\mathcal{S}^2}|_1$},
axis background/.style={fill=white},
x label style={font=\scriptsize,at={(axis description cs:0.5,0.1)}},
y label style={font=\scriptsize,at={(axis description cs:0.05,0.5)}},
yticklabel style = {font=\scriptsize},
xticklabel style = {font=\scriptsize},
]
\addplot [color=red,solid,line width=2.5pt,mark size=2.5pt,mark=*,mark options={solid,fill=red},forget plot]
  table[row sep=crcr]{%
6.4	0.363468\\
3.2	0.216883\\
1.6	0.116306\\
0.8	0.052608\\
0.4	0.0253203\\
};
\addplot [color=black,dashed,line width=2.0pt,forget plot]
  table[row sep=crcr]{%
6.4	0.64\\
3.2	0.32\\
1.6	0.16\\
0.8	0.08\\
0.4	0.04\\
};
\node[right, align=left, text=black]
at (axis cs:1.4,0.2) {$\sim\mathcal{O}(\tau)$};
\end{axis}
\end{tikzpicture}%
				\end{minipage}
				\begin{minipage}{0.49\textwidth}
					\centering
					\begin{tabular}{ccc}
						\hline\noalign{\smallskip}
						$\tau$ & $|\normal-\normal_{\mathcal{S}^2}|_1$ & $\textup{EOC}$ \\
						\noalign{\smallskip}\hline\noalign{\smallskip}
						$6.4$ & $3.63e-01$ & $-$ \\
						$3.2$ & $2.17e-01$ & $0.74$ \\
						$1.6$ & $1.16e-01$ & $0.90$ \\
						$0.8$ & $5.26e-02$ & $1.14$ \\
						$0.4$ & $2.53e-02$ & $1.05$ \\
						\noalign{\smallskip}\hline
					\end{tabular}
				\end{minipage}
			\end{minipage}
	}
    \caption{Left: Normal vector error against timestep width $\tau$. Thereby, $|\cdot|_1$ denotes the spatial $H^1$ semi-norm and $\normal_{\mathcal{S}^2}$ is the normal vector of the unit sphere. Right: Normal vector error together with experimental order of convergence (EOC) values for various timestep widths $\tau$.}
    \label{fig:convergence_study}
\end{figure}

\subsection{Surface Area Conservation}

In the \frankOseenHelfrich\ model and the \landauDeGennesHelfrich\ model, the surface area penalization energy is included in order to achieve constant surface area over time. 
This approximation depends on the penalty parameter $\areaPen$. 
In the following the convergence rate \wrt\ $\areaPen$ is numerically investigated.
Here, we use qualitatively the same setup as in \autoref{fig:surface_response:results}, but with slightly different parameters, see \autoref{tab:parameters}.
The results are shown in \autoref{fig:surfaceAreaTest}, where a linear convergence rate \wrt\ $\areaPen$ can be observed.

\begin{figure}[!h]
    \centering
    \ifthenelse{\boolean{useSnapshots}}
    {
		\ifthenelse{\boolean{forArxiv}}
		{
			\includegraphics[width=0.7125\textwidth]{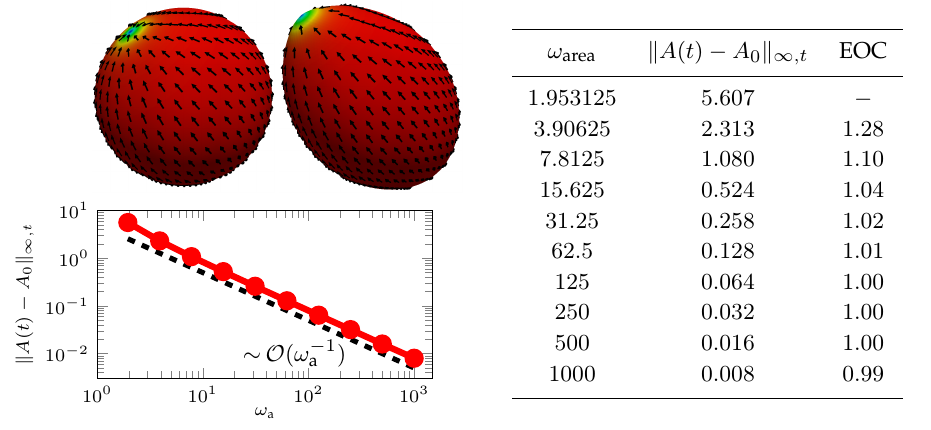}
		}
		{
			\includegraphics[width=0.95\textwidth]{snapshots/fig_11.png}
		}
    }
    {
		\ifthenelse{\boolean{forArxiv}}
		{
			\begin{minipage}{0.75\textwidth}
		}
		{
			\begin{minipage}{\textwidth}
		}
				\centering
				\begin{minipage}{0.49\textwidth}
					\raggedleft
					\includegraphics[width=0.38\textwidth]{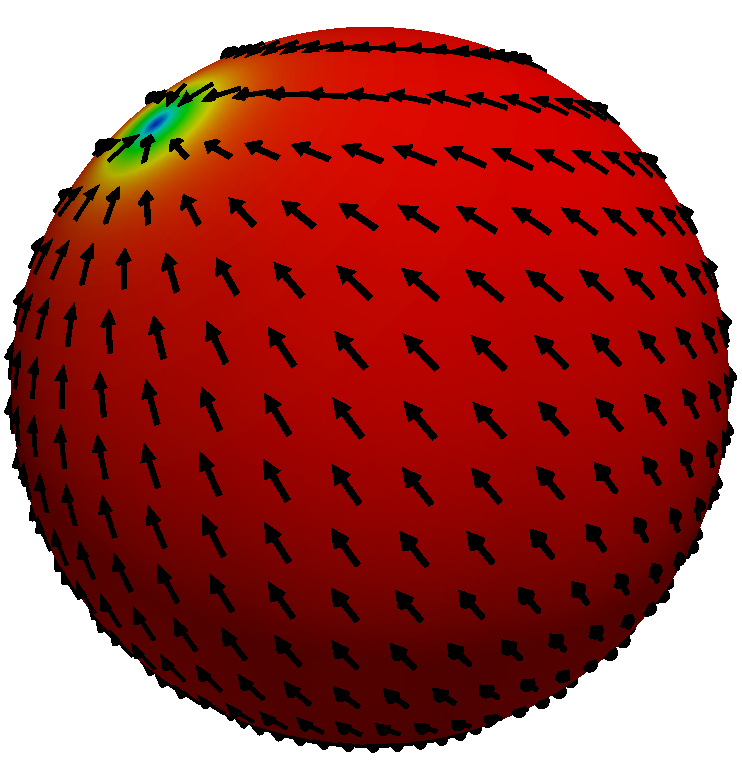}
					\includegraphics[width=0.38\textwidth]{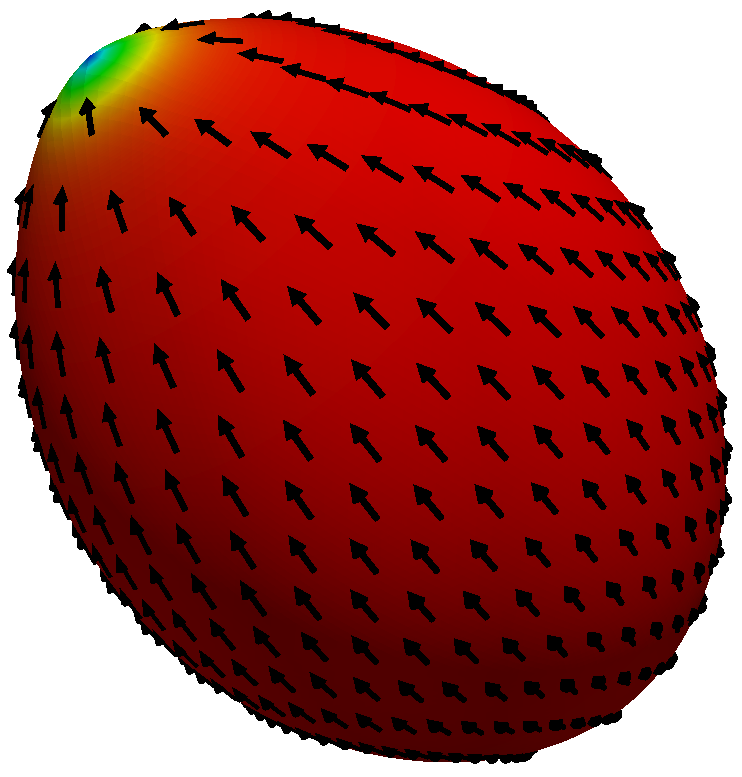}\\
					\centering
%
%
\begin{tikzpicture}

\begin{axis}[%
width=0.7\textwidth,
height=0.35\textwidth,
at={(0,0)},
scale only axis,
separate axis lines,
every outer x axis line/.append style={black},
every x tick label/.append style={font=\color{black}},
xmode=log,
xmin=1,
xmax=1500,
xminorticks=true,
xlabel={$\areaPen$},
every outer y axis line/.append style={black},
every y tick label/.append style={font=\color{black}},
ymode=log,
ymin=0.003,
ymax=10,
yminorticks=true,
ylabel={$\|\area(t)-\areaZero\|_{\infty,t}$},
axis background/.style={fill=white},
x label style={font=\scriptsize,at={(axis description cs:0.5,0.1)}},
y label style={font=\scriptsize,at={(axis description cs:0.05,0.5)}},
yticklabel style = {font=\scriptsize},
xticklabel style = {font=\scriptsize},
]
\addplot [color=red,solid,line width=2.5pt,mark size=2.5pt,mark=*,mark options={solid,fill=red},forget plot]
  table[row sep=crcr]{%
1.953125	5.60698\\
3.90625	2.3128\\
7.8125	1.08\\
15.625	0.5237\\
31.25	0.258100000000001\\
62.5	0.1281\\
125	0.0639000000000003\\
250	0.0319000000000003\\
500	0.0159000000000002\\
1000	0.00799999999999912\\
};
\addplot [color=black,dashed,line width=2.0pt,forget plot]
  table[row sep=crcr]{%
1.953125	2.56\\
3.90625	1.28\\
7.8125	0.64\\
15.625	0.32\\
31.25	0.16\\
62.5	0.08\\
125	0.04\\
250	0.02\\
500	0.01\\
1000	0.005\\
};
\node[right, align=left, text=black]
at (axis cs:20,0.01) {$\sim\mathcal{O}(\omega_{\textup{a}}^{-1})$};
\end{axis}
\end{tikzpicture}%
				\end{minipage}
				\begin{minipage}{0.49\textwidth}
					\centering
					\begin{tabular}{ccc}
						\hline\noalign{\smallskip}
						$\omega_{\textup{area}}$ & $\|A(t)-A_0\|_{\infty,t}$ & $\textup{EOC}$ \\
						\noalign{\smallskip}\hline\noalign{\smallskip}
						$1.953125$ & $5.607$ & $-$ \\
						$3.90625$  & $2.313$ & $1.28$ \\
						$7.8125$   & $1.080$ & $1.10$ \\
						$15.625$   & $0.524$ & $1.04$ \\
						$31.25$    & $0.258$ & $1.02$ \\
						$62.5$     & $0.128$ & $1.01$ \\
						$125$      & $0.064$ & $1.00$ \\
						$250$      & $0.032$ & $1.00$ \\
						$500$      & $0.016$ & $1.00$ \\
						$1000$     & $0.008$ & $0.99$ \\
						\noalign{\smallskip}\hline
					\end{tabular}
				\end{minipage}
			\end{minipage}
	}
    \caption{Left: Initial condition as stable configuration on the unit sphere (top left), reached steady state solution on the deformed surface (top right) and surface area error against surface area penalization parameter $\areaPen$ (bottom) with the temporal maximum norm $\|\cdot\|_{\infty,t}$. Right: Surface area error together with experimental order of convergence (EOC) values for various surface area penalization parameters $\areaPen$.}
    \label{fig:surfaceAreaTest}
\end{figure}

\bibliographystyle{RS}
\bibliography{bib}
\end{document}